\documentclass[10pt,english]{article}
\usepackage[T1]{fontenc}
\usepackage[latin9]{inputenc}
\usepackage[a4paper]{geometry}
\usepackage{babel}
\usepackage{array}
\usepackage{textcomp}
\usepackage{mathtools}
\usepackage{amsmath}
\usepackage{amsthm}
\usepackage{amssymb}
\usepackage{graphicx}
\usepackage{setspace}
\usepackage[unicode=true,pdfusetitle,
 bookmarks=true,bookmarksnumbered=false,bookmarksopen=false,
 breaklinks=false,pdfborder={0 0 1},backref=false,colorlinks=false]
 {hyperref}

\makeatletter

\providecommand{\tabularnewline}{\\}

\theoremstyle{plain}
\newtheorem{thm}{\protect\theoremname}
\theoremstyle{plain}
\newtheorem{assumption}[thm]{\protect\assumptionname}
\theoremstyle{plain}
\newtheorem{prop}[thm]{\protect\propositionname}
\theoremstyle{plain}
\newtheorem{lem}[thm]{\protect\lemmaname}

\usepackage{enumitem}

\usepackage{colortbl}
\definecolor{lg}{gray}{0.8}

\usepackage{color}

\usepackage[labelfont=bf]{caption}
\usepackage[bottom]{footmisc}

\usepackage[sort&compress]{natbib}
\usepackage[running, mathlines]{lineno}

\numberwithin{thm}{section}
\numberwithin{equation}{section}
\numberwithin{table}{section}
\numberwithin{figure}{section}

\newtheoremstyle{boldremark}
  {}
  {}
  {}
  {}
  {\bfseries}
  {.}
  { }
  {}%

\theoremstyle{boldremark}
\newtheorem{brem}{Remark} 
\numberwithin{brem}{section}

\usepackage[title]{appendix}
\usepackage{etoolbox}
\patchcmd{\appendices}{\quad}{: }{}{}

\@ifundefined{showcaptionsetup}{}{%
 \PassOptionsToPackage{caption=false}{subfig}}
\usepackage{subfig}
\makeatother

\providecommand{\assumptionname}{Assumption}
\providecommand{\lemmaname}{Lemma}
\providecommand{\propositionname}{Proposition}
\providecommand{\theoremname}{Theorem}

\begin{document}
\title{A parsimonious neural network approach to solve portfolio optimization
problems without using dynamic programming}
\author{Pieter M. van Staden\thanks{National Australia Bank, Melbourne, Victoria, Australia 3000. The
research results and opinions expressed in this paper are solely those
of the authors, are not investment recommendations, and do not reflect
the views or policies of the NAB Group. \texttt{pieter.vanstaden@gmail.com}} \and Peter A. Forsyth\thanks{Cheriton School of Computer Science, University of Waterloo, Waterloo
ON, Canada, N2L 3G1, \texttt{paforsyt@uwaterloo.ca}} \and Yuying Li\thanks{Cheriton School of Computer Science, University of Waterloo, Waterloo
ON, Canada, N2L 3G1, \texttt{yuying@uwaterloo.ca}} }
\date{\today}
\maketitle
\begin{abstract}
We present a parsimonious neural network approach, which does not
rely on dynamic programming techniques, to solve dynamic portfolio
optimization problems subject to multiple investment constraints.
The number of parameters of the (potentially deep) neural network
remains independent of the number of portfolio rebalancing events, and in contrast to, for example, reinforcement learning, the approach
avoids the computation of high-dimensional conditional expectations.
As a result, the approach remains practical even when considering
large numbers of underlying assets, long investment time horizons
or very frequent rebalancing events. We prove convergence of the numerical
solution to the theoretical optimal solution of a large class of problems
under fairly general conditions, and present ground truth analyses
for a number of popular formulations, including mean-variance and
mean-conditional value-at-risk problems. We also show that it is feasible
to solve Sortino ratio-inspired objectives (penalizing only the variance
of wealth outcomes below the mean) in dynamic trading settings with
the proposed approach. Using numerical experiments, we demonstrate
that if the investment objective functional is separable in the sense
of dynamic programming, the correct time-consistent optimal investment
strategy is recovered, otherwise we obtain the correct pre-commitment
(time-inconsistent) investment strategy. The proposed approach remains
agnostic as to the underlying data generating assumptions, and results
are illustrated using (i) parametric models for underlying asset returns,
(ii) stationary block bootstrap resampling of empirical returns, and
(iii) generative adversarial network (GAN)-generated synthetic asset
returns. 

\medskip{}

\medskip{}

\end{abstract}

\section{Introduction\label{sec:Introduction}}

We present, and analyze the convergence of, a parsimonious and flexible
neural network approach to obtain the numerical solution of a large
class of dynamic (i.e. multi-period) portfolio optimization problems
that can be expressed in the following form, 

\begin{equation}
\inf_{\xi\in\mathbb{R}}\inf_{\mathcal{P}\in\mathcal{A}}\left\{ \:E_{\mathcal{P}}^{t_{0},w_{0}}\left[F\left(W\left(T\right),\xi\right)+G\left(W\left(T\right),E_{\mathcal{P}}^{t_{0},w_{0}}\left[W\left(T\right)\right],w_{0},\xi\right)\begin{array}{c}
\!\!\!\!\!\!\!\!\\
\!\!\!\!\!\!\!\!
\end{array}\right]\:\right\} .\label{eq: Initial full objective}
\end{equation}
While rigorous definitions and assumptions are discussed in subsequent
sections, here we simply note that in general, $F:\mathbb{R}^{2}\rightarrow\mathbb{R}$
and $G:\mathbb{R}^{4}\rightarrow\mathbb{R}$ denote some continuous
functions and $\xi\in\mathbb{R}$ some auxiliary variable, with $T>0$
denoting the investment time horizon, $W\left(t\right),t\in\left[t_{0},T\right]$,
the controlled wealth process, and $\mathcal{P}$ representing the
investment strategy (or control) implemented over $\left[t_{0},T\right]$.
Typically, $\mathcal{P}$ specifies the amount or fraction of wealth
to invest in each of (a potentially large number of) the underlying
assets at each portfolio rebalancing event, which in practice occurs
at some discrete subset of rebalancing times in $\left[t_{0},T\right]$.
$\mathcal{A}$ denotes the set of admissible investment strategies
encoding the (possibly multiple) investment constraints faced by the
investor. Finally, $E_{\mathcal{P}}^{t_{0},w_{0}}\left[\cdot\right]$
denotes the expectation given control $\mathcal{P}$ and initial wealth
$W\left(t_{0}\right)=w_{0}$. 

Although (\ref{eq: Initial full objective}) is written for objective
functions involving the terminal portfolio wealth $W\left(T\right)$,
the approach and convergence analysis could be generalized without
difficulty to objective functions that are wealth path-dependent,
i.e. functions of $\left\{ W\left(t\right):t\in\mathcal{T}\right\} $
for some subset $\mathcal{T}\subseteq\left[t_{0,}T\right]$ - see
\cite{ForsythPvSLi2022,PvSForsythLi2022_CD} for examples. However,
since a sufficiently rich class of problems are of the form (\ref{eq: Initial full objective}),
this will remain the main focus of this paper. 

The proposed approach does not rely on the separability of the objective
functional in (\ref{eq: Initial full objective}) in the sense of
dynamic programming, remains agnostic as to the underlying data generation
assumptions, and is sufficiently flexible such that practical considerations
such as multiple investment constraints and discrete portfolio rebalancing
can be incorporated without difficulty.

Leaving the more formal treatment for subsequent sections, for introductory
purposes we highlight some specific examples of problems of the form
(\ref{eq: Initial full objective}):
\begin{enumerate}
\item Utility maximization (see for example \cite{Vigna_efficiency2014}),
in which case there is no outer optimization problem and $G\equiv0$,
while $w\rightarrow U\left(w\right)$ denotes the investor's utility
function, so that (\ref{eq: Initial full objective}) therefore reduces
to 
\begin{equation}
\sup_{\mathcal{P}\in\mathcal{A}}\left\{ E_{\mathcal{P}}^{t_{0},w_{0}}\left[U\left(W\left(T\right)\right)\right]\right\} .\label{eq: Intro - Utility maximization}
\end{equation}
\item Mean-variance (MV) optimization (see e.g. \cite{ZhouLi2000,LiNg2000}),
with $\rho>0$ denoting the scalarization (or risk aversion) parameter,
where the problem
\begin{equation}
\sup_{\mathcal{P}\in\mathcal{A}}\left\{ E_{\mathcal{P}}^{t_{0},w_{0}}\left[W\left(T\right)\right]-\rho\cdot Var_{\mathcal{P}}^{t_{0},w_{0}}\left[W\left(T\right)\right]\right\} ,\label{eq: Intro - MV optimization}
\end{equation}
 can also be written in the general form (\ref{eq: Initial full objective}). 
\item Mean-Conditional Value-at-Risk (CVaR) optimization, in which case
we do have both an inner and an outer optimization problems (see e.g.
\cite{MillerYang2017,Forsyth2019CVaR}, resulting in a problem of
the form 
\begin{equation}
\inf_{\xi\in\mathbb{R}}\inf_{\mathcal{P}\in\mathcal{A}}\left\{ E_{\mathcal{P}}^{t_{0},w_{0}}\left[F\left(W\left(T\right),\xi\right)\right]\right\} ,\label{eq: Intro mean CVaR}
\end{equation}
for a particular choice of the function $F$ (see (\ref{eq: MCV objective})
below).
\item To illustrate the flexibility and generality of the proposed approach,
we also consider a ``mean semi-variance'' portfolio optimization
problem that is inspired by the popular Sortino ratio (\cite{BodieEtAl2014})
in the case of one-period portfolio analysis, where only the variance
of downside outcomes relative to the mean is penalized. In the case
of dynamic trading strategies, this suggests an objective function
of the form 
\begin{equation}
\sup_{\mathcal{P}\in\mathcal{A}}\left\{ E_{\mathcal{P}}^{t_{0},w_{0}}\left[W\left(T\right)-\rho\cdot\left(\min\left\{ W\left(T\right)-E_{\mathcal{P}}^{t_{0},w_{0}}\left[W\left(T\right)\right],0\right\} \right)^{2}\right]\right\} ,\label{eq: Intro - Sortino}
\end{equation}
where, as in the case of (\ref{eq: Intro - MV optimization}), the
parameter $\rho>0$ encodes the trade-off between risk and return.
Note that (\ref{eq: Intro - Sortino}) is not separable in the sense
of dynamic programming, and in the absence of embedding results (analogous
to those of \cite{ZhouLi2000,LiNg2000} in the case of MV optimization
(\ref{eq: Intro - MV optimization})), problem (\ref{eq: Intro - Sortino})
cannot be solved using traditional dynamic programming-based methods.
\end{enumerate}
However, we emphasize that (\ref{eq: Intro - Utility maximization})-(\ref{eq: Intro - Sortino})
are only a selection of examples, and the proposed approach and theoretical
analysis remains applicable to problems that can be expressed in the
general form (\ref{eq: Initial full objective}).

Portfolio optimization problems of the form (\ref{eq: Initial full objective})
can give rise to investment strategies that are not time-consistent
due to the presence of the (possibly non-linear) function $G$ (\cite{BjorkEtAl_BOOK_2021}).
Since the objective in (\ref{eq: Initial full objective}) is therefore
potentially not separable in the sense of dynamic programming (see
for example (\ref{eq: Intro - MV optimization}) or (\ref{eq: Intro - Sortino})).
This gives rise to two related problems: (i) Since (\ref{eq: Initial full objective})
cannot be solved using a dynamic programming-based approach, some
other solution methodology has to be implemented, or some re-interpretation
of the problem or the concept of ``optimality'' might be required
(see for example \cite{Vigna2017TailOptimality,BjorkMurgoci2014}),
(ii) if the investment strategies are time-inconsistent, this can
raise questions as to whether these strategies are feasible to implement
as practical investment strategies. 

We make the following general observations:

\begin{itemize}

\item It may be desirable to avoid using dynamic programming (DP)
even if (\ref{eq: Initial full objective}) \textit{can} be solved
using DP techniques, such as in the special case where $G\equiv0$
in (\ref{eq: Initial full objective}) and the investment strategies
are time-consistent. For example, it is well known that DP has an
associated ``curse of dimensionality'', in that as the number state
variables increases linearly, the computational burden increases exponentially
(\cite{HanWeinan2016,FernandezEtAl2020}). In addition, since DP techniques
necessarily incur estimation errors at each time step, significant
error amplification can occur which is further exacerbated in high-dimensional
settings (see for example \cite{WangFoster2020,TsangWong2020,LiEtAl2020}).

However, instead of relying on DP-based techniques and attempting
to address the challenges of dimensionality using machine learning
techniques (see for example \cite{DixonEtAl2020_BOOK,ParkEtAl2020,LucarelliEtAl2020,GaoEtAL2020,FernandezEtAl2020,HenryLabordere2017,HureEtAl2021,BachouchEtAl2021}),
the proposed method fundamentally avoids DP techniques altogether.
This is especially relevant in our setting, since we have shown that
in the case of portfolio optimization problems specifically, DP can
be \textit{unnecessarily} high-dimensional even in simple settings
(see \cite{PvSForsythLi2022_stochbm}). This occurs since the objective
functional {\color{black}(or performance
criteria (\cite{OksendalSulemBook}) )} is typically high-dimensional while the optimal investment
strategy (the fundamental quantity of concern) remains relatively
low-dimensional. The proposed method therefore forms part of the significant
recent interest in developing machine learning techniques to solve
multi-period portfolio optimization problems that avoids using DP techniques altogether
(see for example \cite{LiForsyth2019,TsangWong2020,PvSForsythLi2022_stochbm,NiLiForsyth2020}).

\item Time-inconsistent problems naturally arise in financial applications
(see \cite{BjorkEtAl_BOOK_2021} for numerous examples), and as a
result their solution is often an area of active research due to the
unique challenges involved in solving these problems without resorting
to DP techniques. Examples include the mean-variance problem, which
remained an open problem for decades until the solution using the
embedding technique of \cite{ZhouLi2000,LiNg2000}. As a result, being
able to obtain a numerical solution to problems of the form (\ref{eq: Initial full objective})
directly is potentially very valuable for research.

{\color{black}
The solution of time-inconsistent problems is also practical interest,
since in many cases, there exists an induced time consistent objective
function (\cite{StrubEtAl2019_CVAR,StrubLiCui2019,Forsyth2019CVaR}).
The optimal policy for this induced time consistent objective function is identical
to the pre-commitment policy at time zero.  The induced time consistent
strategy is, of course implementable (\cite{Forsyth2019CVaR}), in the
sense that the investor has no incentive to deviate from the strategy 
determined at time zero, at later times.

An alternative approach to handling time-inconsistent problems is
to search for the equilibrium control (\cite{BjorkEtAl_BOOK_2021}).  A
fascinating result obtained in \cite{BjorkMurgociORIGINAL} is that
for every equilibrium control, there exists a standard, time consistent problem
which has the same control, under a different objective function.

This essentially means that the question of time-consistency is a
often matter of perspective, since there may be alternative objective
functions which give rise to the same pre-commitment control, yet are time-consistent.
In fact, other subtle issues arise in comparing pre-commitment and time
consistent controls, see \cite{Vigna2016TC,Vigna2017TailOptimality} for
further discussion.

Furthermore, over very short time horizons such as those encountered
in optimal trade execution, time consistency or its absence may not
be of much concern to the investor or market participant (see for
example \cite{ForsythKennedyTseWindcliff2011,TseForsythKennedy2012}).

In addition, as noted by
\cite{Bernard_2014}, if the strategy is realized in an investment
product sold to a retail investor, then the optimal policy from the
investor's point of view is in fact of pre-commitment type, since the
retail client does not herself trade in the underlying assets during the
lifetime of the contract.

}

\end{itemize}

As a result of these observations, we will consider problem (\ref{eq: Initial full objective})
in its general form. Our method builds on and formalizes the initial
results described in \cite{LiForsyth2019} where a shallow NN was
applied to a portfolio optimization problem with an objective that
is separable in the sense of DP. The contributions of this paper are
as follows:

\begin{enumerate}[label=(\roman*)]

\item We present a flexible neural network (NN) approach to solve
problems of the form (\ref{eq: Initial full objective}) that does
not rely on DP techniques. Our approach only requires the solution
of a \textit{single} optimization problem, and therefore avoids the
error amplification problems associated with the time-recursion in
DP-based techniques, including for example Reinforcement Learning
algorithms such as Q-learning (see for example \cite{ParkEtAl2020,DixonEtAl2020_BOOK,GaoEtAL2020})
or other algorithms relying at some level on the DP principle for
a time-stepping {\color{black} backward} recursion 
(see for example \cite{BachouchEtAl2021,VanHeeswijkPoutre2019}). 
{\color{black}
Perhaps the best descriptor of our approach is Policy Function Approximation,
in the taxonomy in \cite{Powell_2023}.
}

We make very limited assumptions regarding the underlying asset dynamics.
In particular, if underlying asset (and by extension wealth) dynamics
are specified, this can be incorporated as easily as the case where
the underlying dynamics can only be observed without any parametric
assumptions. 

The proposed solution methodology is \textit{parsimonious}, in that
the number of parameters does not scale with the number of rebalancing
events. This contrasts the proposed methodology with for example that
of \cite{HanWeinan2016,TsangWong2020,HureEtAl2021}, and ensures that
our approach remains feasible even for problems with very long time
horizons (for example the accumulation phase of a pension fund - see
\cite{ForsythVetzalWestmacott2019}) or with shorter time horizon
but with frequent trading/rebalancing (for example the trade execution
problems encountered in \cite{ForsythKennedyTseWindcliff2011}). The
solution approach only places very weak requirements on the form of
the investment objective in (\ref{eq: Initial full objective}). In
addition, we find that using relatively shallow neural networks (at
most two hidden layers) in our approach achieve very accurate results
in ground truth testing, thereby ensuring that the resulting NN in
the proposed approach is relatively easy and efficient to train since
it is less likely to be susceptible to problems of vanishing or exploding
gradients associated with very deep neural networks (\cite{GoodfellowEtAl_BOOK}). 

\item We analyze the convergence of the proposed approach, and show
that the theoretical optimal investment strategy of (\ref{eq: Initial full objective}),
provided it exists, can be attained by the numerical solution.

\item Finally, we present ground truth analyses confirming that the
proposed approach is very effective in solving portfolio optimization
problems of the form (\ref{eq: Initial full objective}). The results
illustrate numerically that if (\ref{eq: Initial full objective})
is not separable in the sense of DP, our approach recovers the correct
pre-commitment (time-inconsistent) optimal control, otherwise it recovers
the correct time-consistent optimal control. To emphasize that the
approach remains agnostic to the underlying data generation assumptions,
results are illustrated using (i) parametric models for asset dynamics,
(ii) stationary block bootstrap resampling of empirical asset returns,
and (ii) generative adversarial network (GAN)-generated synthetic
asset returns. 

\end{enumerate}

The remainder of the paper is organized as follows: Section \ref{sec:Problem-formulation}
presents the problem formulation, while Section \ref{sec: NN approach - overview}
provides a summary of the proposed approach, with additional technical
and practical details provided in Appendix \ref{sec:Appendix A:-Technical details and analytical results}
and Appendix \ref{sec: Appendix B - NN approach - practical considerations}.
Section \ref{sec:Convergence-analysis} presents the convergence analysis
of the proposed approach. Finally, Section \ref{sec:Numerical results}
provides ground truth analyses, with Section \ref{sec:Conclusion}
concluding the paper and discussing possible avenues for future research.

\section{Problem formulation\label{sec:Problem-formulation}}

We start by formulating portfolio optimization problems of the form
(\ref{eq: Initial full objective}) more rigorously in a setting of
discrete portfolio rebalancing and multiple investment constraints.
Throughout, we work on filtered probability space $\left(\Omega,\mathcal{F},\left\{ \mathcal{F}\left(t\right)\right\} _{t\in\left[t_{0},T\right]},\mathbb{P}\right)$
satisfying the usual conditions, with $\mathbb{P}$ denoting the actual
(and not the risk-neutral) probability measure.

Let $\mathcal{T}$ denote the set of $N_{rb}$ discrete portfolio
rebalancing times in $\left[t_{0}=0,T\right]$, which we assume to
be equally-spaced to lighten notation, 
\begin{eqnarray}
\mathcal{\mathcal{T}} & = & \left\{ \left.t_{m}=m\Delta t\right|m=0,...,N_{rb}-1\right\} ,\qquad\Delta t=T/N_{rb},\label{eq: Set of rebalancing times}
\end{eqnarray}
where we observe that the last rebalancing event occurs at time $t_{N_{rb}-1}=T-\Delta t$.

At each rebalancing time $t_{m}\in\mathcal{T}$, the investor observes
the $\mathcal{F}\left(t_{m}\right)$-measurable vector $\boldsymbol{X}\left(t_{m}\right)=\left(X_{i}\left(t_{m}\right):i=1,...,\eta_{X}\right)\in\mathbb{R}^{\eta_{X}}$,
which can be interpreted informally as the information taken into
account by the investor in reaching their asset allocation decision.
As a concrete example, we assume below that $\boldsymbol{X}\left(t_{m}\right)$
includes at least the wealth available for investment, an assumption
which can be rigorously justified using analytical results (see for example
\cite{PvSForsythLi2022_stochbm}).

Given $\boldsymbol{X}\left(t_{m}\right)$, the investor then rebalances
a portfolio of $N_{a}$ assets to new positions given by the vector
\begin{eqnarray}
\boldsymbol{p}_{m}\left(t_{m},\boldsymbol{X}\left(t_{m}\right)\right) & = & \left(p_{m,i}\left(t_{m},\boldsymbol{X}\left(t_{m}\right)\right):i=1,..,N_{a}\right)\in\mathbb{R}^{N_{a}},\label{eq: Control at t_m}
\end{eqnarray}
where $p_{m,i}\left(t_{m},\boldsymbol{X}\left(t_{m}\right)\right)$
denotes the fraction of wealth $W\left(t_{m}\right)$ invested in
the $i$th asset at rebalancing time $t_{m}$. The subscript ``$m$''
in the notation $\boldsymbol{p}_{m}$ emphasizes that in general,
each rebalancing time $t_{m}\in\mathcal{T}$ could be associated with
potentially a different function $\boldsymbol{p}_{m}:\mathbb{R}^{\eta_{X}+1}\rightarrow\mathbb{R}^{N_{a}}$,
while the subscript is removed below when we consider a single function
that is simply evaluated at different times, in which case we will
write $\boldsymbol{p}:\mathbb{R}^{\eta_{X}+1}\rightarrow\mathbb{R}^{N_{a}}$.

For purposes of concreteness, we assume that the investor is subject
to the constraints of (i) no short-selling and (ii) no leverage being
allowed, although the proposed methodology can be adjusted without
difficulty to treat different constraint formulations\footnote{As discussed in Section \ref{sec: NN approach - overview} and Appendix
\ref{sec:Appendix A:-Technical details and analytical results}, adjustments
to the output layer of the neural network may be required.}. For illustrative purposes, we therefore assume that each allocation
(\ref{eq: Control at t_m}) is only allowed to take values in $\left(N_{a}-1\right)$-dimensional
probability simplex $\mathcal{Z}$, 
\begin{eqnarray}
\mathcal{Z} & = & \left\{ \left(y_{1},...,y_{N_{a}}\right)\in\mathbb{R}^{N_{a}}:\sum_{i=1}^{N_{a}}y_{i}=1\textrm{ and }y_{i}\geq0\textrm{ for all }i=1,...,N_{a}\right\} .\label{eq: Set Z}
\end{eqnarray}

In this setting, an investment strategy or control $\mathcal{P}$
applicable to $\left[t_{0},T\right]$ is therefore of the form,
\begin{eqnarray}
\mathcal{P} & = & \left\{ \boldsymbol{p}_{m}\left(t_{m},\boldsymbol{X}\left(t_{m}\right)\right)=\left(p_{m,i}\left(t_{m},\boldsymbol{X}\left(t_{m}\right)\right):i=1,..,N_{a}\right):t_{m}\in\mathcal{T}\right\} ,\label{eq: Control - generic}
\end{eqnarray}
while the set of admissible controls $\mathcal{A}$ is defined by
\begin{eqnarray}
\mathcal{A} & = & \left\{ \left.\mathcal{P}=\left\{ \boldsymbol{p}_{m}\left(t_{m},\boldsymbol{X}\left(t_{m}\right)\right):t_{m}\in\mathcal{T}\right\} \right|\boldsymbol{p}_{m}\left(t_{m},\boldsymbol{X}\left(t_{m}\right)\right)\in\mathcal{Z},\forall t_{m}\in\mathcal{T}\right\} .\label{eq: Set A}
\end{eqnarray}

The randomness in the system is introduced through the returns of
the underlying assets. Specifically, let $R_{i}\left(t_{m}\right)$
denote the $\mathcal{F}\left(t_{m+1}\right)$-measurable return observed
on asset $i$ over the interval $\left[t_{m},t_{m+1}\right]$. We
make no assumptions regarding the underlying asset dynamics, but at
a minimum, we do require ($\mathbb{P}$) integrability, i.e. $\mathbb{E}\left|R_{i}\left(t_{m}\right)\right|<\infty$
for all $i\in\left\{ 1,...,N_{a}\right\} $ and $m\in\left\{ 0,...,N_{rb}-1\right\} $.
Informally, we will refer to the set 
\begin{eqnarray}
\boldsymbol{Y} & = & \left\{ \left(Y_{i}\left(t_{m}\right)\coloneqq1+R_{i}\left(t_{m}\right):i=1,...,N_{a}\right)^{\top}:m\in\left\{ 0,...,N_{rb}-1\right\} \right\} \label{eq: Path Y of asset returns}
\end{eqnarray}
as the \textit{path} of (joint) asset returns over the investment
time horizon $\left[t_{0},T\right]$. 

To clarify the subsequent notation, for any functional $\psi\left(t\right),t\in\left[t_{0},T\right]$
we will use the notation $\psi\left(t^{-}\right)$ and $\psi\left(t^{+}\right)$
as shorthand for the one-sided limits $\psi\left(t^{-}\right)=\lim_{\epsilon\downarrow0}\psi\left(t-\epsilon\right)$
and $\psi\left(t^{+}\right)=\lim_{\epsilon\downarrow0}\psi\left(t+\epsilon\right)$,
respectively. 

Given control $\mathcal{P}\in\mathcal{A}$, asset returns $\boldsymbol{Y}$,
initial wealth $W\left(t_{0}^{-}\right)\coloneqq w_{0}>0$ and a (non-random)
cash contribution schedule $\left\{ q\left(t_{m}\right):t_{m}\in\mathcal{T}\right\} $,
the portfolio wealth dynamics for $m=0,...,N_{rb}-1$ are given by
the general recursion
\begin{eqnarray}
W\left(t_{m+1}^{-};\mathcal{P},\boldsymbol{Y}\right) & = & \left[W\left(t_{m}^{-};\mathcal{P},\boldsymbol{Y}\right)+q\left(t_{m}\right)\right]\cdot\sum_{i=1}^{N_{a}}p_{m,i}\left(t_{m},\boldsymbol{X}\left(t_{m}\right)\right)\cdot Y_{i}\left(t_{m}\right).\label{eq: W dynamics - general formulation}
\end{eqnarray}
Note that we write $W\left(u\right)=W\left(u;\mathcal{P},\boldsymbol{Y}\right)$
to emphasize the dependence of wealth on the control $\mathcal{P}$
and the (random) path of asset returns in $\boldsymbol{Y}$ that relates
to the time period $t\in\left[t_{0},u\right]$. In other words, despite
using $\boldsymbol{Y}$ in the notation for simplicity, $W\left(u;\mathcal{P},\boldsymbol{Y}\right)$
is $\mathcal{F}\left(u\right)$-measurable. Finally, since there are
no contributions or rebalancing at maturity, we simply have $W\left(t_{N_{rb}}^{-}\right)=W\left(T^{-}\right)=W\left(T\right)=W\left(T;\mathcal{P},\boldsymbol{Y}\right)$. 

\subsection{Investment objectives\label{subsec:Investment-objectives}}

Given this general investment setting and wealth dynamics (\ref{eq: W dynamics - general formulation}),
our goal is to solve dynamic portfolio optimization problems of the
general form
\begin{equation}
\inf_{\xi\in\mathbb{R}}\inf_{\mathcal{P}\in\mathcal{A}}J\left(\mathcal{P},\xi;t_{0},w_{0}\right),\label{eq: Basic general problem}
\end{equation}
where, for some given continuous functions $F:\mathbb{R}^{2}\rightarrow\mathbb{R}$
and $G:\mathbb{R}^{3}\rightarrow\mathbb{R}$, the objective functional
$J$ is given by 
\begin{eqnarray}
J\left(\mathcal{P},\xi;t_{0},w_{0}\right) & = & E_{\mathcal{P}}^{t_{0},w_{0}}\left[F\left(W\left(T;\mathcal{P},\boldsymbol{Y}\right),\xi\right)+G\left(W\left(T;\mathcal{P},\boldsymbol{Y}\right),E_{\mathcal{P}}^{t_{0},w_{0}}\left[W\left(T;\mathcal{P},\boldsymbol{Y}\right)\right],w_{0},\xi\right)\begin{array}{c}
\!\!\!\!\!\!\!\!\\
\!\!\!\!\!\!\!\!
\end{array}\right].\label{eq: Objective functional J}
\end{eqnarray}
Note that the expectations $E^{t_{0},w_{0}}\left[\cdot\right]$ in
(\ref{eq: Objective functional J}) are taken over $\boldsymbol{Y}$,
given initial wealth $W\left(t_{0}^{-}\right)=w_{0}$, control $\mathcal{P}\in\mathcal{A}$
and auxiliary variable $\xi\in\mathbb{R}$. In addition to the assumption
of continuity of $F$ and $G$, we will make only the minimal assumptions
regarding the exact properties of $J$, including that $\xi\rightarrow F\left(\cdot,\xi\right)$
and $\xi\rightarrow G\left(\cdot,\cdot,w_{0},\xi\right)$ are convex
for all admissible controls $\mathcal{P}\in\mathcal{A}$, and the
standard assumption (see for example \cite{BjorkEtAl_BOOK_2021})
that an optimal control $\mathcal{P}^{\ast}\in\mathcal{A}$ exists.

For illustrative and ground truth analysis purposes, we consider a
number of examples of problems of the form (\ref{eq: Basic general problem})-(\ref{eq: Objective functional J}).

As noted in the Introduction, the simplest examples of problems of
the form (\ref{eq: Basic general problem}) arise in the special case
where $G\equiv0$ and there is no outer optimization problem over
$\xi$, such as in the case of standard utility maximization problems.
As concrete examples of this class of objective functions, we will
consider the quadratic target minimization (or quadratic utility)
described in for example \cite{Vigna_efficiency2014,ZhouLi2000},
\begin{eqnarray}
\left(DSQ\left(\gamma\right)\right): &  & \inf_{\mathcal{P}\in\mathcal{A}}\left\{ E^{t_{0},w_{0}}\left[\left(W\left(T;\mathcal{P},\boldsymbol{Y}\right)-\gamma\right)^{2}\right]\right\} ,\qquad\gamma>0,\label{eq: DSQ objective}
\end{eqnarray}
as well as the (closely-related) one-sided quadratic loss minimization
used in for example \cite{DangForsyth2016,LiForsyth2019}, 
{\color{black}
\begin{eqnarray}
\left(OSQ\left(\gamma\right)\right):\qquad &  & 
     \inf_{\mathcal{P}\in\mathcal{A}}
        \left\{ 
              E^{t_{0},w_{0}}
          \left[ 
              \left(
                \min \left\{ 
                          W\left(T;\mathcal{P},\boldsymbol{Y}\right) -\gamma, 0
                    \right\} 
              \right)^{2}
                - \epsilon \cdot W\left(T;\mathcal{P},\boldsymbol{Y}\right)
          \right]
      \right\} ,
                  \qquad  \gamma>0.\label{eq: OSQ objective}
\end{eqnarray}
}
The term $\epsilon W(\cdot)$ in equation (\ref{eq: OSQ objective}) ensures
that the problem remains well-posed\footnote{
{\color{black}  Although this is a
mathematical necessity (see e.g. \citep{LiForsyth2019}), in practice, if we use a very small value of $\epsilon$, then
this has no perceptible effect on the summary statistics. In the numerical results of Section \ref{sec:Numerical results}, we use $\epsilon = 10^{-6}$; see Appendix \ref{sec: Appendix B - NN approach - practical considerations} for a discussion.}} in the event that $W(t) \gg \gamma$. Observe that problems of the form (\ref{eq: DSQ objective})
or (\ref{eq: OSQ objective}) are separable in the sense of dynamic
programming, so that the resulting optimal control is therefore time-consistent.

As a classical example of the case where $G$ is nonlinear and the
objective functional (\ref{eq: Objective functional J}) is not separable
in the sense of dynamic programming, we consider the mean-variance
(MV) objective with scalarization or risk-aversion parameter $\rho>0$
(see for example \cite{BjorkKhapkoMurgoci2016}),
\begin{eqnarray}
\left(MV\left(\rho\right)\right): &  & \sup_{\mathcal{P}\in\mathcal{A}}\left\{ E^{t_{0},w_{0}}\left[W\left(T;\mathcal{P},\boldsymbol{Y}\right)\right]-\rho\cdot Var^{t_{0},w_{0}}\left[W\left(T;\mathcal{P},\boldsymbol{Y}\right)\right]\right\} ,\qquad\rho>0.\nonumber \\
 & = & \sup_{\mathcal{P}\in\mathcal{A}}E_{\mathcal{P}}^{t_{0},w_{0}}\left[W\left(T;\mathcal{P},\boldsymbol{Y}\right)-\rho\cdot\left(W\left(T;\mathcal{P},\boldsymbol{Y}\right)-E_{\mathcal{P}}^{t_{0},w_{0}}\left[W\left(T;\mathcal{P},\boldsymbol{Y}\right)\right]\right)^{2}\right].\label{eq: MV objective}
\end{eqnarray}
Note that issues relating to the time-inconsistency of the optimal
control of (\ref{eq: MV objective}) are discussed in Remark \ref{rem: Time-consistency}
below, along with the relationship between (\ref{eq: DSQ objective})
and (\ref{eq: MV objective}).

As an example of a problem involving both the inner and outer optimization
in (\ref{eq: Basic general problem}), we consider the Mean - Conditional
Value-at-Risk (or Mean-CVaR) problem, subsequently simply abbreviated
the MCV problem. First, as a measure of tail risk, the CVaR at level
$\alpha$, or $\alpha$-CVaR, is the expected value of the worst $\alpha$
percent of wealth outcomes, with typical values being $\alpha\in\left\{ 1\%,5\%\right\} $.
As in \cite{Forsyth2019CVaR}, a \textit{larger} value of the CVaR
is preferable to smaller value, since our definition of $\alpha$-CVaR
is formulated in terms of the terminal \textit{wealth}, not in terms
of the \textit{loss}. Informally, if the distribution of terminal
wealth $W\left(T\right)$ is continuous with PDF $\hat{\psi}$, then
the $\alpha$-CVaR in this case is given by
\begin{eqnarray}
\textrm{CVAR}_{\alpha} & = & \frac{1}{\alpha}\int_{-\infty}^{w_{\alpha}^{\ast}}W\left(T\right)\cdot\hat{\psi}\left(W\left(T\right)\right)\cdot dW\left(T\right),\label{eq: alpha-CVaR definition}
\end{eqnarray}
where $w_{\alpha}^{\ast}$ is the corresponding Value-at-Risk (VaR)
at level $\alpha$ defined such that $\int_{-\infty}^{w_{\alpha}^{\ast}}\hat{\psi}\left(W\left(T\right)\right)dW\left(T\right)=\alpha$.
We follow for example \cite{Forsyth2019CVaR} in defining the MCV
problem with scalarization parameter $\rho>0$ formally as
\begin{equation}
\sup_{\mathcal{P}\in\mathcal{A}}\left\{ \rho\cdot E^{t_{0},w_{0}}\left[W\left(T\right)\right]+\textrm{CVAR}_{\alpha}\right\} ,\qquad\rho>0.\label{eq: Mean-CVaR - objective 1}
\end{equation}
However, instead of (\ref{eq: alpha-CVaR definition}), we use the
definition of CVaR from \cite{RockafellarUryasev2002} that is applicable
to more general terminal wealth distributions, so that the MCV problem
definition used subsequently aligns with the definition given in \cite{MillerYang2017,Forsyth2019CVaR}),
\begin{eqnarray}
\left(MCV\left(\rho\right)\right): &  & \inf_{\xi\in\mathbb{R}}
   \inf_{\mathcal{P}\in\mathcal{A}}E^{t_{0},w_{0}}\left[-\rho\cdot W\left(T;\mathcal{P},\boldsymbol{Y}\right)
     - \xi+\frac{1}{\alpha}\max\left(\xi-W\left(T;\mathcal{P},\boldsymbol{Y}\right),0\right)\right],\quad\rho>0.
       \label{eq: MCV objective}
\end{eqnarray}

Finally, as noted in the Introduction, we apply the ideas underlying
the Sortino ratio where the variance of returns below the mean are
penalized, to formulate the following objective function for dynamic
trading,
\begin{equation}
\left(MSemiV\left(\rho\right)\right):\sup_{\mathcal{P}\in\mathcal{A}}
    \left\{ E_{\mathcal{P}}^{t_{0},w_{0}}\left[W\left(T;\mathcal{P},\boldsymbol{Y}\right)
      -\rho\cdot\left(\min\left\{ W\left(T;\mathcal{P},\boldsymbol{Y}\right)
      -E_{\mathcal{P}}^{t_{0},w_{0}}\left[W\left(T;\mathcal{P},\boldsymbol{Y}\right)\right],0\right\} \right)^{2}\right]\right\} 
     ,\label{eq: Sortino}
\end{equation}
which we refer to as the ``Mean- Semi-variance'' problem, with scalarization
(or risk-aversion) parameter $\rho>0$.\footnote{
{\color{black}
In continuous time, the unconstrained Mean-Semi-variance problem is ill-posed (\cite{Jin_2005}).
However, we will impose bounded leverage constraints, which is, of course,
a realistic condition.  This makes problem $\left(MSemiV\left(\rho\right)\right)$
well posed.
}
}

The following remark discusses issues relating to the possible time-inconsistency
of the optimal controls of (\ref{eq: MV objective}) , (\ref{eq: MCV objective})
and (\ref{eq: Sortino}).

{\color{black}
\noindent \begin{brem}\label{rem: Time-consistency}(Time-\textit{in}consistency
and induced time-consistency) Formally, the optimal controls for problems
$MV\left(\rho\right)$, $MCV\left(\rho\right)$ and $MSemiV\left(\rho\right)$
are not time-consistent, but instead are of the pre-commitment type
(see \cite{BasakChabakauri2010,BjorkMurgoci2014,Forsyth2019CVaR}).
However, in many cases, there exists an induced time consistent problem
formulation which has the same controls at time zero as the 
pre-commitment problem (see \cite{StrubLiCui2019,Forsyth2019CVaR,StrubEtAl2019_CVAR}).

As a concrete example of induced time-consistency, the
embedding result of \cite{LiNg2000,ZhouLi2000} establishes that the
$DSQ\left(\gamma\right)$ objective is the induced time-consistent
objective function associated with the $MV\left(\rho\right)$ problem,
which is a result that we exploit for ground truth analysis purposes
in Section \ref{sec:Numerical results}. 

Similarly, there is an induced time consistent objective
function for the  Mean-CVAR problem $MCV\left(\rho\right)$ in (\ref{eq: MCV objective}) - see  
\cite{Forsyth2019CVaR}.

Consequently, when we refer to a strategy as optimal, for either
the Mean-CVAR ($MCV\left(\rho\right)$) or Mean-Variance ($MV\left(\rho\right)$) problems, this will
be understood to mean that at any $t>t_{0}$, the investor follows
the associated induced time-consistent strategy rather than a pre-commitment
strategy. 

In the Mean-Semi-variance $\left(MSemiV\left(\rho\right)\right)$ case as per (\ref{eq: Sortino}), there is no obvious induced time consistent objective function.
In this case, we seek the pre-commitment policy.

For a detailed discussion of the many subtle issues involved
in the case of time-inconsistency, induced time-consistency, 
and equilibrium controls, see
for example \cite{Forsyth2019CVaR,StrubLiCui2019,StrubEtAl2019_CVAR,Vigna2017TailOptimality,Vigna2016TC,BjorkMurgoci2014,BjorkEtAl_BOOK_2021}.
\qed\end{brem}
}

\section{Neural network approach\label{sec: NN approach - overview}}

In this section, we provide an overview of the neural network (NN)
approach. Additional technical details and practical considerations
are discussed in Appendices \ref{sec:Appendix A:-Technical details and analytical results}
and \ref{sec: Appendix B - NN approach - practical considerations},
while the theoretical justification via convergence analysis will
be discussed in Section \ref{sec:Convergence-analysis} (and Appendix
\ref{sec: Appendix B - NN approach - practical considerations}).

Recall from (\ref{eq: Control at t_m}) that $\boldsymbol{X}\left(t_{m}\right)\in\mathbb{R}^{\eta_{X}}$
denotes the information taken into account in determining the investment
strategy (\ref{eq: Control at t_m}) at rebalancing time $t_{m}$.
Using the initial experimental results of \cite{LiForsyth2019} and
the analytical results of \cite{PvSForsythLi2022_stochbm} applied
to this setting, we assume that $\boldsymbol{X}\left(t_{m}\right)$
includes at least the wealth available for investment at time $t_{m}$,
so that 
\begin{eqnarray}
W\left(t_{m}^{+};\mathcal{P},\boldsymbol{Y}\right)\coloneqq W\left(t_{m}^{-};\mathcal{P},\boldsymbol{Y}\right)+q\left(t_{m}\right) & \in & \boldsymbol{X}\left(t_{m}\right),\qquad\forall t_{m}\in\mathcal{T}.\label{eq: Minimal form of X}
\end{eqnarray}
However, we emphasize that $\boldsymbol{X}\left(t_{m}\right)$ may
include additional variables in different settings. For example, in
non-Markovian settings or in the case of certain solution approaches
involving auxiliary variables, it is natural to ``lift the state
space'' by including additional quantities in $\boldsymbol{X}$ such
as relevant historical quantities related to market variables, or
other auxiliary variables - see for example \cite{Forsyth2019CVaR,MillerYang2017,TsangWong2020}. 

Let $\mathcal{D}_{\boldsymbol{\phi}}\subseteq\mathbb{R}^{\eta_{X}+1}$
be the set such that $\left(t_{m},\boldsymbol{X}\left(t_{m}\right)\right)\in\mathcal{D}_{\boldsymbol{\phi}}$
for all $t_{m}\in\mathcal{T}$. Let $C\left(\mathcal{D}_{\boldsymbol{\phi}},\mathcal{Z}\right)$
denote the set of all continuous functions from $\mathcal{D}_{\boldsymbol{\phi}}$
to $\mathcal{Z}\subset\mathbb{R}^{N_{a}}$ (see (\ref{eq: Set Z})).
We will use the notation $\boldsymbol{X}^{\ast}$ to denote the information
taken into account by the optimal control, since in the simplest case
implied by (\ref{eq: Minimal form of X}), we simply have $\boldsymbol{X}^{\ast}=W^{\ast}$,
where $W^{\ast}$ denotes the wealth under the optimal strategy. We
make the following assumption.
\begin{assumption}
\label{assu: Existence and Continuity of the control}(Properties
of the optimal control) Considering the general form of the problem
(\ref{eq: Basic general problem}), we assume that there exists an
optimal feedback control $\mathcal{P}^{\ast}\in\mathcal{A}$. Specifically,
we assume that at each rebalancing time $t_{m}\in\mathcal{T}$, the
time $t_{m}$ itself together with the information vector under optimal
behavior $\boldsymbol{X}^{\ast}\left(t_{m}\right)$, which includes
at least the wealth $W^{\ast}\left(t_{m}^{+}\right)$ available for
investment (see (\ref{eq: Minimal form of X})), are sufficient to
fully determine the optimal asset allocation $\boldsymbol{p}_{m}^{\ast}\left(t_{m},\boldsymbol{X}^{\ast}\left(t_{m}\right)\right)$. 

Furthermore, we assume that there exists a continuous function $\boldsymbol{p}^{\ast}\in C\left(\mathcal{D}_{\boldsymbol{\phi}},\mathcal{Z}\right)$
such that $\boldsymbol{p}_{m}^{\ast}\left(t_{m},\boldsymbol{X}^{\ast}\left(t_{m}\right)\right)=\boldsymbol{p}^{\ast}\left(t_{m},\boldsymbol{X}^{\ast}\left(t_{m}\right)\right)$
for all $t_{m}\in\mathcal{T}$, so that the optimal control $\mathcal{P}^{\ast}$
can be expressed as 
\begin{eqnarray}
\mathcal{P}^{\ast} & = & \left\{ \boldsymbol{p}^{\ast}\left(t_{m},\boldsymbol{X}^{\ast}\left(t_{m}\right)\right):\forall t_{m}\in\mathcal{T}\right\} ,\quad\textrm{where }\quad\boldsymbol{p}^{\ast}\in C\left(\mathcal{D}_{\boldsymbol{\phi}},\mathcal{Z}\right).\label{eq: Assumption CONTINUOUS control}
\end{eqnarray}
\end{assumption}

We make the following observations regarding Assumption \ref{assu: Existence and Continuity of the control}:
\begin{enumerate}
\item Continuity of $\boldsymbol{p}^{\ast}$ in space \textit{and time}:
While assuming the optimal control is a continuous map in the state
space $\boldsymbol{X}$ is fairly standard in the literature, especially
in the context of using neural network approximations (see for example
\cite{HanWeinan2016,TsangWong2020,HureEtAl2021}), the assumption
of continuity in time in (\ref{eq: Assumption CONTINUOUS control})
is therefore worth emphasizing. This assumption enforces the requirement
that in the limit of continuous rebalancing (i.e. when $\Delta t\rightarrow0$),
the control remains a continuous function of time, which is a practical
requirement for any reasonable investment policy. In particular, this
ensures that the asset allocation retains its smooth behavior as the
number of rebalancing events in $\left[0,T\right]$ is increased,
which we consider a fundamental requirement ensuring that the resulting
investment strategy is reasonable. In addition, in Section \ref{sec:Numerical results}
we demonstrate how the known theoretical solution to a problem assuming
continuous rebalancing ($\Delta t\rightarrow0$) can be approximated
very well using $\Delta t\gg0$ in the NN approach, even though the
resulting NN approximation is only truly optimal in the case of $\Delta t\gg0$.
\item The control is a \textit{single} function for \textit{all} rebalancing
times; note that the function $\boldsymbol{p}^{\ast}$ is not subscripted
by time. If the portfolio is rebalanced only at discrete time intervals,
the investment strategy can be found (as suggested in (\ref{eq: Assumption CONTINUOUS control}))
by evaluating this continuous function at discrete time intervals,
i.e. $\left(t_{m},\boldsymbol{X}\left(t_{m}\right)\right)\rightarrow\boldsymbol{p}^{\ast}\left(t_{m},\boldsymbol{X}\left(t_{m}\right)\right)=\left(p_{i}^{\ast}\left(t_{m},\boldsymbol{X}\left(t_{m}\right)\right):i=1,...,N_{a}\right)$,
for all $t_{m}\in\mathcal{T}$. We discuss below how we solve for
this (single) function directly, without resorting to dynamic programming,
which avoids not only the challenge with error propagation due to
value iteration over multiple timesteps, but also avoids solving for
the high-dimensional conditional expectation (also termed
the performance criteria by \cite{OksendalSulemBook}) if we are only interested in the
relatively low-dimensional optimal control (see for example \cite{PvSForsythLi2022_stochbm}). 
\end{enumerate}
These observations ultimately suggest the NN approach discussed below,
while the soundness of Assumption \ref{assu: Existence and Continuity of the control}
is experimentally confirmed in the ground truth results presented
in Section \ref{sec:Numerical results}.

Given Assumption \ref{assu: Existence and Continuity of the control}
and in particular (\ref{eq: Assumption CONTINUOUS control}), we therefore
limit our consideration to controls of the form
\begin{eqnarray}
\mathcal{P} & = & \left\{ \boldsymbol{p}\left(t_{m},\boldsymbol{X}\left(t_{m}\right)\right):\forall t_{m}\in\mathcal{T}\right\} ,\quad\textrm{for some }\quad\boldsymbol{p}\in C\left(\mathcal{D}_{\boldsymbol{\phi}},\mathcal{Z}\right).\label{eq: Assumption general control with continuity}
\end{eqnarray}
To simplify notation, we identify an arbitrary control $\mathcal{P}$
of the form (\ref{eq: Assumption general control with continuity})
with its associated function $\boldsymbol{p}=\left(p_{i}:i=1,...,N_{a}\right)\in C\left(\mathcal{D}_{\boldsymbol{\phi}},\mathcal{Z}\right)$,
so that the objective functional (\ref{eq: Objective functional J})
is written as
\begin{eqnarray}
J\left(\boldsymbol{p},\xi;t_{0},w_{0}\right) & = & E^{t_{0},w_{0}}\left[F\left(W\left(T;\boldsymbol{p},\boldsymbol{Y}\right),\xi\right)+G\left(W\left(T\right),E^{t_{0},w_{0}}\left[W\left(T;\boldsymbol{p},\boldsymbol{Y}\right)\right],w_{0},\xi\right)\begin{array}{c}
\!\!\!\!\!\!\!\!\\
\!\!\!\!\!\!\!\!
\end{array}\right].\label{eq: Objective functional with continuous control p}
\end{eqnarray}
In (\ref{eq: Objective functional with continuous control p}), $W\left(\cdot;\boldsymbol{p},\boldsymbol{Y}\right)$
denotes the controlled wealth process using a control of the form
(\ref{eq: Assumption general control with continuity}), so that the
wealth dynamics (\ref{eq: W dynamics - general formulation}) for
$t_{m}\in\mathcal{T}$ (recall $t_{N_{rb}}^{-}=T$) now becomes 
\begin{eqnarray}
W\left(t_{m+1}^{-};\boldsymbol{p},\boldsymbol{Y}\right) & = & \left[W\left(t_{m}^{-};\boldsymbol{p},\boldsymbol{Y}\right)+q\left(t_{m}\right)\right]\cdot\sum_{i=1}^{N_{a}}p_{i}\left(t_{m},\boldsymbol{X}\left(t_{m}\right)\right)\cdot Y_{i}\left(t_{m}\right).\label{eq: W dynamics - continuous control}
\end{eqnarray}

Therefore, using Assumption \ref{assu: Existence and Continuity of the control}
and (\ref{eq: Objective functional with continuous control p})-(\ref{eq: W dynamics - continuous control}),
problem (\ref{eq: Basic general problem}) is therefore expressed
as 

\begin{eqnarray}
V\left(t_{0},w_{0}\right) & = & \inf_{\xi\in\mathbb{R}}\inf_{\boldsymbol{p}\in C\left(\mathcal{D}_{\boldsymbol{\phi}},\mathcal{Z}\right)}J\left(\boldsymbol{p},\xi;t_{0},w_{0}\right).\label{eq: Original problem - CONTINUOUS control}
\end{eqnarray}

We now provide a brief overview of the proposed methodology to solve
problems of the form (\ref{eq: Original problem - CONTINUOUS control}).
This consists of two steps discussed in the following subsections,
namely (i) the NN approximation to the control, and (ii) computational
estimate of the optimal control. 

\subsection{Step 1: NN approximation to control\label{subsec:Step 1 NN approx}}

Let $n\in\mathbb{N}$. Consider a fully-connected, feedforward NN
$\boldsymbol{f}_{n}$ with parameter vector $\boldsymbol{\theta}_{n}\in\mathbb{R}^{\nu_{n}}$
and a fixed number $\mathcal{L}^{h}\geq1$ of hidden layers, where
each hidden layer contains $\hbar\left(n\right)\in\mathbb{N}$ nodes.
The NN has $\left(\eta_{X}+1\right)$ input nodes, mapping feature
(input) vectors of the form $\boldsymbol{\phi}\left(t\right)=\left(t,\boldsymbol{X}\left(t\right)\right)\in\mathcal{D}_{\boldsymbol{\phi}}$
to $N_{a}$ output nodes. For a more detailed introduction to neural
networks, see for example \cite{GoodfellowEtAl_BOOK}. 

Additional technical and practical details can be found in Appendices
\ref{sec:Appendix A:-Technical details and analytical results} and
\ref{sec: Appendix B - NN approach - practical considerations}. For
this discussion, we simply note that the index $n\in\mathbb{N}$ is
used for the purposes of the analytical results and convergence analysis,
where we fix a choice of $\mathcal{L}^{h}\geq1$ while $\hbar\left(n\right),n\in\mathbb{N}$
is assumed to be a monotonically increasing sequence such that $\lim_{n\rightarrow\infty}\hbar\left(n\right)=\infty$
(see Section \ref{sec:Convergence-analysis} and Appendix \ref{sec:Appendix A:-Technical details and analytical results}).
However, for practical implementation, a fixed value of $\hbar\left(n\right)\in\mathbb{N}$
is chosen (along with $\mathcal{L}^{h}\geq1$) to ensure the NN has
sufficient depth and complexity to solve the problem under consideration
(see Appendix \ref{sec: Appendix B - NN approach - practical considerations}).

Any NN considered is constructed such that $\boldsymbol{f}_{n}:\mathcal{D}_{\boldsymbol{\phi}}\rightarrow\mathcal{Z}\subset\mathbb{R}^{N_{a}}$.
In other words, the values of the $N_{a}$ outputs are automatically
in the set $\mathcal{Z}$ defined in (\ref{eq: Set Z}) for any $\boldsymbol{\phi}\in\mathcal{D}_{\boldsymbol{\phi}}$,
\begin{eqnarray}
\boldsymbol{f}_{n}\left(\boldsymbol{\phi}\left(t\right);\boldsymbol{\theta}_{n}\right) & = & \left(f_{n,i}\left(\boldsymbol{\phi}\left(t\right);\boldsymbol{\theta}_{n}\right):i=1,...,N_{a}\right)\in\mathcal{Z}.\label{eq: NN f_n as a vector}
\end{eqnarray}
As a result, the outputs of the NN $\boldsymbol{f}_{n}$ in (\ref{eq: NN f_n as a vector})
can be interpreted as portfolio weights satisfying the required investment
constraints. While a more detailed discussion of the structure can
be found in Assumption \ref{assu: Appendix NN structure assumptions}
in Appendix \ref{sec:Appendix A:-Technical details and analytical results},
we summarize some key aspects of the NN structure illustrated in Figure
\ref{fig: NN diagram}:
\begin{enumerate}
\item We emphasize that the rebalancing time is an \textit{input} into the
NN as per the feature vector $\boldsymbol{\phi}\left(t\right)=\left(t,\boldsymbol{X}\left(t\right)\right)\in\mathcal{D}_{\boldsymbol{\phi}}$,
so that the NN parameter vector $\boldsymbol{\theta}_{n}$ itself
does not depend on time. 
\item While we assume sigmoid activations for the hidden nodes for concreteness
and convenience (see Assumption \ref{assu: Appendix NN structure assumptions}),
any of the commonly-used activation functions can be implemented with
only minor modifications to the technical results presented in Section
\ref{sec:Convergence-analysis}.
\item Since we are illustrating the approach using the particular form of
$\mathcal{Z}$ in (\ref{eq: Set Z}) because of its wide applicability
(no short-selling and no leverage), a softmax output layer is used
to ensure the NN output remains in $\mathcal{Z}\subset\mathbb{R}^{N_{a}}$
for any $\boldsymbol{\phi}\left(t\right)$ (see (\ref{eq: NN f_n as a vector})).
However, different admissible control set formulations can be handled
without difficulty\footnote{For example, position limits and limited leverage can be introduced
using minor modifications to the output layer. Perhaps the only substantial
challenge is offered by unrealistic investment scenarios, such as
insisting that trading should continue in the event of bankruptcy,
in which case consideration should be given to the possibility of
wealth being identically zero or negative.}. 
\end{enumerate}
\noindent 
\begin{figure}[!tbh]
\noindent \begin{centering}
\includegraphics[scale=0.45]{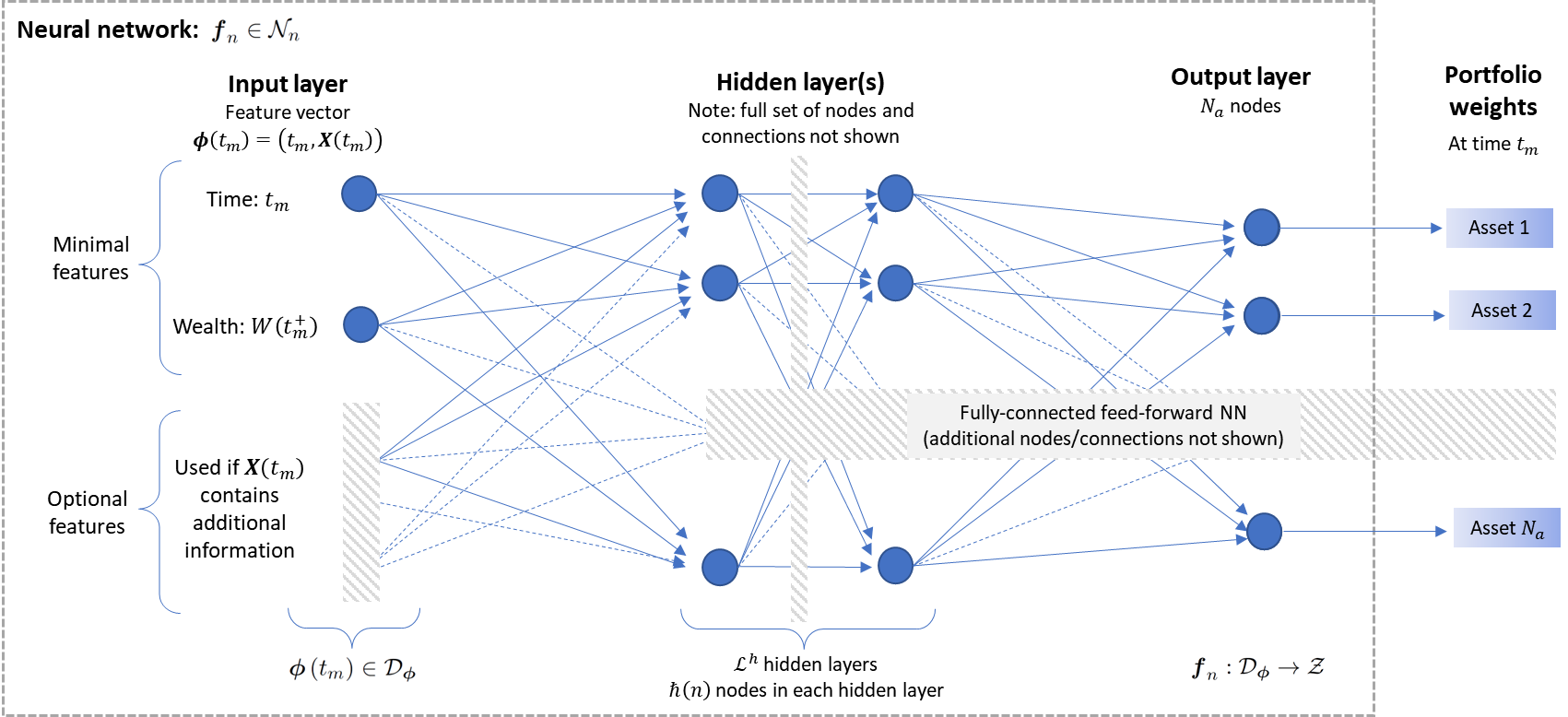} 
\par\end{centering}
\caption{Illustration of the structure of the NN as per (\ref{eq: NN f_n as a vector}).
Additional construction and implementation details can be found in
Appendix \ref{sec:Appendix A:-Technical details and analytical results}.
\label{fig: NN diagram}}
\end{figure}

For some fixed value of the index $n\in\mathbb{N}$, let $\mathcal{N}_{n}$
denote the set of NNs constructed in the same way as $\boldsymbol{f}_{n}$
for the fixed and given values of $\mathcal{L}^{h}$ and $\hbar\left(n\right)$.
While a formal definition of the set $\mathcal{N}_{n}$ is provided
in Appendix \ref{sec:Appendix A:-Technical details and analytical results},
here we simply note that each NN $\boldsymbol{f}_{n}\left(\cdot;\boldsymbol{\theta}_{n}\right)\in\mathcal{N}_{n}$
only differs in terms of the parameter values constituting its parameter
vector $\boldsymbol{\theta}_{n}$ (i.e. for a fixed $n$, each $\boldsymbol{f}_{n}\in\mathcal{N}_{n}$
has the same number of hidden layers $\mathcal{L}^{h}$, hidden nodes
$\hbar\left(n\right)$, activation functions etc.). 

Observing that $\mathcal{N}_{n}\subset C\left(\mathcal{D}_{\boldsymbol{\phi}},\mathcal{Z}\right)$,
our first step is to approximate (\ref{eq: Original problem - CONTINUOUS control})
by performing the optimization over $\boldsymbol{f}_{n}\left(\cdot;\boldsymbol{\theta}_{n}\right)\in\mathcal{N}_{n}$
instead. In other words, we approximate the control $\boldsymbol{p}$
by a neural network $\boldsymbol{f}_{n}\in\mathcal{N}_{n}$, 
\begin{eqnarray}
\boldsymbol{p}\left(\boldsymbol{\phi}\left(t\right)\right) & \simeq & \boldsymbol{f}_{n}\left(\boldsymbol{\phi}\left(t\right);\boldsymbol{\theta}_{n}\right),\qquad\textrm{where }\quad\boldsymbol{\phi}\left(t\right)=\left(t,\boldsymbol{X}\left(t\right)\right),\boldsymbol{p}\in C\left(\mathcal{D}_{\boldsymbol{\phi}},\mathcal{Z}\right),\boldsymbol{f}_{n}\in\mathcal{N}_{n}.\label{eq: NN as control}
\end{eqnarray}
We identify the NN $\boldsymbol{f}_{n}\left(\cdot;\boldsymbol{\theta}_{n}\right)$
with its parameter vector $\boldsymbol{\theta}_{n}$, so that the
(approximate) objective functional using approximation (\ref{eq: NN as control})
is written as
\begin{eqnarray}
J_{n}\left(\boldsymbol{\theta}_{n},\xi;t_{0},w_{0}\right) & = & E^{t_{0},w_{0}}\left[F\left(W\left(T;\boldsymbol{\theta}_{n},\boldsymbol{Y}\right),\xi\right)+G\left(W\left(T;\boldsymbol{\theta}_{n},\boldsymbol{Y}\right),E^{t_{0},w_{0}}\left[W\left(T;\boldsymbol{\theta}_{n},\boldsymbol{Y}\right)\right],w_{0},\xi\right)\begin{array}{c}
\!\!\!\!\!\!\!\!\\
\!\!\!\!\!\!\!\!
\end{array}\right].\label{eq: J_n}
\end{eqnarray}
Combining (\ref{eq: NN f_n as a vector}) and (\ref{eq: NN as control}),
the wealth dynamics (\ref{eq: W dynamics - continuous control}) is
expressed as
\begin{eqnarray}
W\left(t_{m+1}^{-};\boldsymbol{\theta}_{n},\boldsymbol{Y}\right) & = & \left[W\left(t_{m}^{-};\boldsymbol{\theta}_{n},\boldsymbol{Y}\right)+q\left(t_{m}\right)\right]\cdot\sum_{i=1}^{N_{a}}f_{n,i}\left(\boldsymbol{\phi}\left(t_{m}\right);\boldsymbol{\theta}_{n}\right)\cdot Y_{i}\left(t_{m}\right),\quad m=0,...,N_{rb}-1.\label{eq: W dynamics with NN as control}
\end{eqnarray}
Using (\ref{eq: NN as control}) and (\ref{eq: J_n}), for fixed and
given values of $\mathcal{L}^{h}$ and $\hbar\left(n\right)$, we
therefore approximate problem (\ref{eq: Original problem - CONTINUOUS control})
by
\begin{eqnarray}
V_{n}\left(t_{0},w_{0}\right) & = & \inf_{\xi\in\mathbb{R}}\inf_{\boldsymbol{f}_{n}\left(\cdot;\boldsymbol{\theta}_{n}\right)\in\mathcal{N}_{n}}J_{n}\left(\boldsymbol{\theta}_{n},\xi;t_{0},w_{0}\right)\label{eq: V_n in terms of N_n}\\
 & = & \inf_{\xi\in\mathbb{R}}\inf_{\boldsymbol{\theta}_{n}\in\mathbb{R}^{\nu_{n}}}J_{n}\left(\boldsymbol{\theta}_{n},\xi;t_{0},w_{0}\right)\nonumber \\
 & = & \inf_{\left(\boldsymbol{\theta}_{n},\xi\right)\in\mathbb{R}^{\nu_{n}+1}}J_{n}\left(\boldsymbol{\theta}_{n},\xi;t_{0},w_{0}\right).\label{eq: V_n in terms of theta_n}
\end{eqnarray}

We highlight that the optimization in (\ref{eq: V_n in terms of theta_n})
is unconstrained since, by construction, each NN $\boldsymbol{f}_{n}\left(\cdot;\boldsymbol{\theta}_{n}\right)\in\mathcal{N}_{n}$
always generates outputs in $\mathcal{Z}$. 

The notation $\left(\boldsymbol{\theta}_{n}^{\ast},\xi^{\ast}\right)$
and the associated NN $\boldsymbol{f}_{n}^{\ast}\left(\cdot;\boldsymbol{\theta}_{n}^{\ast}\right)\in\mathcal{N}_{n}$
are subsequently used to denote the values achieving the optimum in
(\ref{eq: V_n in terms of theta_n}) for given values of $\mathcal{L}^{h}$
and $\hbar\left(n\right)$. Note however that we do \textit{not} assume
that the optimal control $\boldsymbol{p}^{\ast}\in C\left(\mathcal{D}_{\boldsymbol{\phi}},\mathcal{Z}\right)$
satisfying Assumption \ref{assu: Existence and Continuity of the control}
is also a NN in $\mathcal{N}_{n}$, since by the universal approximation
results (see for example \cite{HornikEtAl1989}), we would expect
that the error in approximating (\ref{eq: Original problem - CONTINUOUS control})
by (\ref{eq: V_n in terms of theta_n}) can be made arbitrarily small
for sufficiently large $\hbar\left(n\right)$. These claims are rigorously
confirmed in Section \ref{sec:Convergence-analysis} below, where
we consider a sequence of NNs $\boldsymbol{f}_{n}\left(\cdot;\boldsymbol{\theta}_{n}\right)\in\mathcal{N}_{n}$
obtained by letting $\hbar\left(n\right)\rightarrow\infty$ as $n\rightarrow\infty$
(for any fixed value of $\mathcal{L}^{h}\geq1$). 

\subsection{Step 2 : Computational estimate of the optimal control\label{subsec:Step 2 Computational estimate}}

In order to solve the approximation (\ref{eq: V_n in terms of theta_n})
to problem (\ref{eq: Original problem - CONTINUOUS control}), we
require estimates of the expectations in (\ref{eq: J_n}). For computational
purposes, suppose we take as given a set $\mathcal{Y}_{n}\in\mathbb{R}^{n\times N_{a}\times N_{rb}}$,
consisting of $n\in\mathbb{N}$ independent realizations of the paths
of joint asset returns $\boldsymbol{Y}$, 
\begin{eqnarray}
\mathcal{Y}_{n} & = & \left\{ \boldsymbol{Y}^{\left(j\right)}:j=1,...,n\right\} .\label{eq: Y training}
\end{eqnarray}
We highlight that each entry $\boldsymbol{Y}^{\left(j\right)}\in\mathcal{Y}_{n}$
consists of a \textit{path} of joint asset returns (see (\ref{eq: Path Y of asset returns})),
and we assume that the paths are independent, we do \textit{not} assume
that the asset returns constituting each path are independent. In
particular, both cross-correlations and autocorrelation structures
within each path of returns are permitted.

Constructing the set $\mathcal{Y}_{n}$ in practical applications
is further discussed in Appendix \ref{sec: Appendix B - NN approach - practical considerations}.
In the numerical examples in Section \ref{sec:Numerical results},
we use examples where $\mathcal{Y}_{n}$ is generated using (i) Monte
Carlo simulation of parametric asset dynamics, (ii) stationary block
bootstrap resampling of empirical asset returns, (\cite{Cederburg_2022})  and (iii) generative
adversarial network (GAN)-generated synthetic asset returns (\cite{YoonTimeGAN2019}). While
we let $n\rightarrow\infty$ in (\ref{eq: Y training}) for convergence
analysis purposes, in practical applications (e.g. the results of
Section \ref{sec:Numerical results}) we simply choose $n$ sufficiently
large such that we are reasonably confident that reliable numerical
estimates of the expectations in (\ref{eq: J_n}) are obtained. 

Given a NN $\boldsymbol{f}_{n}\left(\cdot;\boldsymbol{\theta}_{n}\right)\in\mathcal{N}_{n}$
and set $\mathcal{Y}_{n}$, the wealth dynamics (\ref{eq: W dynamics with NN as control})
along path $\boldsymbol{Y}^{\left(j\right)}\in\mathcal{Y}_{n}$ is
given by
\begin{eqnarray}
W^{\left(j\right)}\left(t_{m+1}^{-};\boldsymbol{\theta}_{n},\mathcal{Y}_{n}\right) & = & \left[W^{\left(j\right)}\left(t_{m}^{-};\boldsymbol{\theta}_{n},\mathcal{Y}_{n}\right)+q\left(t_{m}\right)\right]\cdot\sum_{i=1}^{N_{a}}f_{n,i}\left(\boldsymbol{\phi}^{\left(j\right)}\left(t_{m}\right);\boldsymbol{\theta}_{n}\right)\cdot Y_{i}^{\left(j\right)}\left(t_{m}\right),\label{eq: W dynamics along path j}
\end{eqnarray}
for $m=0,...,N_{rb}-1$. We introduce the superscript $\left(j\right)$
to emphasize that the quantities are obtained along the $j$th entry
of (\ref{eq: Y training}). 

The computational estimate of $J_{n}\left(\boldsymbol{\theta}_{n},\xi;t_{0},w_{0}\right)$
in (\ref{eq: J_n}) is then given by 
\begin{eqnarray}
\hat{J}_{n}\left(\boldsymbol{\theta}_{n},\xi;t_{0},w_{0},\mathcal{Y}_{n}\right) & = & \frac{1}{n}\sum_{j=1}^{n}F\left(W^{\left(j\right)}\left(T;\boldsymbol{\theta}_{n},\mathcal{Y}_{n}\right),\xi\right)\nonumber \\
 &  & +\frac{1}{n}\sum_{j=1}^{n}G\left(W^{\left(j\right)}\left(T;\boldsymbol{\theta}_{n},\mathcal{Y}_{n}\right),\frac{1}{n}\sum_{k=1}^{n}W^{\left(k\right)}\left(T;\boldsymbol{\theta}_{n},\mathcal{Y}_{n}\right),w_{0},\xi\right),\label{eq: J_n_HAT approx}
\end{eqnarray}
so that we approximate problem (\ref{eq: V_n in terms of theta_n})
by 
\begin{eqnarray}
\hat{V}_{n}\left(t_{0},w_{0};\mathcal{Y}_{n}\right) & = & \inf_{\left(\boldsymbol{\theta}_{n},\xi\right)\in\mathbb{R}^{\nu_{n}+1}}\hat{J}_{n}\left(\boldsymbol{\theta}_{n},\xi;t_{0},w_{0},\mathcal{Y}_{n}\right).\label{eq: V_n_HAT approx}
\end{eqnarray}
The numerical solution of (\ref{eq: V_n_HAT approx}) can then proceed
using standard (stochastic) gradient descent techniques (see Appendix
\ref{sec: Appendix B - NN approach - practical considerations}).
For subsequent reference, let $\left(\hat{\boldsymbol{\theta}}_{n}^{\ast},\hat{\xi}_{n}^{\ast}\right)$
denote the optimal point in (\ref{eq: V_n_HAT approx}) relative to
the training data set $\mathcal{Y}_{n}$ in (\ref{eq: V_n_HAT approx}).

In the case of sufficiently large datasets (\ref{eq: Y training}),
in other words as $n\rightarrow\infty$, we would expect that the
error in approximating (\ref{eq: V_n in terms of theta_n}) by (\ref{eq: V_n_HAT approx})
can be made arbitrarily small. However, as noted above, as $n\rightarrow\infty$
and the number of hidden nodes $\hbar\left(n\right)\rightarrow\infty$
(for any fixed $\mathcal{L}^{h}\geq1$), (\ref{eq: V_n in terms of theta_n})
is also expected to approximate (\ref{eq: Original problem - CONTINUOUS control})
more accurately. As a result, we obtain the necessary intuition for
establishing the convergence of (\ref{eq: V_n_HAT approx}) to (\ref{eq: Original problem - CONTINUOUS control})
under suitable conditions, which is indeed confirmed in the results
of Section \ref{sec:Convergence-analysis}.

Note that since $\mathcal{Y}_{n}$ is used in (\ref{eq: V_n_HAT approx})
to obtain the optimal NN parameter vector $\hat{\boldsymbol{\theta}}_{n}^{\ast}$,
it is usually referred to as the NN ``training'' dataset (see for
example \cite{GoodfellowEtAl_BOOK}). Naturally, we can also construct
a ``testing'' dataset $\mathcal{Y}_{\hat{n}}^{test}$, that is of
a similar structure as (\ref{eq: Y training}), but typically based
on a different implied distribution of $\boldsymbol{Y}$ as a result
of different data generation assumptions. For example, $\mathcal{Y}_{\hat{n}}^{test}$
can be obtained using a different time period of historical data for
its construction, or different process parameters if there are parametric
asset dynamics specified. The resulting approximation $\boldsymbol{f}_{n}^{\ast}\left(\cdot;\hat{\boldsymbol{\theta}}_{n}^{\ast}\right)\in\mathcal{N}_{n}$
to the optimal control $\boldsymbol{p}^{\ast}\in C\left(\mathcal{D}_{\boldsymbol{\phi}},\mathcal{Z}\right)$
obtained using the training dataset in (\ref{eq: V_n_HAT approx})
can then be implemented on the testing dataset for out-of-sample testing
or scenario analysis. This is discussed in more detail in Appendix
\ref{sec: Appendix B - NN approach - practical considerations}. 

\noindent \begin{brem}\label{rem: Wealth path-dependent objectives}(Extension
to wealth path-dependent objectives) As noted in the Introduction,
the NN approach as well as the convergence analysis of Section \ref{sec:Convergence-analysis}
can be extended to objective functions that depend on the entire wealth
path $\left\{ W\left(t\right):t\in\mathcal{T}\right\} $ instead of
just the terminal wealth $W\left(T\right)$. This is achieved by simply
modifying (\ref{eq: J_n_HAT approx}) appropriately and ensuring the
wealth is assessed at the desired intervals using (\ref{eq: W dynamics along path j}).
\qed\end{brem}

\subsection{Advantages of the NN approach\label{subsec:Advantages-of-this approach}}

The following observations highlight some advantages of the proposed
NN approach:

\begin{enumerate}[label=(\roman*)]

\item The approach does not rely on dynamic programming (DP) methods
for the solution of problem (\ref{eq: V_n_HAT approx}), and therefore
does not require value iteration or backward time stepping. In particular,
we observe that due to the explicit time-dependence of the NN feature
vector, the optimization problem (\ref{eq: V_n_HAT approx}) itself
only indirectly depends on the number of rebalancing events, while
time recursion is limited to the (computationally inexpensive) wealth
dynamics (\ref{eq: W dynamics along path j}). As result, problems
relating to the error amplification associated with DP methods (\cite{WangFoster2020,TsangWong2020,LiEtAl2020})
are avoided, and only a single optimization problem that is independent
of the number of portfolio rebalancing events is solved, in contrast
to DP-based methods (see for example \cite{BachouchEtAl2021,VanHeeswijkPoutre2019}). 

Not relying on DP techniques also makes the approach significantly
more flexible, in that it can directly handle objective functions
that are not separable in the sense of DP, without requiring theoretical
results such as embedding in the case of MV optimization (see for
example \cite{ZhouLi2000,LiNg2000}). As an example of this, we present
the solution of the mean - semi-variance problem (\ref{eq: Sortino})
in Section \ref{sec:Numerical results}. 

\item The proposed methodology is parsimonious, in the sense that
the NN parameter vector remains independent of number of rebalancing
events. Specifically, we observe that the NN parameter vector $\boldsymbol{\theta}_{n}\in\mathbb{R}^{\nu_{n}}$
of the NN does \textit{not} depend on the rebalancing time $t_{m}\in\mathcal{T}$
or on the sample path $j$. This contrasts our approach with the approaches
of for example \cite{HanWeinan2016,TsangWong2020},\footnote{
{\color{black}
\cite{TsangWong2020} use a stacked NN approach, with a different NN at each
rebalancing time. 
}
}
where the number
of parameters scale with the number of rebalancing events. As a result,
the NN approach presented here can lead to potentially significant
computational advantages in the cases of (i) long investment time
horizons or (ii) short trading time horizons with a frequent number
of portfolio rebalancing events. 

A natural question might be whether the NNs in the proposed approach
are required to be very deep, thus potentially exposing the training
of the NN in (\ref{eq: V_n_HAT approx}) to problem of vanishing or
exploding gradients (see for example \cite{GoodfellowEtAl_BOOK}).
However, the ground truth results presented in Section \ref{sec:Numerical results}
demonstrate that we obtain very accurate results with relatively shallow
NNs (at most two hidden layers). We suspect this might be due to the
optimal control being relatively low-dimensional compared to the high-dimensional
objective functionals in portfolio optimization problems with discrete
rebalancing (see \cite{PvSForsythLi2022_stochbm} for a rigorous analysis),
while in this NN approach approach the optimal control is obtained
directly without requiring the solution of the (high-dimensional)
objective functional at rebalancing times.

\end{enumerate}

Note that these advantages also contrast the NN approach with Reinforcement
Learning-based algorithms to solve portfolio optimization problems,
as the following remark discusses.

\noindent \begin{brem}\label{rem: Contrast to RL}(Contrast of NN
approach to Reinforcement Learning). Reinforcement learning (RL) algorithms
(for example, Q-learning) relies fundamentally on the DP principle
for the numerical solution of the portfolio optimization problem (see
for example \cite{GaoEtAL2020,LucarelliEtAl2020,ParkEtAl2020}). This
requires, at each value iteration step, the approximation of a (high-dimensional)
conditional expectation. As a result, RL is associated with standard
DP-related concerns related to error amplification and the curse of
dimensionality discussed above, and also cannot solve general problems
of the form (\ref{eq: Initial full objective}) without relying on
for example an embedding approach to obtain an associated problem
that can be solved using DP methods.\qed\end{brem}

\section{Convergence analysis\label{sec:Convergence-analysis}}

In this section, we present the theoretical justification of the proposed
NN approach as outlined in Section \ref{sec: NN approach - overview}.
We confirm that the numerical solution of (\ref{eq: V_n_HAT approx})
can be used to approximate the theoretical solution of (\ref{eq: Original problem - CONTINUOUS control})
arbitrarily well (in probability) under suitable conditions. This
section only summarizes the key convergence results which are among
the main contributions of this paper, while additional technical details
and proofs are provided in Appendix \ref{sec:Appendix A:-Technical details and analytical results}.

We start with Theorem \ref{thm: Validity of NN approximation}, which confirms
the validity of Step 1 (Subsection \ref{subsec:Step 1 NN approx}),
namely using a NN $\boldsymbol{f}_{n}\left(\cdot;\boldsymbol{\theta}_{n}\right)\in\mathcal{N}_{n}$
to approximate the control. Note that Theorem \ref{thm: Validity of NN approximation}
relies on two assumptions, presented in Appendix \ref{subsec: Appendix assumptions for convergence}:
We emphasize that Assumption \ref{assu: Convergence assumptions - convenience}
is purely made for purposes of convenience, since its requirements
can easily be relaxed with only minor modifications to the proofs
(as discussed in Remark \ref{rem: Appendix Relaxing Convergence assumptions - convenience}),
but at the cost of significant notational complexity and no additional
insights. In contrast, Assumption \ref{assu: Appendix Convergence assumptions - critical}
is critical to establish the result of Theorem \ref{thm: Validity of NN approximation},
and requires that the optimal investment strategy (or control) satisfies
Assumption \ref{assu: Existence and Continuity of the control}, places
some basic requirements on $F$ and $G$, and assumes that the sequence
of NNs $\left\{ \boldsymbol{f}_{n}\left(\cdot;\boldsymbol{\theta}_{n}\right),n\in\mathbb{N}\right\} $
is constructed such that the number of nodes in each hidden layer
$\hbar\left(n\right)\rightarrow\infty$ as $n\rightarrow\infty$ (no
assumptions are yet required regarding the exact form of $n\rightarrow\hbar\left(n\right)$). 
\begin{thm}
\label{thm: Validity of NN approximation}(Validity of NN approximation)
We assume that Assumption \ref{assu: Appendix Convergence assumptions - critical}
holds, and for ease of exposition, we also assume that Assumption
\ref{assu: Convergence assumptions - convenience} holds. Then the
NN approximation to the control in (\ref{eq: NN as control}) is valid,
in the sense that $V\left(t_{0},w_{0}\right)$ in (\ref{eq: Original problem - CONTINUOUS control})
can be approximated arbitrarily well by $V_{n}\left(t_{0},w_{0}\right)$
in (\ref{eq: V_n in terms of theta_n}) for sufficiently large $n$,
since
\begin{eqnarray}
\lim_{n\rightarrow\infty}\left|V_{n}\left(t_{0},w_{0}\right)-V\left(t_{0},w_{0}\right)\right| & = & \lim_{n\rightarrow\infty}\left|\inf_{\left(\boldsymbol{\theta}_{n},\xi\right)\in\mathbb{R}^{\nu_{n}+1}}J_{n}\left(\boldsymbol{\theta}_{n},\xi;t_{0},w_{0}\right)-\inf_{\xi\in\mathbb{R}}\inf_{\boldsymbol{p}\in C\left(\mathcal{D}_{\boldsymbol{\phi}},\mathcal{Z}\right)}J\left(\boldsymbol{p},\xi;t_{0},w_{0}\right)\right|\nonumber \\
 & = & 0.\label{eq: Result thm validity of NN approx}
\end{eqnarray}
\end{thm}

\begin{proof}
See Appendix \ref{subsec:Appendix Proof-of-Theorem NN approx}.
\end{proof}
Having justified Step 1 of the approach, Theorem \ref{thm: Validity of NN computational estimate}
now confirms the validity of Step 2 of the NN approach (see Subsection
\ref{subsec:Step 2 Computational estimate}), namely using the computational
estimate $\boldsymbol{f}_{n}^{\ast}\left(\cdot;\hat{\boldsymbol{\theta}}_{n}^{\ast}\right)\in\mathcal{N}_{n}$
from (\ref{eq: V_n_HAT approx}) as an approximation of the true optimal
control $\boldsymbol{p}^{\ast}\in C\left(\mathcal{D}_{\boldsymbol{\phi}},\mathcal{Z}\right)$.
Note that in addition to the assumptions of Theorem \ref{thm: Validity of NN approximation},
Theorem \ref{thm: Validity of NN computational estimate} also requires
Assumption \ref{assu: Appendix Convergence assumptions - computational},
which by necessity includes computational considerations such as the
structure of the training dataset $\mathcal{Y}_{n}$, the rate of
divergence of the number of hidden nodes $\hbar\left(n\right)\rightarrow\infty$
as $n\rightarrow\infty$, and assumptions regarding the optimization
algorithm used in solving problem (\ref{eq: V_n_HAT approx}). 
\begin{thm}
\label{thm: Validity of NN computational estimate}(Validity of computational
estimate) We assume that Assumption \ref{assu: Appendix Convergence assumptions - critical},
Assumption \ref{assu: Convergence assumptions - convenience} and
Assumption \ref{assu: Appendix Convergence assumptions - computational}
hold. Then the computational estimate to the optimal control (\ref{eq: Assumption CONTINUOUS control})
obtained using (\ref{eq: NN as control}) and (\ref{eq: V_n_HAT approx})
is valid, in the sense that the value function $V\left(t_{0},w_{0}\right)$
in (\ref{eq: Original problem - CONTINUOUS control}) can be approximated
arbitrarily well in probability by $\hat{V}_{n}\left(t_{0},w_{0};\mathcal{Y}_{n}\right)$
in (\ref{eq: V_n_HAT approx}) for sufficiently large $n$, since 
\end{thm}

\begin{eqnarray}
\left|\hat{V}_{n}\left(t_{0},w_{0};\mathcal{Y}_{n}\right)-V\left(t_{0},w_{0}\right)\right| & = & \left|\inf_{\left(\boldsymbol{\theta}_{n},\xi\right)\in\mathbb{R}^{\eta_{n}+1}}\hat{J}_{n}\left(\boldsymbol{\theta}_{n},\xi;t_{0},w_{0},\mathcal{Y}_{n}\right)-\inf_{\xi\in\mathbb{R}}\inf_{\boldsymbol{p}\in C\left(\mathcal{D}_{\boldsymbol{\phi}},\mathcal{Z}\right)}J\left(\boldsymbol{p},\xi;t_{0},w_{0}\right)\right|\nonumber \\
 & \overset{P}{\longrightarrow} & 0,\qquad\textrm{as }n\rightarrow\infty.\label{eq: Result thm validity of NN computational estimate}
\end{eqnarray}

\begin{proof}
See Appendix \ref{subsec:Appendix Proof-of-Theorem NN approx}.
\end{proof}
Taken together, Theorem \ref{thm: Validity of NN approximation} and
Theorem \ref{thm: Validity of NN computational estimate} establish
the theoretical validity of the NN approach to solve problems of the
form (\ref{eq: Initial full objective}).

\section{Numerical results\label{sec:Numerical results}}

In this section, we present numerical results obtained by implementing
the NN approach described in Section \ref{sec: NN approach - overview}.
For illustrative purposes, the examples focus on investment objectives
as outlined in Subsection \ref{subsec:Investment-objectives}, and
we use three different data generation techniques for obtaining the
training data set $\mathcal{Y}_{n}$ of the NN: (i) parametric models
for underlying asset returns, (ii) stationary block bootstrap resampling
of empirical returns, and (ii) generative adversarial network (GAN)-generated
synthetic asset returns.

\subsection{Closed-form solution: $DSQ\left(\gamma\right)$ with continuous rebalancing\label{subsec:Ground-truth: DSQ with cont rebal}}

Under certain conditions, some of the optimization problems in Subsection
\ref{subsec:Investment-objectives} can be solved analytically. In
this subsection, we demonstrate how a closed-form solution of problem
$DSQ\left(\gamma\right)$ in (\ref{eq: DSQ objective}), assuming
\textit{continuous} rebalancing (i.e. if we let $\Delta t\rightarrow0$
in (\ref{eq: Set of rebalancing times})), can be approximated very
accurately using a very simple NN (1 hidden layer, only 3 hidden nodes)
using discrete rebalancing with $\Delta t\gg0$ in (\ref{eq: Set of rebalancing times}).
This simultaneously illustrates how parsimonious the NN approach is,
as well as how useful the imposition of time-continuity is in ensuring
the smooth behavior of the (approximate) optimal control. 

In this subsection as well as in Subsection \ref{subsec: Ground truth problem MCV},
we assume parametric dynamics for the underlying assets. For concreteness,
we consider the scenario of two assets, $N_{a}=2$, with unit values
$S_{i},i=1,2$, evolving according to the following dynamics, 
\begin{eqnarray}
\frac{dS_{i}\left(t\right)}{S_{i}\left(t^{-}\right)} & = & \left(\mu_{i}-\lambda_{i}\kappa_{i}^{\left(1\right)}\right)\cdot dt+\sigma_{i}\cdot dZ_{i}\left(t\right)+d\left(\sum_{k=1}^{\pi_{i}\left(t\right)}\left(\vartheta_{i}^{\left(k\right)}-1\right)\right),\qquad i=1,2.\label{eq: Parametric dynamics of underlying assets}
\end{eqnarray}

Note that (\ref{eq: Parametric dynamics of underlying assets}) takes
the form of the standard jump diffusion models in finance - see e.g.
\cite{MertonJumps1976,KouOriginal} for more information. For each
asset $i$ in (\ref{eq: Parametric dynamics of underlying assets}),
$\mu_{i}$ and $\sigma_{i}$ denote the (actual, not risk-neutral)
drift and volatility, respectively, $Z_{i}$ denotes a standard Brownian
motion, $\pi_{i}\left(t\right)$ denotes a Poisson process with intensity
$\lambda_{i}\geq0$, and $\vartheta_{i}^{\left(k\right)}$ are i.i.d.
random variables with the same distribution as $\vartheta_{i}$, which
represents the jump multiplier of the $i$th risky asset with $\kappa_{i}^{\left(1\right)}=\mathbb{E}\left[\vartheta_{i}-1\right]$
and $\kappa_{i}^{\left(2\right)}=\mathbb{E}\left[\left(\vartheta_{i}-1\right)^{2}\right]$.
While the Brownian motions can be correlated with $dZ_{1}\left(t\right)dZ_{2}\left(t\right)=\rho_{1,2}\cdot dt$,
we make the standard assumption that the jump components are independent
(see for example \cite{ForsythVetzal2022CutLosses}).

For this subsection only, we treat the first asset ($i=1$ in (\ref{eq: Parametric dynamics of underlying assets}))
as a ``risk-free'' asset, and set $\mu_{1}=r>0$ where $r$ is the
risk-free rate, so that we have $\lambda_{1}=0$, $\sigma_{1j}=0$
$\forall j$, and $Z_{1}\equiv0$, while the second asset ($i=2$
in (\ref{eq: Parametric dynamics of underlying assets})) is assumed
to be a broad equity market index (the ``risky asset''). In this
scenario, if problem $DSQ\left(\gamma\right)$ in (\ref{eq: DSQ objective})
is solved subject to dynamics (\ref{eq: Parametric dynamics of underlying assets})
together with the assumptions of costless continuous trading, infinite
leverage, and uninterrupted trading in the event of insolvency, then
the $DSQ\left(\gamma\right)$-optimal control can be obtained analytically
as
\begin{align}
\boldsymbol{p}^{\ast}\left(t,W^{\ast}\left(t\right)\right) & =\left[1-p_{2}^{\ast}\left(t,W^{\ast}\left(t\right)\right),p_{2}^{\ast}\left(t,W^{\ast}\left(t\right)\right)\right]\in\mathbb{R}^{2},\label{eq: Full DSQ optimal control}
\end{align}
where the fraction of wealth in the broad stock market index (asset
$i=2$) is given by (\cite{ZwengLi2011})
\begin{align}
p_{2}^{\ast}\left(t,W^{\ast}\left(t\right)\right) & =\frac{\mu_{2}-r}{\sigma_{2}^{2}+\lambda_{2}\kappa_{2}^{\left(2\right)}}\cdot\left[\frac{\gamma e^{-r\left(T-t\right)}-W^{\ast}\left(t\right)}{W^{\ast}\left(t\right)}\right],\qquad w_{0}<\gamma e^{-r\left(T-t\right)}.\label{eq: Analytical solution DSQ}
\end{align}

By design, the NN approach is not constructed to solve problems with
unrealistic assumptions such as continuous trading, infinite leverage
and short-selling, or trading in the event of bankruptcy, all of which
are required to derive (\ref{eq: Analytical solution DSQ}). However,
if the implicit quadratic wealth target for the DSQ problem (i.e.
the value of $\gamma$, see \cite{Vigna_efficiency2014}) is not too
aggressive, the analytical solution (\ref{eq: Analytical solution DSQ})
does not require significant leverage or lead to a large probability
of insolvency. In such a scenario, we can use the NN approach to approximate
(\ref{eq: Analytical solution DSQ}).

We select $w_{0}=100$, $T=1$ year and $\gamma=138.33$,
and simulate $n=2.56\times10^{6}$ paths of the underlying assets
using (\ref{eq: Parametric dynamics of underlying assets}) and parameters
as in Table \ref{tab: Params for ground truth DSQ with cont rebal}
(Appendix \ref{sec: Appendix Parameters-for-numerical results}).
On this set of paths, the true analytical solution (\ref{eq: Analytical solution DSQ})
is implemented using 7,200 time steps. In contrast, for the NN approach,
we use only 4 rebalancing events in $\left[0,T=1\right]$, and therefore
aggregate the simulated returns in quarterly time intervals to construct
the training data set $\mathcal{Y}_{n}$. We consider only a very
shallow NN, consisting of a single hidden layer and only 3 hidden
nodes. 

Figure \ref{fig: Num_04_Heatmaps DSQ only} compares the resulting
optimal investment strategies by illustrating the optimal proportion
of wealth invested in the the broad equity market index (asset $i=2$)
as a function of time and wealth. We emphasize that the NN strategy
in Figure \ref{fig: Num_04_Heatmaps DSQ only}(b) is not expected
to be exactly identical to the analytical solution in Figure \ref{fig: Num_04_Heatmaps DSQ only}(a),
since it is based on fundamentally different assumptions such as discrete
rebalancing and investment constraints (\ref{eq: Set A}). 

\noindent 
\begin{figure}[!tbh]
\noindent \begin{centering}
\subfloat[Closed-form solution, $\Delta t\rightarrow0$ ($w_{0}\in\mathbb{R}$)]{\includegraphics[scale=0.6]{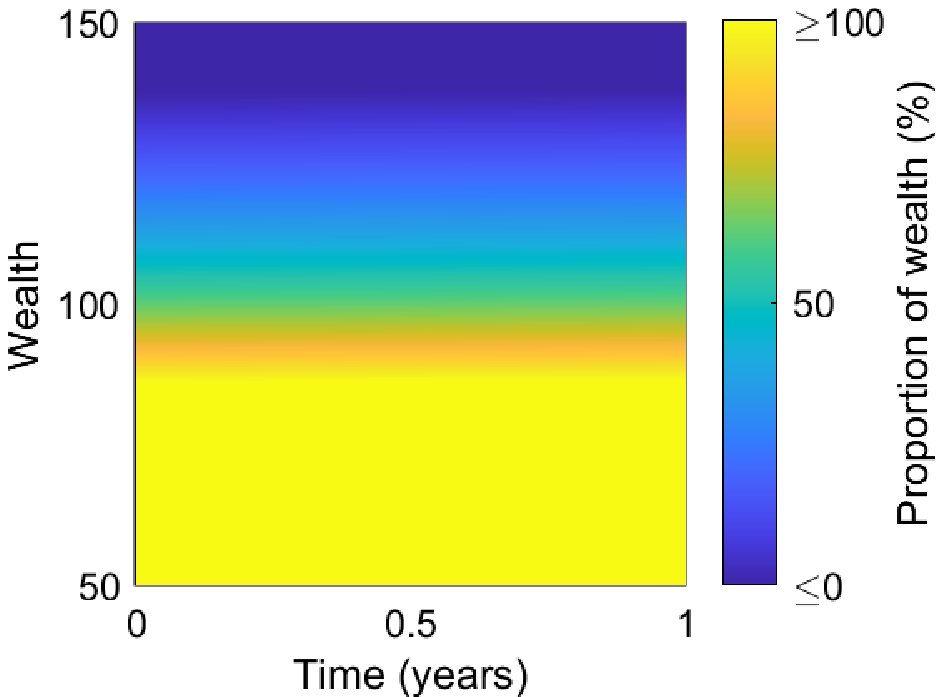}

}$\qquad$\subfloat[Shallow NN, $\Delta t=0.25$ ($w_{0}=100$)]{\includegraphics[scale=0.6]{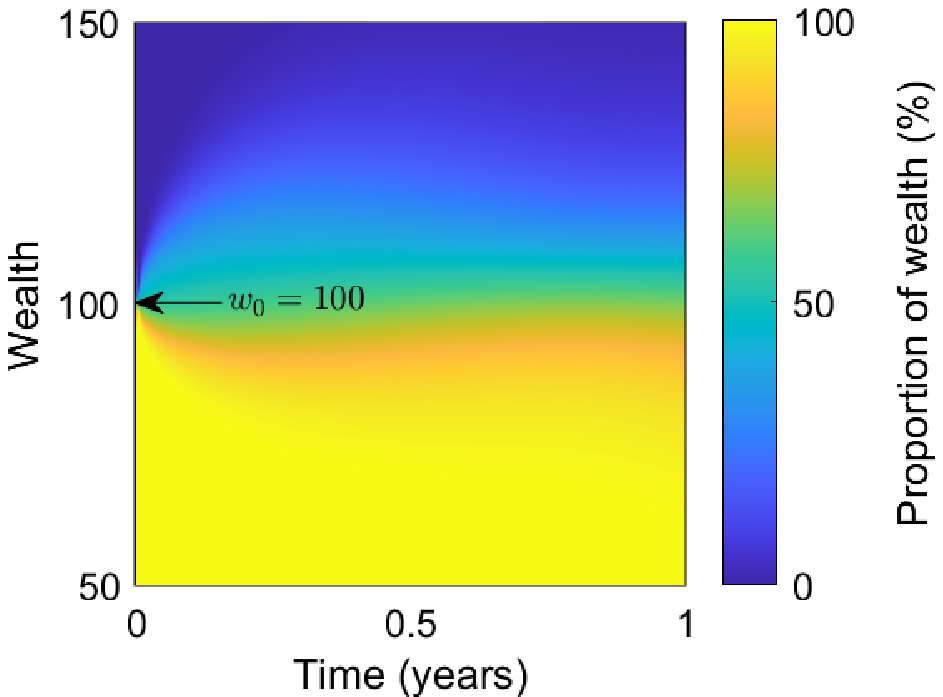}

}
\par\end{centering}
\caption{Closed-form solution - $DSQ\left(\gamma\right)$ with continuous rebalancing:
Optimal proportion of wealth invested in the broad equity market index
as a function of time and wealth. The NN approximation is obtained
for a specific initial wealth of $w_{0}=$100, and only four rebalancing
events in $\left[0,T\right]$. \label{fig: Num_04_Heatmaps DSQ only}}
\end{figure}

However, requiring that the NN feature vector includes time in the
proposed NN approach, together with a NN parameter vector that does
not depend on time, we guarantee the smooth behavior in time of the
NN approximation observed in Figure \ref{fig: Num_04_Heatmaps DSQ only}(b).
As a result, Table \ref{tab:Ground-truth--DSQ analytical W percentiles}
shows that the shallow NN strategy trained with $\Delta t\gg0$ results
in a remarkably accurate and parsimonious approximation to the true
analytical solution where $\Delta t\rightarrow0$, since we obtain
nearly identical optimal terminal wealth distributions.

\noindent 
\begin{table}
\caption{Closed-form solution - $DSQ\left(\gamma\right)$ with continuous rebalancing:
Percentiles of the simulated ( $n=2.56\times10^{6}$) terminal wealth
distributions obtained by implementing the optimal strategies in Figure
\ref{fig: Num_04_Heatmaps DSQ only}. In both cases, a mean terminal
wealth of 105 is obtained. Note that the NN approximation was obtained
under the assumption of quarterly rebalancing only, no leverage or
short-selling, and therefore no trading in insolvency.\label{tab:Ground-truth--DSQ analytical W percentiles}}

\noindent \centering{}%
\begin{tabular}{|c|c|>{\centering}p{1cm}|>{\centering}p{1cm}|>{\centering}p{1cm}|>{\centering}p{1cm}|>{\centering}p{1cm}|}
\cline{3-7} \cline{4-7} \cline{5-7} \cline{6-7} \cline{7-7} 
\multicolumn{1}{c}{} &  & \multicolumn{5}{c|}{{\footnotesize{}$W\left(T\right)$ percentiles}}\tabularnewline
\hline 
{\footnotesize{}Solution approach} & {\footnotesize{}Rebalancing} & {\footnotesize{}5th} & {\footnotesize{}20th} & {\footnotesize{}50th} & {\footnotesize{}80th} & {\footnotesize{}95th}\tabularnewline
\hline 
\hline 
{\footnotesize{}Closed-form solution} & {\footnotesize{}Continuous, $\Delta t$$\rightarrow0$} & {\footnotesize{}86.81} & {\footnotesize{}98.02} & {\footnotesize{}106.35} & {\footnotesize{}112.82} & {\footnotesize{}118.15}\tabularnewline
\hline 
{\footnotesize{}Shallow NN approximation} & {\footnotesize{}Discrete, $\Delta t=0.25$, total of $N_{rb}=4$ only} & {\footnotesize{}86.62} & {\footnotesize{}97.30} & {\footnotesize{}105.67} & {\footnotesize{}112.54} & {\footnotesize{}118.85}\tabularnewline
\hline 
\end{tabular}
\end{table}

\subsection{Ground truth: Problem $MCV\left(\rho\right)$\label{subsec: Ground truth problem MCV}}

In the case of the Mean-CVaR problem $MCV\left(\rho\right)$ in (\ref{eq: MCV objective}),
\cite{ForsythVetzal2022CutLosses} obtain an MCV-optimal investment
strategy subject to the same investment constraints as in Section
\ref{sec:Problem-formulation} (namely discrete rebalancing, no short-selling
or leverage allowed, and no trading in insolvency) using the partial
(integro-)differential equation (PDE) approach of \cite{Forsyth2019CVaR}. 

For ground truth analysis purposes, we therefore consider the same
investment scenario as in \cite{ForsythVetzal2022CutLosses}, where
two underlying assets are considered, namely 30-day US T-bills and
a broad equity market index (the CRSP VWD index) - see Appendix \ref{sec: Appendix Parameters-for-numerical results}
for definitions. However, in contrast to the preceding section where
one asset was taken as the risk-free asset, both assets are now assumed
to evolve according to dynamics of the form (\ref{eq: Parametric dynamics of underlying assets}),
using the double-exponential \cite{KouOriginal} formulation for the
jump distributions. The NN training data set is therefore constructed
by simulating the same underlying dynamics. While further details
regarding the context and motivation for the investment scenario can
be found in \cite{ForsythVetzal2022CutLosses}, here we simply note
that the scenario involves $T=5$ years, quarterly rebalancing, a
set of admissible strategies satisfying (\ref{eq: Set A}), and parameters
for (\ref{eq: Parametric dynamics of underlying assets}) as in Table
\ref{tab: Params for ground truth MCV}. 

As discussed in Appendix \ref{sec: Appendix B - NN approach - practical considerations},
the inherently higher complexity of the Mean-CVaR optimal control
requires the NN to be deeper than in the case of the problem considered
in Subsection \ref{subsec:Ground-truth: DSQ with cont rebal}. As
a result, we consider approximating NNs with two hidden layers, each
with 8 hidden nodes, while relatively large mini-batches of 2,000
paths were used in the stochastic gradient descent algorithm (see
Appendix \ref{sec: Appendix B - NN approach - practical considerations})
to ensure sufficiently accurate sampling of the tail of the returns
distribution in selecting the descent direction at each step. Note
that despite using a deeper NN, this NN structure is still very parsimonious
and relatively shallow compared to the rebalancing time-dependent
structures considered in for example \cite{HanWeinan2016}, where
a new set of parameters is introduced at each rebalancing event.

Table \ref{tab:MCV ground truth} compares the PDE results reported
in \cite{ForsythVetzal2022CutLosses} with the corresponding NN results.
{\color{black}
Note that the PDE optimal control  was determined by solving
a Hamilton-Jacobi-Bellman PDE numerically.  
The statistics
for the PDE generated control were computed using $n=2.56\times10^{6}$
Monte Carlo simulations 
}
of the joint underlying asset dynamics in
order to calculate the results of Table \ref{tab:MCV ground truth},
while the NN was trained on $n=2.56\times10^{6}$ paths of the same
underlying asset dynamics but which were independently simulated.
While some variability of the results are therefore to be expected
due to the underlying samples, the results in Table \ref{tab:MCV ground truth}
demonstrate the robustness of the proposed NN approach.

\noindent 
\begin{table}[!tbh]
\caption{Ground truth - problem $MCV\left(\rho\right)$: The PDE results are
obtained from \cite{ForsythVetzal2022CutLosses} for selected points
on the Mean-CVaR \textquotedblleft efficient frontier\textquotedblright .
The \textquotedblleft Value function\textquotedblright{} column reports
the value of the objective function (\ref{eq: Mean-CVaR - objective 1})
under the corresponding optimal control, while \textquotedblleft\%
difference\textquotedblright{} reports the percentage difference in the
reported value functions for the NN solution compared to the PDE solution.
\label{tab:MCV ground truth}}

\noindent \centering{}{\footnotesize{}}%
\begin{tabular}{|>{\centering}p{1.5cm}|>{\centering}p{1.5cm}|>{\centering}p{1.5cm}|>{\centering}p{1.5cm}|>{\centering}p{1.5cm}|>{\centering}p{1.5cm}|>{\centering}p{1.5cm}|>{\centering}p{1.5cm}|}
\hline 
{\footnotesize{}$\rho$} & \multicolumn{2}{c|}{{\footnotesize{}5\% CVaR}} 
       & \multicolumn{2}{c|}{{\footnotesize{}$E^{t_{0},w_{0}}\left[W\left(T\right)\right]$}} 
       & \multicolumn{2}{c|}{{\footnotesize{}Value function}} 
       & {\footnotesize{}\% difference}
   \tabularnewline
          \cline{2-7} \cline{3-7} \cline{4-7} \cline{5-7} \cline{6-7} \cline{7-7} 
          & {\footnotesize{}PDE} 
          & {\footnotesize{}NN} 
          & {\footnotesize{}PDE} 
          & {\footnotesize{}NN} 
          & {\footnotesize{}PDE} 
            & {\footnotesize{}NN} & 
     \tabularnewline
\hline 
{\footnotesize{}0.10} & {\footnotesize{}940.60} & {\footnotesize{}940.55} & {\footnotesize{}1069.19} & {\footnotesize{}1062.97} & {\footnotesize{}1047.52} & {\footnotesize{}1046.85} & {\footnotesize{}-0.06\%}\tabularnewline
\hline 
{\footnotesize{}0.25} & {\footnotesize{}936.23} & {\footnotesize{}937.39} & {\footnotesize{}1090.89} & {\footnotesize{}1081.99} & {\footnotesize{}1208.95} & {\footnotesize{}1207.88} & {\footnotesize{}-0.09\%}\tabularnewline
\hline 
{\footnotesize{}1.00} & {\footnotesize{}697.56} & {\footnotesize{}690.11} & {\footnotesize{}1437.73} & {\footnotesize{}1444.16} & {\footnotesize{}2135.29} & {\footnotesize{}2134.27} & {\footnotesize{}-0.05\%}\tabularnewline
\hline 
{\footnotesize{}1.50} & {\footnotesize{}614.92} & {\footnotesize{}611.65} & {\footnotesize{}1508.10} & {\footnotesize{}1510.07} & {\footnotesize{}2877.07} & {\footnotesize{}2876.76} & {\footnotesize{}-0.01\%}\tabularnewline
\hline 
\end{tabular}{\footnotesize\par}
\end{table}

\subsection{Ground truth: Problems $MV\left(\rho\right)$ and $DSQ\left(\gamma\right)$\label{subsec: Ground truth problems MV and DSQ}}

In this subsection, we demonstrate that if the investment objective
(\ref{eq: Initial full objective}) is separable in the sense of dynamic
programming, the correct time-consistent optimal investment strategy
is recovered, otherwise we obtain the correct pre-commitment (time-inconsistent)
investment strategy. 

To demonstrate this, the theoretical embedding result of \cite{LiNg2000,ZhouLi2000},
which establishes the equivalence of problems $MV\left(\rho\right)$
and $DSQ\left(\gamma\right)$ under fairly general conditions, can
be exploited for ground truth analysis purposes as follows. Suppose
we solved problems $MV\left(\rho\right)$ and $DSQ\left(\gamma\right)$
on the same underlying training data set. We remind the reader that
in the proposed NN approach, problem $MV\left(\rho\right)$ can indeed
be solved directly without difficulty, which is not possible in dynamic
programming-based approaches. Then, considering the numerical results,
there should be values of parameters $\rho\equiv\tilde{\rho}$ and
$\gamma\equiv\tilde{\gamma}$ such that the optimal strategy of $MV\left(\rho\equiv\tilde{\rho}\right)$
corresponds exactly to the optimal strategy of $DSQ\left(\gamma\equiv\tilde{\gamma}\right)$,
with a specific relationship holding between $\tilde{\rho}$ and $\tilde{\gamma}$.
The NN approach can therefore enable us to numerically demonstrate
the embedding result of \cite{LiNg2000,ZhouLi2000} in a setting where
the underlying asset dynamics are not explicitly specified and where
multiple investment constraints are present. We start by recalling
the embedding result. 
\begin{prop}
\label{prop: Embedding result}(Embedding result of \cite{LiNg2000,ZhouLi2000})
Fix a value $\tilde{\rho}>0$. If $\mathcal{P}^{\ast}\in\mathcal{A}$
is the optimal control of problem $MV\left(\rho\equiv\tilde{\rho}\right)$
in (\ref{eq: MV objective}), then $\mathcal{P}^{\ast}$ is also the
optimal control for problem $DSQ\left(\gamma=\tilde{\gamma}\right)$
in (\ref{eq: DSQ objective}), provided that 
\begin{eqnarray}
\tilde{\gamma} & = & \frac{1}{2\tilde{\rho}}+E^{t_{0},w_{0}}\left[W^{\ast}\left(T;\mathcal{P}^{\ast},\boldsymbol{Y}\right)\right].\label{eq: Embedding result}
\end{eqnarray}
\end{prop}

\begin{proof}
See \cite{LiNg2000,ZhouLi2000}. We also highlight the alternative
proof provided in \cite{DangForsyth2016}, which shows that this result
is valid for any admissible control set $\mathcal{A}$.
\end{proof}
Since (\ref{eq: Embedding result}) is valid for any admissible control
set $\mathcal{A}$, we consider a factor investing scenario where
portfolios are constructed using popular long-only investable equity
factor indices (Momentum, Value, Low Volatility, Size), a broad equity
market index (the CRSP VWD index), 30-day T-bills and 10-year Treasury
bonds (see Appendix \ref{sec: Appendix Parameters-for-numerical results}
for definitions). For illustrative purposes in the case of an investor
primarily concerned with long-run factor portfolio performance, we
use a horizon of $T=10$ years, $w_{0}=120$, annual contributions
of $q\left(t_{m}\right)=12$, and annual rebalancing.

Given historical returns data for the underlying assets, we construct
training and testing (out-of-sample) data sets for the NN, $\mathcal{Y}_{n}$
and $\mathcal{Y}_{\hat{n}}^{test}$, respectively, using stationary
block bootstrap resampling of empirical historical asset returns (see
Appendix \ref{sec: Appendix Parameters-for-numerical results}), which
is popular with practitioners 
(\cite{Cederburg_2022,CogneauZakalmouline2013,dichtl2016,Scott_2017,Scott_2022,Simonian_2022})
and is designed to handle weakly stationary time series with serial
dependence.   See \cite{NiLiForsyth2020} for a discussion concerning
the probability of obtaining a repeated path in block bootstrap 
resampling (which is negligible for any realistic number
of samples).  Due to availability of historical data we use inflation-adjusted
monthly empirical returns from 1963:07 to 2020:12. The training data
set ($n=10^{6}$) is obtained using an expected block size of 6 months
of joint returns from 1963:07 to 2009:12, while the testing data set
($n=10^{6}$) uses an expected block size of 3 months and returns
from 2010:01 to 2020:12. We consider NNs with two hidden layers, each
with only eight hidden nodes.

Choosing two values of $\tilde{\rho}>0$ to illustrate different levels
of risk aversion (see Table \ref{tab: MV vs DSQ ground truth}), we
solve problem $MV\left(\rho=\tilde{\rho}\right)$ in (\ref{eq: MV objective})
directly using the proposed approach to obtain the optimal investment
strategy $\boldsymbol{f}\left(\cdot;\hat{\boldsymbol{\theta}}_{mv}^{\ast}\right)$.
Note that since we consider a fixed NN structure in this setting rather
than a sequence of NNs, we drop the subscript ``$n$'' in the notation
$\boldsymbol{f}\left(\cdot;\hat{\boldsymbol{\theta}}_{mv}^{\ast}\right)$.
Using this result together with (\ref{eq: Embedding result}), we
can approximate the associated value of $\tilde{\gamma}$ by 
\begin{eqnarray}
\tilde{\gamma} & \simeq & \frac{1}{2\tilde{\rho}}+\frac{1}{n}\sum_{j=1}^{n}W^{\ast\left(j\right)}\left(T;\hat{\boldsymbol{\theta}}_{mv}^{\ast},,\mathcal{Y}_{n}\right),\label{eq: Gamma approx for DSQ}
\end{eqnarray}
and solve problem $DSQ\left(\gamma=\tilde{\gamma}\right)$ independently
using the proposed approach on the same training data set $\mathcal{Y}_{n}$. 

According to Proposition \ref{prop: Embedding result}, the resulting
investment strategy $\boldsymbol{f}\left(\cdot;\hat{\boldsymbol{\theta}}_{dsq}^{\ast}\right)$
should be (approximately) identical to the strategy $\boldsymbol{f}\left(\cdot;\hat{\boldsymbol{\theta}}_{mv}^{\ast}\right)$
if the proposed approach works as required. Note that the parameter
vectors are expected to be different (i.e. $\hat{\boldsymbol{\theta}}_{dsq}^{\ast}\neq\hat{\boldsymbol{\theta}}_{mv}^{\ast}$)
due to a variety of reasons (multiple local minima, optimization using
SGD, etc.), but the resulting wealth distributions and asset allocation
should agree, i.e. $\boldsymbol{f}\left(\cdot;\hat{\boldsymbol{\theta}}_{dsq}^{\ast}\right)\simeq\boldsymbol{f}\left(\cdot;\hat{\boldsymbol{\theta}}_{mv}^{\ast}\right)$. 

Figure \ref{fig: MV vs DSQ ground truth} demonstrates the investment
strategies $\boldsymbol{f}\left(\cdot;\hat{\boldsymbol{\theta}}_{mv}^{\ast}\right)$
and $\boldsymbol{f}\left(\cdot;\hat{\boldsymbol{\theta}}_{dsq}^{\ast}\right)$
obtained by training the NNs on the same training data set using values
of $\tilde{\rho}=0.017$ and $\tilde{\gamma}=429.647$, respectively.
Note that the values $\tilde{\rho}$ and $\tilde{\gamma}$ are rounded
to three decimal places, and Figure \ref{fig: MV vs DSQ ground truth}
corresponds to Results set 1 in Table \ref{tab: MV vs DSQ ground truth}.
In this example, only four of the underlying candidate assets have
non-zero investments, which is to be expected due to the high correlation
between long-only equity factor indices. 

\noindent 
\begin{figure}[!tbh]
\noindent \begin{centering}
\subfloat[$MV\left(\rho=\tilde{\rho}\right)$ - Momentum]{\includegraphics[scale=0.5]{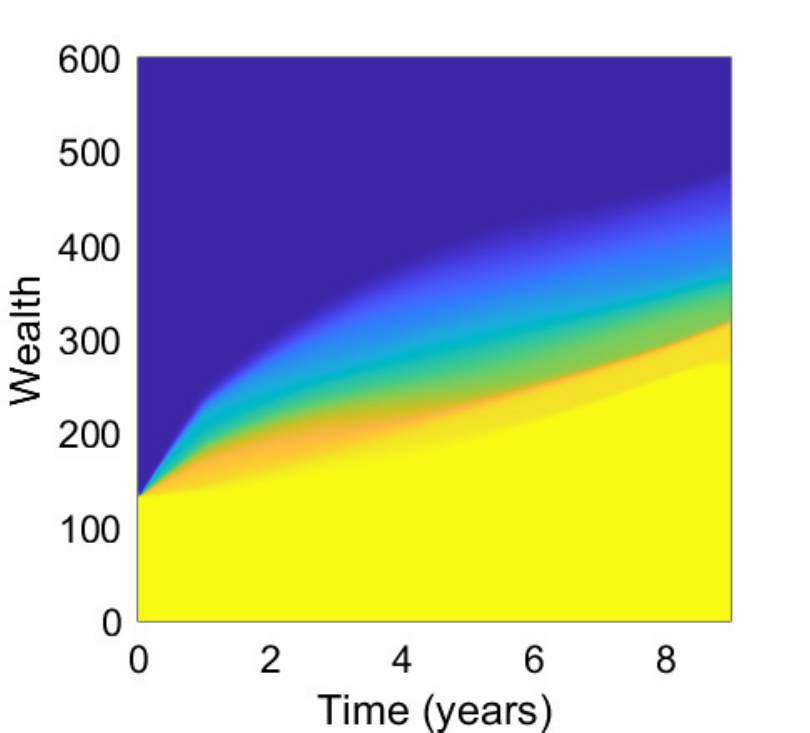}

}\subfloat[$MV\left(\rho=\tilde{\rho}\right)$ - Value]{\includegraphics[scale=0.5]{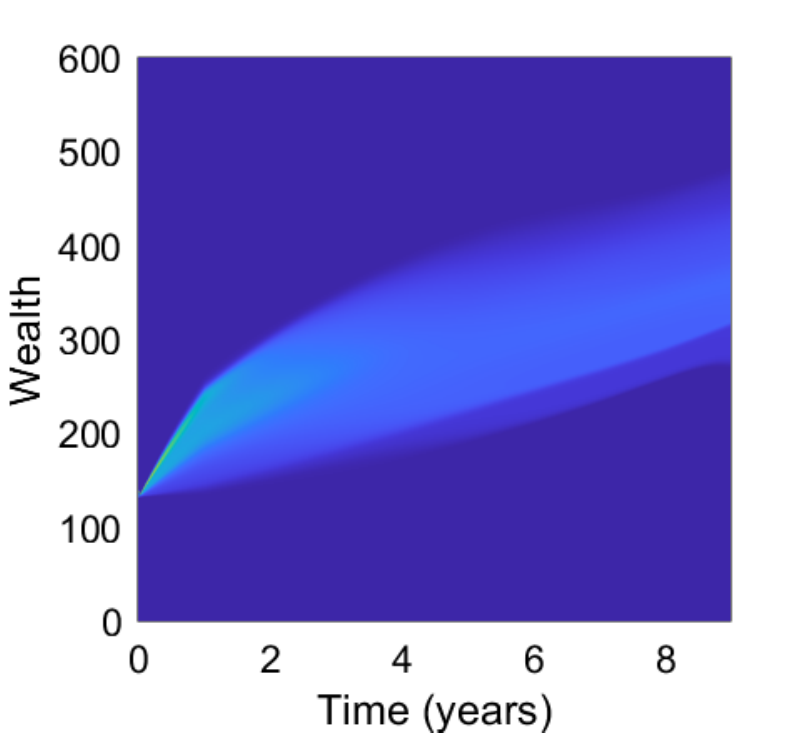}

}\subfloat[$MV\left(\rho=\tilde{\rho}\right)$ - B10]{\includegraphics[scale=0.5]{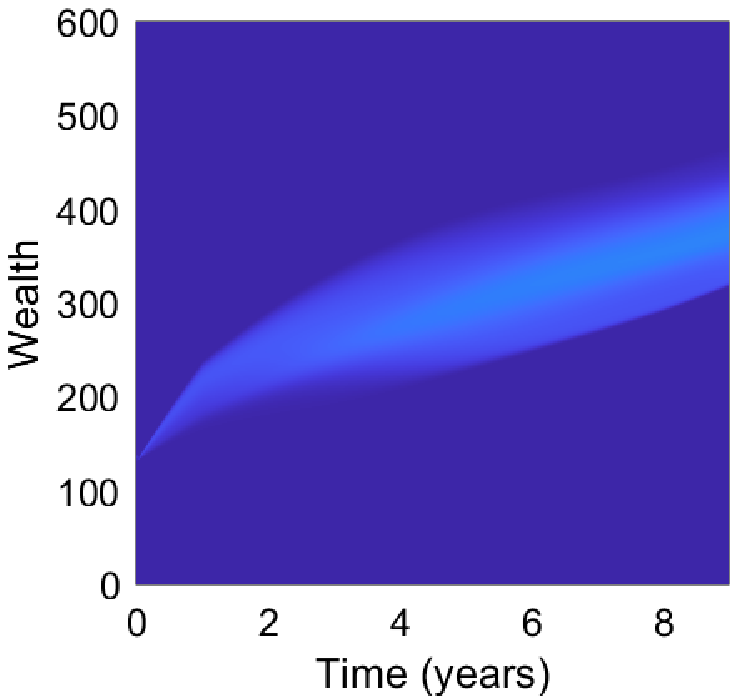}

}\subfloat[$MV\left(\rho=\tilde{\rho}\right)$ - T30]{\includegraphics[scale=0.5]{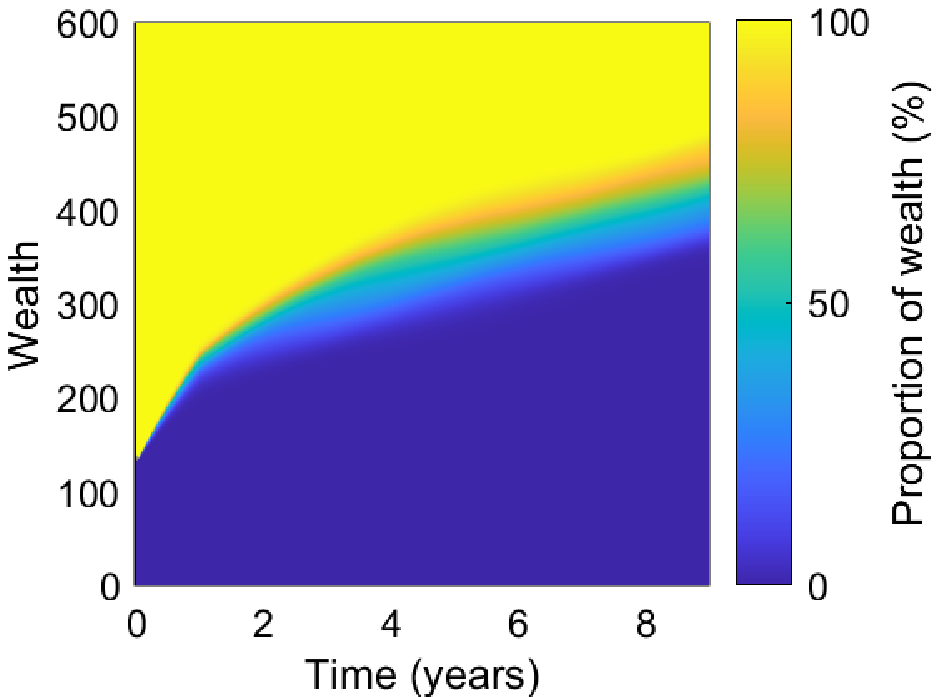}

}
\par\end{centering}
\noindent \begin{centering}
\subfloat[$DSQ\left(\gamma=\tilde{\gamma}\right)$ - Momentum]{\includegraphics[scale=0.5]{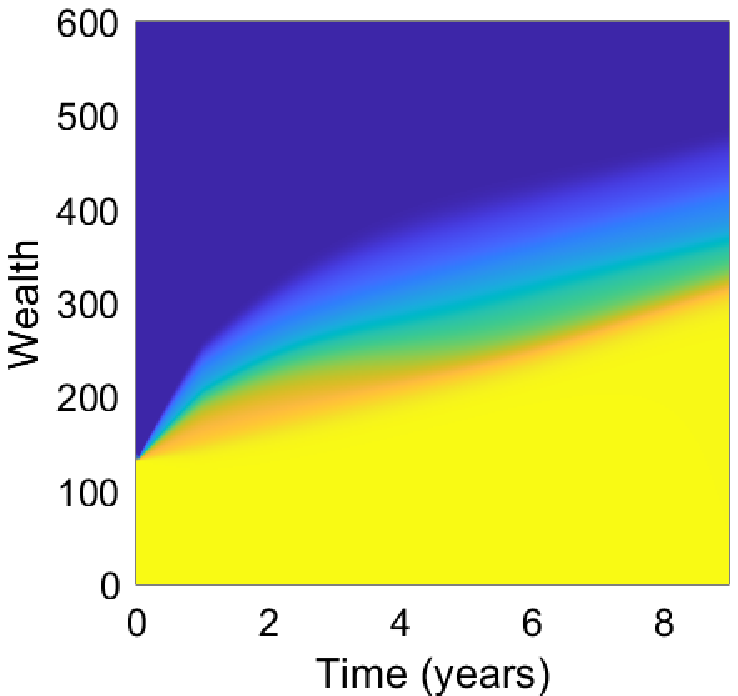}

}\subfloat[$DSQ\left(\gamma=\tilde{\gamma}\right)$ - Value]{\includegraphics[scale=0.5]{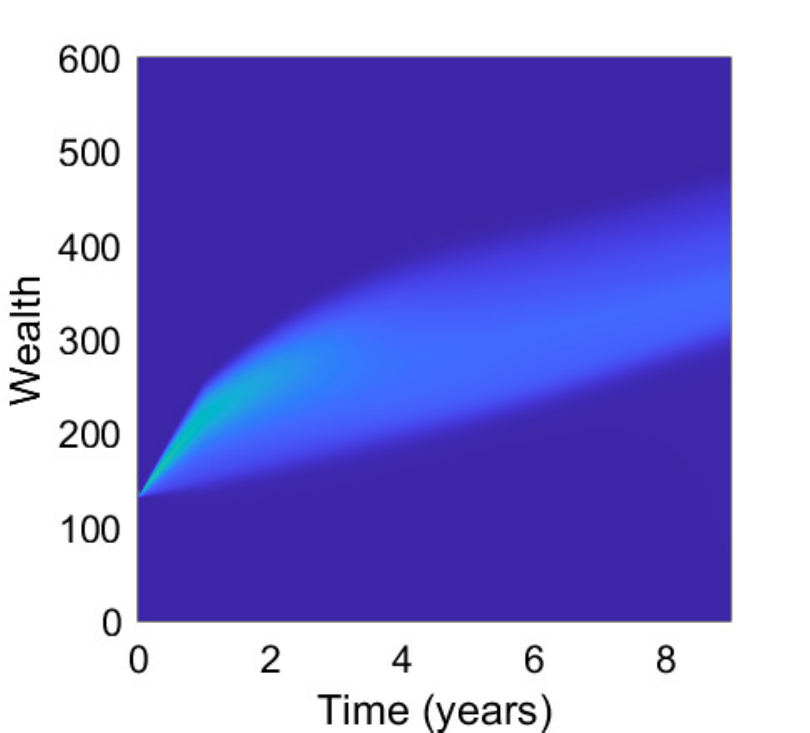}

}\subfloat[$DSQ\left(\gamma=\tilde{\gamma}\right)$ - B10]{\includegraphics[scale=0.5]{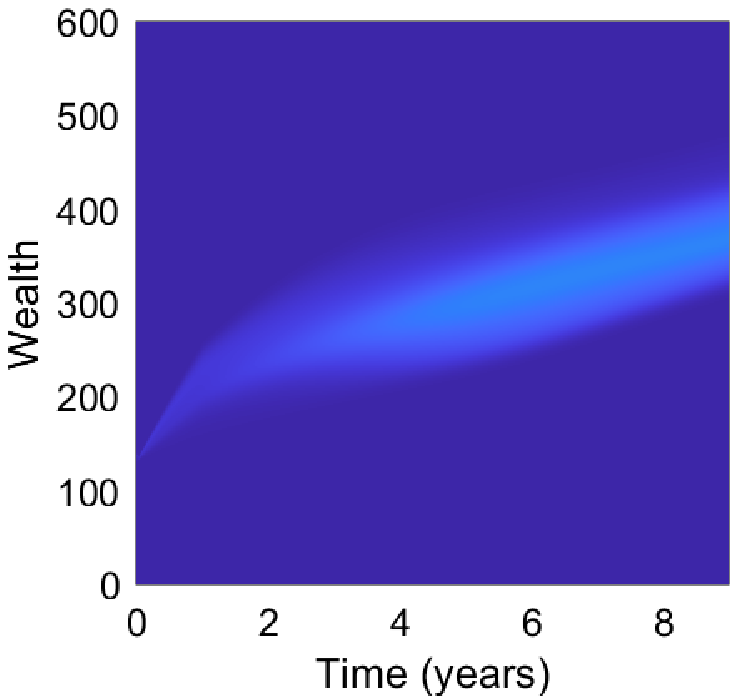}

}\subfloat[$DSQ\left(\gamma=\tilde{\gamma}\right)$ - T30]{\includegraphics[scale=0.5]{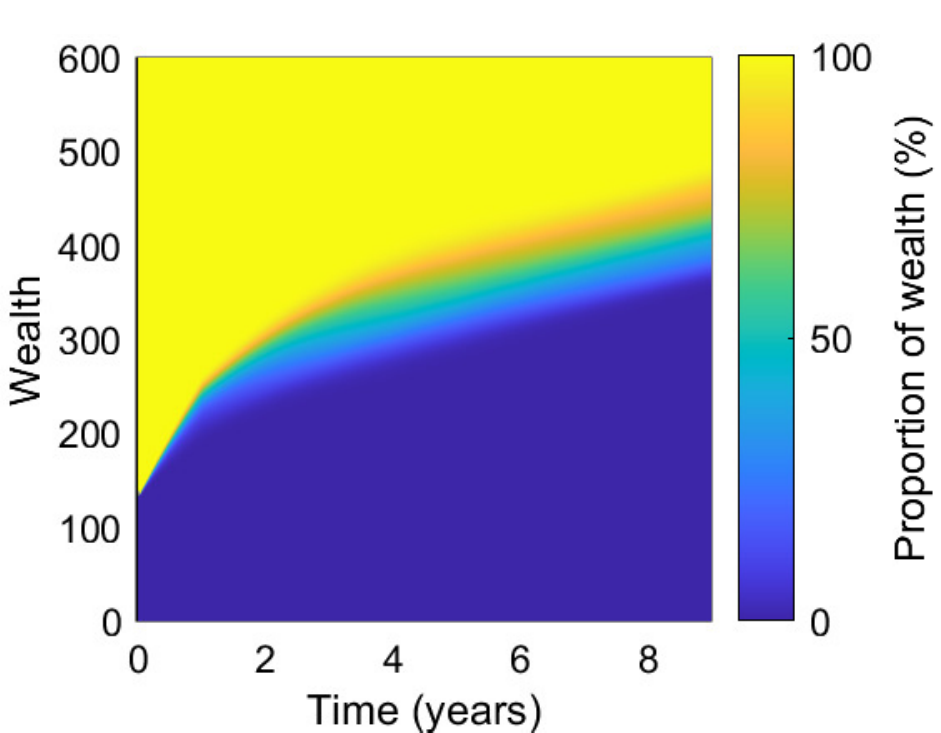}

}
\par\end{centering}
\caption{Ground truth - problems $MV\left(\rho=\tilde{\rho}\right)$ and $DSQ\left(\gamma=\tilde{\gamma}\right)$:
investment strategies $\boldsymbol{f}\left(\cdot;\hat{\boldsymbol{\theta}}_{mv}^{\ast}\right)$
and $\boldsymbol{f}\left(\cdot;\hat{\boldsymbol{\theta}}_{dsq}^{\ast}\right)$
obtained by training the NNs using values of $\tilde{\rho}=0.017$
and $\tilde{\gamma}=429.647$ (rounded to three decimal places), respectively.
Each figure shows the proportion of wealth invested in the asset as
a function of the minimal NN features, namely time and available wealth.
Zero investment under the optimal strategies in the broad market index
and the Size factor. \label{fig: MV vs DSQ ground truth}}
\end{figure}

Table \ref{tab: MV vs DSQ ground truth} confirms that the associated
optimal terminal wealth distributions of $MV\left(\rho=\tilde{\rho}\right)$
and $DSQ\left(\gamma=\tilde{\gamma}\right)$ indeed correspond, both
in-sample (training data set) and out-of-sample (testing data set).

\noindent 
\begin{table}[!tbh]
\caption{Ground truth - problems $MV\left(\rho=\tilde{\rho}\right)$ and $DSQ\left(\gamma=\tilde{\gamma}\right)$:
Terminal wealth results obtained using $n=10^{6}$ joint paths for
the underlying assets. Note that the values of $\tilde{\rho}$ and
$\tilde{\gamma}$ are rounded to three decimal places, . \label{tab: MV vs DSQ ground truth}}

\noindent \centering{}{\footnotesize{}}%
\begin{tabular}{|>{\centering}p{2cm}||>{\centering}p{1.3cm}|>{\centering}p{1.3cm}|>{\centering}p{1.3cm}|>{\centering}p{1.3cm}||>{\centering}p{1.3cm}|>{\centering}p{1.3cm}|>{\centering}p{1.3cm}|>{\centering}p{1.3cm}|}
\cline{2-9} \cline{3-9} \cline{4-9} \cline{5-9} \cline{6-9} \cline{7-9} \cline{8-9} \cline{9-9} 
\multicolumn{1}{>{\centering}p{2cm}||}{} & \multicolumn{4}{c||}{{\footnotesize{}Results set 1: $\tilde{\rho}=0.017$, $\tilde{\gamma}=429.647$}} & \multicolumn{4}{c|}{{\footnotesize{}Results set 2: $\tilde{\rho}=0.0097$, $\tilde{\gamma}=493.196$}}\tabularnewline
\hline 
{\footnotesize{}$W\left(T\right)$} & \multicolumn{2}{c|}{{\footnotesize{}Training data}} & \multicolumn{2}{c||}{{\footnotesize{}Testing data}} & \multicolumn{2}{c|}{{\footnotesize{}Training data}} & \multicolumn{2}{c|}{{\footnotesize{}Testing data}}\tabularnewline
\cline{2-9} \cline{3-9} \cline{4-9} \cline{5-9} \cline{6-9} \cline{7-9} \cline{8-9} \cline{9-9} 
{\footnotesize{}distribution} & {\footnotesize{}MV} & {\footnotesize{}DSQ} & {\footnotesize{}MV} & {\footnotesize{}DSQ} & {\footnotesize{}MV} & {\footnotesize{}DSQ} & {\footnotesize{}MV} & {\footnotesize{}DSQ}\tabularnewline
\hline 
\hline 
{\footnotesize{}Mean} & {\footnotesize{}400.2} & {\footnotesize{}400.3} & {\footnotesize{}391.2} & {\footnotesize{}391.6} & {\footnotesize{}441.5} & {\footnotesize{}441.8} & {\footnotesize{}441.8} & {\footnotesize{}441.5}\tabularnewline
\hline 
{\footnotesize{}Stdev} & {\footnotesize{}55.4} & {\footnotesize{}55.4} & {\footnotesize{}26.2} & {\footnotesize{}25.7} & {\footnotesize{}79.6} & {\footnotesize{}79.7} & {\footnotesize{}39.4} & {\footnotesize{}39.5}\tabularnewline
\hline 
{\footnotesize{}5th percentile} & {\footnotesize{}276.5} & {\footnotesize{}276.4} & {\footnotesize{}346.6} & {\footnotesize{}347.5} & {\footnotesize{}255.2} & {\footnotesize{}254.6} & {\footnotesize{}367.8} & {\footnotesize{}367.1}\tabularnewline
\hline 
{\footnotesize{}25th percentile} & {\footnotesize{}391.8} & {\footnotesize{}392.3} & {\footnotesize{}382.4} & {\footnotesize{}382.8} & {\footnotesize{}422.4} & {\footnotesize{}423.6} & {\footnotesize{}430.9} & {\footnotesize{}430.7}\tabularnewline
\hline 
{\footnotesize{}50th percentile} & {\footnotesize{}416.1} & {\footnotesize{}416.3} & {\footnotesize{}396.5} & {\footnotesize{}396.8} & {\footnotesize{}469.8} & {\footnotesize{}470.1} & {\footnotesize{}451.3} & {\footnotesize{}451.2}\tabularnewline
\hline 
{\footnotesize{}75th percentile} & {\footnotesize{}429.9} & {\footnotesize{}429.8} & {\footnotesize{}406.4} & {\footnotesize{}406.7} & {\footnotesize{}487.7} & {\footnotesize{}489.6} & {\footnotesize{}465.0} & {\footnotesize{}464.8}\tabularnewline
\hline 
{\footnotesize{}95th percentile} & {\footnotesize{}452.1} & {\footnotesize{}452.1} & {\footnotesize{}418.9} & {\footnotesize{}419.0} & {\footnotesize{}516.1} & {\footnotesize{}516.5} & {\footnotesize{}480.9} & {\footnotesize{}480.2}\tabularnewline
\hline 
\end{tabular}{\footnotesize\par}
\end{table}

The proposed NN approach therefore clearly works as expected, in that
we demonstrated that the result of Proposition \ref{prop: Embedding result}
in a completely model-independent way in a portfolio optimization
setting where no known analytical solutions exist. In particular,
we emphasize that no assumptions were made regarding parametric underlying
asset dynamics, the results are entirely data-driven. As a result,
we can interpret the preceding results as showing that the approach
correctly recovers the time-inconsistent (or pre-commitment) strategy
without difficulty if the objective is not separable in the sense
of dynamic programming, such as in the case of the $MV\left(\rho\right)$
problem, whereas if the objective is separable in the sense of dynamic
programming, such as in the case of the $DSQ\left(\gamma\right)$
problem, the approach correctly recovers the associated time-consistent
strategy. 

\subsection{Mean - Semi-variance strategies\label{subsec: Num results - mean semi-variance strategies}}

Having demonstrated the reliability of the results obtained using
the proposed NN approach with the preceding ground truth analyses,
we now consider the solution of the Mean - Semi-variance problem (\ref{eq: Sortino}).
To provide the necessary context to interpret the $MSemiV\left(\rho\right)$-optimal
results, we compare the results of the optimal solutions of the $MCV\left(\rho=\rho_{mcv}\right)$,
$MSemiV\left(\rho=\rho_{msv}\right)$, and $OSQ\left(\gamma=\gamma_{osq}\right)$
problems, where the values of $\rho_{mcv}$, $\rho_{msv}$ and $\gamma_{osq}$
are selected to obtain the same expected value of terminal wealth
on the NN training data set. This is done since the MCV- and OSQ-optimal
strategies have been analyzed in great detail (\cite{DangForsyth2016,Forsyth2019CVaR}),
and are therefore well understood. Note that since all three strategies
are related to the maximization of the mean terminal wealth and while
simultaneously minimizing some risk measure (which is implicitly done
in the case of the OSQ problem, see \cite{DangForsyth2016}), it is
natural to compare the strategies on the basis of equal expectation
of terminal wealth.

To highlight the main qualitative features of the $MSemiV\left(\rho\right)$-optimal
results, we consider a simple investment scenario of two assets, namely
30-day T-bills and a broad equity market index (the VWD index) - see
Appendix \ref{sec: Appendix Parameters-for-numerical results} for
definitions. We choose $T=$5 years, $w_{0}=1000$, and zero contributions
to demonstrate a lump sum investment scenario with quarterly rebalancing. 

To illustrate the flexibility of the NN approach to underlying data
generating assumptions, the NN training data sets are constructed
using generative adversarial network (GAN)-generated synthetic asset
returns obtained by implementing the TimeGAN algorithm proposed by
\cite{YoonTimeGAN2019}. In more detail, using empirical monthly asset
returns from 1926:01 to 2019:12 for the underlying assets (data sources
are specified in Appendix \ref{sec: Appendix Parameters-for-numerical results}),
the TimeGAN is trained with default parameters as in \cite{YoonTimeGAN2019}
using block sizes of 6 months to capture both correlation and serial
correlation aspects of the (joint) time series.\footnote{\color{black}
It appears that the actual code in \cite{YoonTimeGAN2019} implements the following steps: (i) takes as input
actual price data, (ii) forms rolling blocks of price data and (iii) forms
a single synthetic price path (which is the same length as the original path)
by randomly sampling (without replacement) from the set of rolling blocks.
Step (iii) corresponds to the non-overlapping block bootstrap using a 
fixed block size.
This should be contrasted with stationary block bootstrap resampling of \cite{politis1994}.
Step (i) does not make sense as input to a bootstrap technique, since
the data set is about 10 years long, with an initial price of
\$50 and a final price of \$1200.  We therefore changed Step (i), so that all data was converted to returns prior to being used as input.
}
Once trained, the
TimeGAN is then used to generate a set of $n=10^{6}$ paths of synthetic
asset returns, which is used as the training data set to train the
NNs corresponding to the MCV, MSemiV and OSQ-optimal investment strategies. 

Figure \ref{fig: Num_03_Heatmaps} illustrates the resulting optimal
investment strategies, and we observe that the MSemiV-optimal strategy
is fundamentally different from the MCV and OSQ-optimal strategies,
while featuring elements of both. Specifically, Figure \ref{fig: Num_03_PDFs_CDFs},
which illustrates the resulting optimal terminal wealth distributions
(with the same expectation), demonstrates that the MSemiV strategy,
like the MCV strategy, can offer better downside protection than the
OSQ strategy, while the MSemiV strategy retains some of the qualitative
elements of the OSQ distribution such as the left skew. 

Having illustrated that the MSemiV problem can be solved in a dynamic
trading setting using the proposed NN approach to obtain investment
strategies that offer potentially valuable characteristics, we leave
a more in-depth investigation of the properties and applications of
MSemiV-optimal strategies for future work. 

\noindent 
\begin{figure}[!tbh]
\noindent \begin{centering}
\subfloat[$MCV\left(\rho=\rho_{mcv}\right)$]{\includegraphics[scale=0.6]{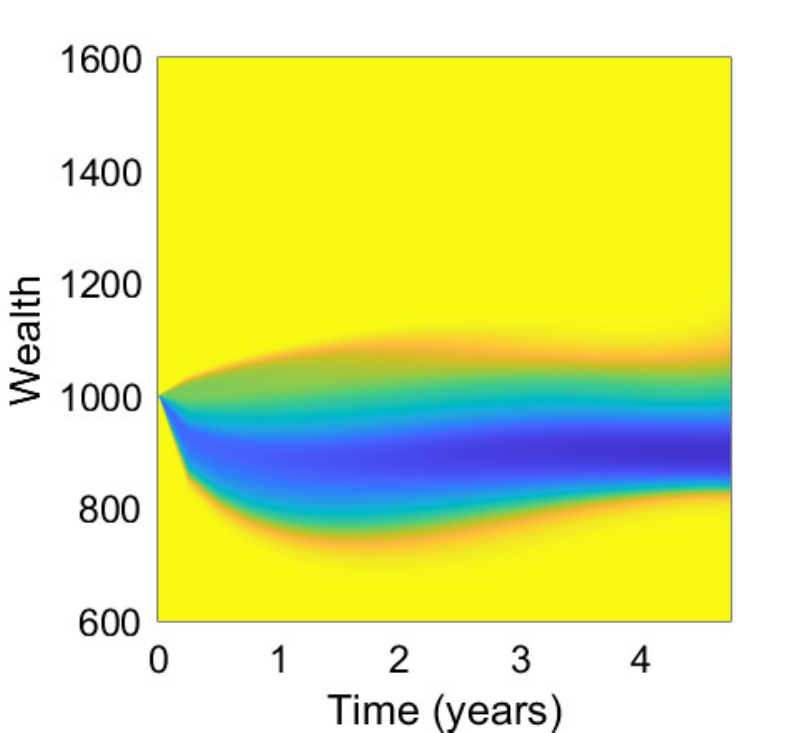}

}\subfloat[$MSemiV\left(\rho=\rho_{msv}\right)$]{\includegraphics[scale=0.6]{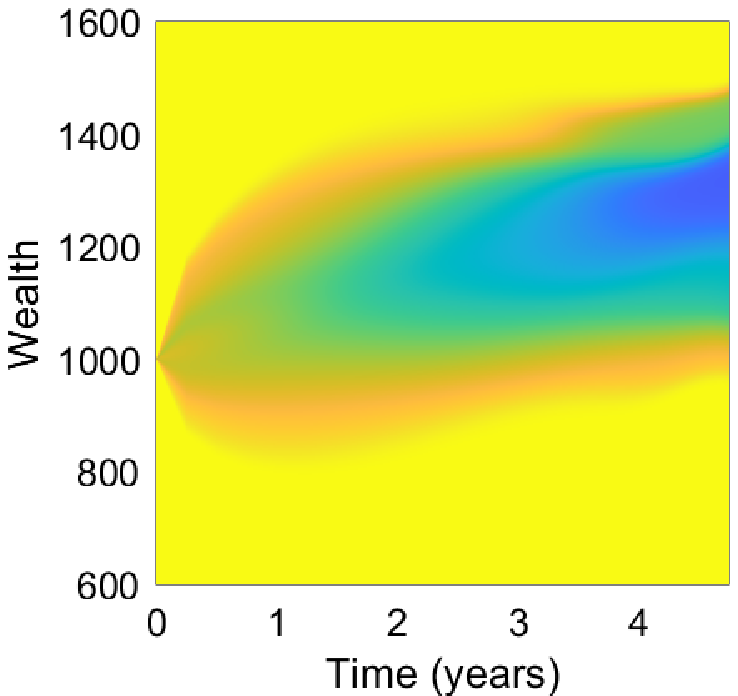}

}\subfloat[$OSQ\left(\gamma=\gamma_{osq}\right)$]{\includegraphics[scale=0.6]{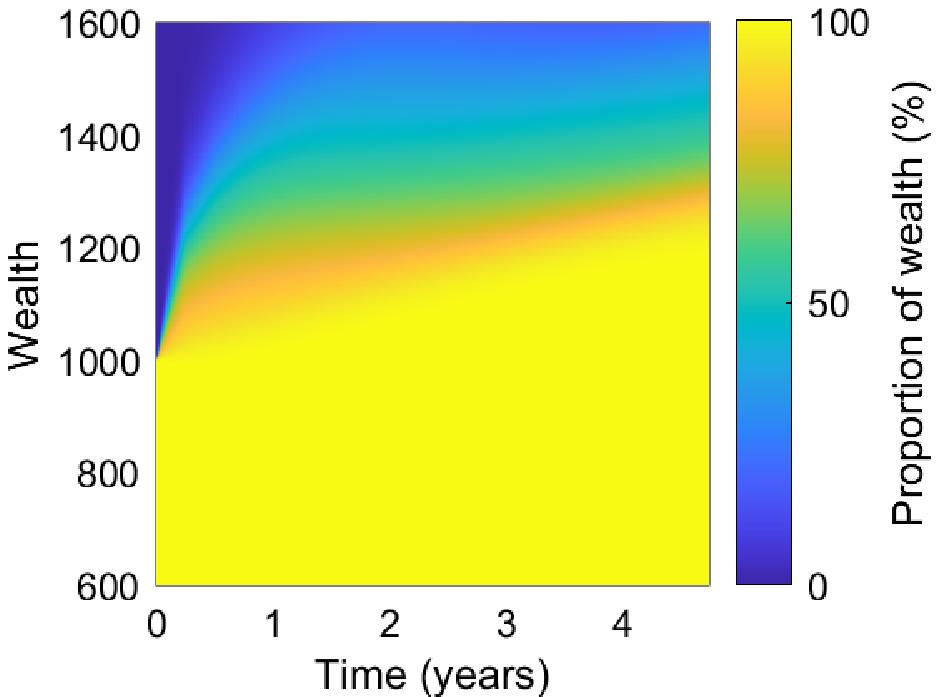}

}
\par\end{centering}
\caption{Optimal investment strategies for the $MCV\left(\rho=\rho_{mcv}\right)$,
$MSemiV\left(\rho=\rho_{msv}\right)$, and $OSQ\left(\gamma=\gamma_{osq}\right)$
strategies, obtaining identical expectation of terminal wealth on
the training data set. Each figure shows the proportion of wealth
invested in the broad equity market index as a function of the minimal
NN features, namely time and available wealth. \label{fig: Num_03_Heatmaps}}
\end{figure}

\noindent 
\begin{figure}[!tbh]
\noindent \begin{centering}
\subfloat[PDFs of $W^{\ast}\left(T\right)$]{\includegraphics[scale=0.72]{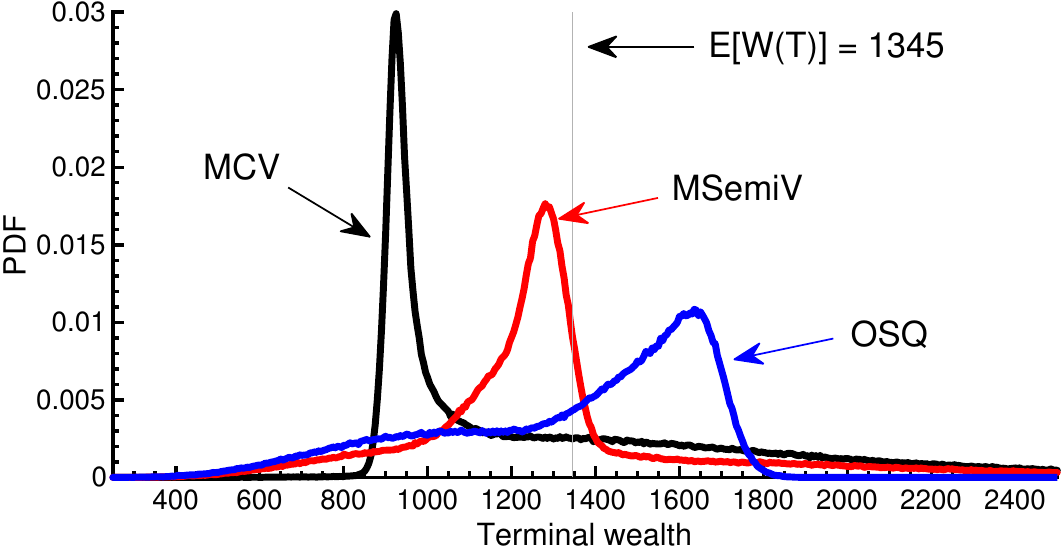}

}$\qquad$\subfloat[CDFs of $W^{\ast}\left(T\right)$]{\includegraphics[scale=0.72]{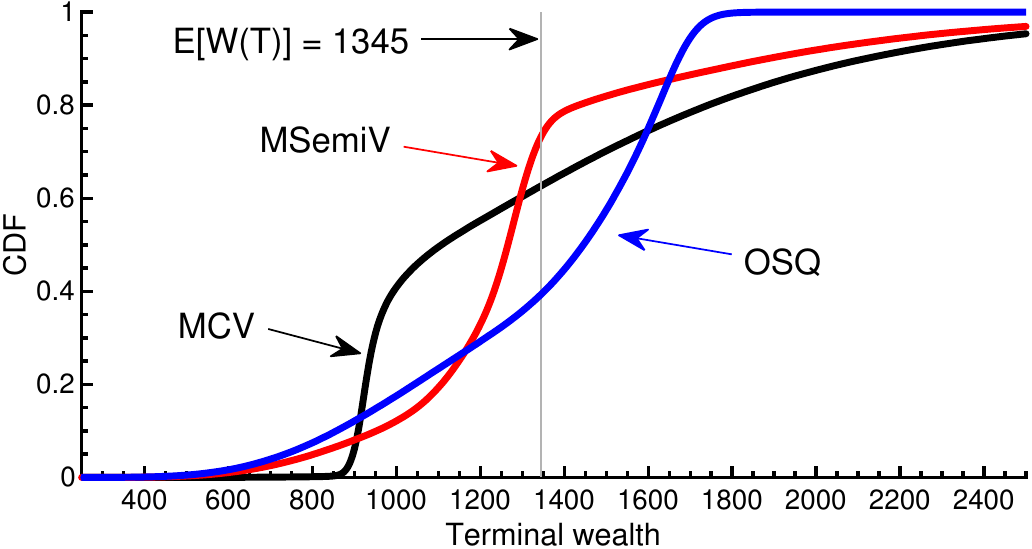}

}
\par\end{centering}
\caption{PDFs and CDFs of optimal terminal wealth obtained under the $MCV\left(\rho=\rho_{mcv}\right)$,
$MSemiV\left(\rho=\rho_{msv}\right)$, and $OSQ\left(\gamma=\gamma_{osq}\right)$
strategies, where the values of $\rho_{mcv}$, $\rho_{msv}$ and $\gamma_{osq}$
are selected to obtain the same expected value of optimal terminal
wealth on the NN training data set. \label{fig: Num_03_PDFs_CDFs}}
\end{figure}

\section{Conclusion\label{sec:Conclusion}}

In this paper, we presented a flexible NN approach, which does not
rely on dynamic programming techniques, to solve a large class of
dynamic portfolio optimization problems. In the proposed approach,
a single optimization problem is solved, issues of instability and
error propagation involved in estimating high-dimensional conditional
expectations are avoided, and the resulting NN is parsimonious in
the sense that the number of parameters does not scale with the number
of rebalancing events. 

We also presented theoretical convergence analysis results which show
that the numerical solution obtained using the proposed approach can
recover the optimal investment strategy, provided it exists, regardless
of whether the resulting optimal investment strategy is time-consistent
or (formally) time-inconsistent. 

Numerical results confirmed the advantages of the NN approach, and
showed that accurate results can be obtained in ground truth analyses
in a variety of settings. The numerical results also highlighted that
the approach remains agnostic as to the underlying data generating
assumptions, so that for example empirical asset returns or synthetic
asset returns can be used without difficulty. 

We conclude by noting that the NN approach is not necessarily limited
to portfolio optimization problems such a those encountered during
the accumulation phase of pension funds, and could be extended to
address the significantly more challenging problems encountered during
the decumulation phase of defined contribution pension funds (see
for example \cite{Forsyth2020DC}). We leave this extension for future
work.

\section{Declarations}
The authors have no competing interests to declare that are relevant to the content of this article. P.A. Forsyth's work was supported by the Natural Sciences and Engineering Research Council of Canada (NSERC) grant RGPIN-2017-03760.

\noindent \setlength{\bibsep}{1pt plus 0.3ex} 
\small

\bibliographystyle{mynatbib}
\bibliography{References_v26}

\noindent \begin{appendices}

\normalsize

\section{NN approach: technical details and analytical results\label{sec:Appendix A:-Technical details and analytical results}}

In this appendix, additional analytical results, relating to the convergence
analysis presented in Section \ref{sec:Convergence-analysis}, are
presented.

\subsection{NN structural assumptions\label{subsec: Appendix NN-structural-assumptions}}

In this section, we discuss the NN structural assumptions. First,
we introduce the necessary notation - for a more detailed treatment
of NNs, see for example \cite{GoodfellowEtAl_BOOK}. Consider a fully-connected,
feed-forward NN $\boldsymbol{f}_{n}$ with $\mathcal{L}^{h}\geq1$
hidden layers. The NN layers are indexed by $\ell\in\left\{ 0,...,\mathcal{L}\right\} $,
where $\ell=0$ and $\ell=\mathcal{L}^{h}+1\equiv\mathcal{L}$ denote
the input and output layers, respectively. Let $\eta_{n,\ell}\in\mathbb{N}$
denote the number of nodes in layer $\ell$ of $\boldsymbol{f}_{n}$.
With the exception of the input layer, each layer $\ell\in\left\{ 1,...,\mathcal{L}\right\} $
is associated with a weights matrix $\boldsymbol{x}_{n}^{\left[\ell\right]}\in\mathbb{R}^{\eta_{n,\ell}\times\eta_{n,\ell-1}}$
into the layer, an optional bias vector $\boldsymbol{b}_{n}^{\left[\ell\right]}\in\mathbb{R}^{\eta_{n,\ell}}$,
as well as an activation function $\boldsymbol{\mathfrak{a}}_{n}^{\left[\ell\right]}:\mathbb{R}^{\eta_{\ell}}\rightarrow\mathbb{R}^{\eta_{\ell}}$
which is applied to the weighted inputs into the layer.

The parameter vector of the NN $\boldsymbol{f}_{n}$, which consists
of all weights and biases, is denoted by $\boldsymbol{\theta}_{n}\in\mathbb{R}^{\nu_{n}}$,
where $\nu_{n}\in\mathbb{N}$ denotes the total number of weights
and biases. In other words, the weights matrices $\left\{ \boldsymbol{x}_{n}^{\left[\ell\right]}:\ell=1,...,\mathcal{L}\right\} $
and optional bias vectors $\left\{ \boldsymbol{b}_{n}^{\left[\ell\right]}:\ell=1,...,\mathcal{L}\right\} $
are transformed into a single vector $\boldsymbol{\theta}_{n}=\left(\theta_{1},...,\theta_{\nu_{n}}\right)$,
where each $\theta_{n,i}\in\boldsymbol{\theta}_{n}$ can be uniquely
mapped to a single weight or bias in some layer.

Note that no activation function is applied at the input layer ($\ell=0$),
so that the $\eta_{0}\equiv\eta_{n,0}$ output values of the input
layer corresponds to feature (input) vector of the NN, which will
be denoted by $\boldsymbol{\phi}\in\mathbb{R}^{\eta_{0}}$. Recalling
that $\eta_{\mathcal{L}}\equiv\eta_{n,\mathcal{L}}$ is the number
of nodes in the output layer ($\ell=\mathcal{L}$) and setting the
bias vectors $\boldsymbol{b}_{n}^{\left[\ell\right]}\equiv\boldsymbol{0}$
for convenience, the NN can therefore be written as a single function
$\boldsymbol{f}_{n}\left(\boldsymbol{\phi};\boldsymbol{\theta}_{n}\right):\mathbb{R}^{\eta_{0}}\rightarrow\mathbb{R}^{\eta_{\mathcal{L}}}$,
where 
\begin{eqnarray}
\boldsymbol{f}_{n}\left(\boldsymbol{\phi};\boldsymbol{\theta}_{n}\right) & \coloneqq & \left(f_{n,1}\left(\boldsymbol{\phi};\boldsymbol{\theta}_{n}\right),...,f_{n,\eta_{\mathcal{L}}}\left(\boldsymbol{\phi};\boldsymbol{\theta}_{n}\right)\right),\qquad\boldsymbol{\phi}\in\mathbb{R}^{\eta_{0}},\boldsymbol{\theta}_{n}\in\mathbb{R}^{\nu_{n}}\label{eq: NN as a single function}
\end{eqnarray}
We highlight that the output of the $i$th node in the output layer
is given by $f_{n,i}\left(\boldsymbol{\phi};\boldsymbol{\theta}_{n}\right)=\mathfrak{a}_{n,i}^{\left[\mathcal{L}\right]}$.

Given this standard fully-connected, feedforward NN formulation, we
introduce the following NN structural assumption. 

\begin{assumption} \label{assu: Appendix NN structure assumptions}(NN
structure) Let $\boldsymbol{f}_{n}\left(\cdot;\boldsymbol{\theta}_{n}\right),n\in\mathbb{N}$,
be a sequence of fully-connected feedforward neural networks, and
let $\hbar\left(n\right),n\in\mathbb{N}$ be a monotonically increasing
sequence (i.e. $\hbar\left(n\right)<\hbar\left(n+1\right)$, $\forall n\in\mathbb{N}$)
such that $\lim_{n\rightarrow\infty}\hbar\left(n\right)=\infty$.
For each $n\in\mathbb{N}$, the NN $\boldsymbol{f}_{n}$ is constructed
to satisfy the following structural assumptions.

\begin{enumerate}[label=(\roman*)]

\item The number of hidden layers $\mathcal{L}^{h}\geq1$ ($\mathcal{L}^{h}\in\mathbb{N}$)
remains fixed for all $n\in\mathbb{N}$. For notational simplicity,
we assume that each of the $\mathcal{L}^{h}$ hidden layers of the
NN $\boldsymbol{f}_{n}$ has the same number $\hbar\left(n\right)$
of hidden nodes,
\begin{eqnarray}
\eta_{n,\ell} & \equiv & \hbar\left(n\right),\qquad\forall\ell=1,...,\mathcal{L}-1,\quad\textrm{ for some }\hbar\left(n\right)\in\mathbb{N}.\label{eq: nr of hidden nodes in each hidden layer}
\end{eqnarray}

\item For convenience, we assume that the sigmoid activation function
$\sigma^{h}$ is applied at each hidden node,
\begin{eqnarray}
\sigma^{h}\left(y\right)=\frac{1}{1+e^{-y}} & \equiv & \mathfrak{a}_{n,i}^{\left[\ell\right]}\left(y\right),\quad\textrm{where }y=\left(\sum_{k=1}^{\eta_{n,\ell-1}}x_{n,ik}^{\left[\ell\right]}\mathfrak{a}_{n,k}^{\left[\ell-1\right]}\right)+b_{n,i}^{\left[\ell\right]},\label{eq: Hidden activation}
\end{eqnarray}
for all $\ell=1,...,\mathcal{L}^{h}$ and $i=1,...,\hbar\left(n\right)$.
Note that in principle, any of the popular activation functions can
be used instead of (\ref{eq: Hidden activation}), with minor modifications
to the theoretical analysis presented in this paper.

\item The NN $\boldsymbol{f}_{n}$ has $\eta_{0}=\eta_{X}+1\equiv\eta_{n,0}$
input nodes (i.e. the number of input nodes are independent of $n\in\mathbb{N}$),
with feature (input) vectors $\boldsymbol{\phi}\in\mathbb{R}^{\eta_{0}}$
of the form
\begin{eqnarray}
\boldsymbol{\phi}\coloneqq\boldsymbol{\phi}\left(t\right) & \coloneqq & \left(t,\boldsymbol{X}\left(t\right)\right)\in\mathcal{D}_{\boldsymbol{\phi}}\subseteq\mathbb{R}^{\eta_{X}+1},\qquad\textrm{with }\boldsymbol{X}\left(t\right)=\left(W\left(t^{+}\right),\hat{\boldsymbol{X}}\left(t\right)\right),\label{eq: Feature vectors NN}
\end{eqnarray}
where $W\left(t^{+}\right)$ denotes the wealth available for investment
at time $t$ after any contributions to the portfolio at time $t$,
while $\hat{\boldsymbol{X}}\left(t\right)$ denotes a vector of additional
information taken into account by the investment strategy. We emphasize
that (\ref{eq: Feature vectors NN}) clarifies that at time $t\in\left[t_{0},T\right]$,
at least time $t$ itself and $W\left(t^{+}\right)$ are always assumed
to be inputs into the NN.

\item The NN $\boldsymbol{f}_{n}$ has $N_{a}=\eta_{n,\mathcal{L}}$
output nodes (i.e. the number of output nodes are independent of $n\in\mathbb{N}$),
with the output of node $i$, denoted by $f_{n,i}\left(\boldsymbol{\phi}\left(t\right);\boldsymbol{\theta}_{n}\right)$,
being associated with the proportion of available wealth $W\left(t^{+}\right)$
invested in asset $i\in\left\{ 1,...,N_{a}\right\} $ after rebalancing
the portfolio at time $t$.

\item The output layer ($\ell=\mathcal{L}=\mathcal{L}^{h}+1$) of
each NN $\boldsymbol{f}_{n}$ uses the softmax activation function
(see for example \cite{GaoPavel2018}). Therefore we have $\boldsymbol{\mathfrak{a}}_{n}^{\left[\mathcal{L}\right]}=\boldsymbol{\psi}:\mathbb{R}^{N_{a}}\rightarrow\mathbb{R}^{N_{a}}$,
where the $i$th component of $\boldsymbol{\psi}=\left(\psi_{i}:i=1,..,N_{a}\right)$
is given by
\begin{equation}
\psi_{i}=\mathfrak{a}_{n,i}^{\left[\mathcal{L}\right]}=\frac{\exp\left\{ z_{n,i}^{\left[\mathcal{L}\right]}\right\} }{\sum_{m=1}^{N_{a}}\exp\left\{ z_{n,m}^{\left[\mathcal{L}\right]}\right\} },\quad\textrm{where }z_{n,i}^{\left[\mathcal{L}\right]}=\sum_{k=1}^{N_{a}}x_{n,ik}^{\left[\mathcal{L}\right]}\mathfrak{a}_{n,k}^{\left[\mathcal{L}-1\right]}+b_{n,i}^{\left[\mathcal{L}\right]},\quad i=1,...,N_{a}.\label{eq: Softmax output}
\end{equation}

\end{enumerate}

\end{assumption}

For a given $n\in\mathbb{N}$, we define the set $\mathcal{N}_{n}$
as the set of all neural networks satisfying Assumption (\ref{assu: Appendix NN structure assumptions}),

\begin{equation}
\mathcal{N}_{n}=\left\{ \left.\boldsymbol{f}_{n}:\mathcal{D}_{\boldsymbol{\phi}}\rightarrow\mathcal{Z}\right|\boldsymbol{f}_{n}\left(\cdot;\boldsymbol{\theta}_{n}\right)\textrm{ satisfies Assumption }\ref{assu: Appendix NN structure assumptions}\textrm{ with }\hbar\left(n\right)\textrm{ nodes in each hidden layer}\right\} .\label{eq: Definition set N_n}
\end{equation}
In other words, each $\boldsymbol{f}_{n}\left(\cdot;\boldsymbol{\theta}_{n}\right)\in\mathcal{N}_{n}$
has the same number of hidden nodes $\hbar\left(n\right)$ in each
hidden layer, but a potentially different parameter vector $\boldsymbol{\theta}_{n}$
(i.e. different values associated with the weights and biases).

We make the following observations regarding Assumption \ref{assu: Appendix NN structure assumptions}:
\begin{itemize}
\item Any NN constructed to satisfy Assumption \ref{assu: Appendix NN structure assumptions}
will, for any input vector $\boldsymbol{\phi}\left(t\right)$, automatically
generate an output in the set $\mathcal{Z}$, hence the definition
(\ref{eq: Definition set N_n}) noting that $\boldsymbol{f}_{n}:\mathcal{D}_{\boldsymbol{\phi}}\rightarrow\mathcal{Z}$.
In other words, the given constraints are automatically satisfied.
However, different sets of constraints simply requires modifications
to the output activation, or post-processing of NN outputs, without
affecting the technical results. 
\item Note that further assumptions regarding the rate of at which the sequence
$\hbar\left(n\right)$ increases relative to that of the sequence
$\left\{ n\right\} _{n\in\mathbb{N}}$ will be introduced in the convergence
analysis of Section \ref{sec:Convergence-analysis} (see Assumption
\ref{assu: Appendix Convergence assumptions - computational}.
\item In practical applications, it is not necessary to consider a sequence
of NNs; instead, we will use a single NN $\boldsymbol{f}_{\tilde{n}}$
with $\hbar\left(\tilde{n}\right)$ hidden nodes in each of the hidden
layers to get a reasonable trade-off between accuracy and computational
efficiency. However, we emphasize that any such $\boldsymbol{f}_{\tilde{n}}$
is still constructed to satisfy Assumption \ref{assu: Appendix NN structure assumptions}.
\end{itemize}

\subsection{Assumptions for convergence analysis\label{subsec: Appendix assumptions for convergence}}

Assumption \ref{assu: Appendix Convergence assumptions - critical}
introduces the main assumptions used in rigorously justifying the
approximation (\ref{eq: NN as control}) and therefore to prove Theorem
\ref{thm: Validity of NN approximation}.

\begin{assumption} \label{assu: Appendix Convergence assumptions - critical}
(Convergence analysis: NN approximation to control) To establish the
validity of the NN approximation to the control, we make the following
assumptions:

\begin{enumerate}[label=(\roman*)]

\item The optimal investment strategy (or control) satisfies Assumption
\ref{assu: Existence and Continuity of the control}.

\item The functions $F$ and $G$ in the objective functional $J\left(\boldsymbol{p},\xi;t_{0},w_{0}\right)$
(see (\ref{eq: Objective functional with continuous control p}))
are continuous, and $\xi\rightarrow F\left(\cdot,\xi\right)$ and
$\xi\rightarrow G\left(\cdot,\cdot,\cdot,\xi\right)$ are convex for
any admissible strategy $\boldsymbol{p}\in C\left(\mathcal{D}_{\boldsymbol{\phi}},\mathcal{Z}\right)$.
Note that in for example the Mean - Conditional Value-at-Risk problem
(\ref{eq: MCV objective}) where there is an inner and outer optimization
problem, this assumption is standard in computational settings (\cite{Forsyth2019CVaR}).

\item The NN approximation $\eqref{eq: NN as control}$ of the investment
strategy $\boldsymbol{p}\in C\left(\mathcal{D}_{\boldsymbol{\phi}},\mathcal{Z}\right)$
is implemented by a NN $\boldsymbol{f}_{n}\left(\cdot;\boldsymbol{\theta}_{n}\right)\in\mathcal{N}_{n}$,
where $\mathcal{N}_{n}$ is given by (\ref{eq: Definition set N_n}).
In other words, each approximating NN in the sequence of NNs $\boldsymbol{f}_{n},n\in\mathbb{N}$
is constructed according to Assumption \ref{assu: Appendix NN structure assumptions}. 

\end{enumerate}

\end{assumption}

Note that Assumption \ref{assu: Appendix Convergence assumptions - critical}(iii)
specifically requires that Assumption \ref{subsec: Appendix NN-structural-assumptions}
is satisfied, so each $\boldsymbol{f}_{n},n\in\mathbb{N}$, has $\hbar\left(n\right)$
nodes in each hidden layer, where we recall that the sequence $\hbar\left(n\right),n\in\mathbb{N},$
is monotonically increasing and satisfies $\hbar\left(n\right)\rightarrow\infty$
as $n\rightarrow\infty$. However, we make no further assumptions
yet regarding the form of $n\rightarrow\hbar\left(n\right)$. 

For ease of exposition, we introduce Assumption \ref{assu: Convergence assumptions - convenience}
below. We emphasize that Assumption \ref{assu: Convergence assumptions - convenience}
is purely for the sake of convenience, with Remark \ref{rem: Appendix Relaxing Convergence assumptions - convenience}
below discussing briefly how each component of Assumption \ref{assu: Convergence assumptions - convenience}
can be relaxed with only minor (but tedious and notationally demanding)
modifications to the subsequent proofs.

\begin{assumption} \label{assu: Convergence assumptions - convenience}
(Convergence analysis: Assumptions for ease of exposition) For convenience,
we introduce the following assumptions which can be relaxed without
difficulty, as discussed in Remark \ref{rem: Appendix Relaxing Convergence assumptions - convenience}
below. 

\begin{enumerate}[label=(\roman*)]

\item We assume that the optimal control $\boldsymbol{p}^{\ast}$
as per Assumption \ref{assu: Existence and Continuity of the control}
is a function of time and wealth only, i.e. $\boldsymbol{X}^{\ast}\left(t_{m}\right)=W^{\ast}\left(t_{m}^{+}\right)$
for each $t_{m}\in\mathcal{T}$ in (\ref{eq: Assumption CONTINUOUS control}).
As a result, we work with the minimal form of the NN feature vector
satisfying Assumption \ref{assu: Appendix NN structure assumptions}.
Specifically, in the subsequent results we will always assume that
$\boldsymbol{X}\left(t\right)=W\left(t^{+}\right)$, so that we will
consider feature vectors (\ref{eq: Feature vectors NN}) of the form
\begin{eqnarray}
\boldsymbol{\phi}\left(t\right) & = & \left(t,W\left(t^{+};\boldsymbol{\theta}_{n},\boldsymbol{Y}\right)\right)\in\mathcal{D}_{\boldsymbol{\phi}}\subseteq\mathbb{R}^{2}.\label{eq: Minimal feature vector}
\end{eqnarray}

\item The wealth process with dynamics given by (\ref{eq: W dynamics with NN as control})
remains bounded. In other words, we assume that there exists a value
$w_{max}>0$ such that
\begin{equation}
0\leq W\left(t;\boldsymbol{\theta}_{n},\boldsymbol{Y}\right)\leq w_{max}\quad\textrm{a.s.}\qquad\textrm{for all }t\in\left[t_{0},T\right],\boldsymbol{\theta}_{n}\in\mathbb{R}^{\nu_{n}},\label{eq: assumption boundedness of wealth}
\end{equation}
so that $\mathcal{D}_{\boldsymbol{\phi}}$ in (\ref{eq: Minimal feature vector})
satisfies
\begin{eqnarray}
\mathcal{D}_{\boldsymbol{\phi}} & = & \left[t_{0},T\right]\times\left[0,w_{max}\right].\label{eq: D_phi bounded domain}
\end{eqnarray}

\end{enumerate}

\end{assumption}

The following remark discusses how Assumption \ref{assu: Convergence assumptions - convenience}
can be relaxed. 

\noindent \begin{brem}\label{rem: Appendix Relaxing Convergence assumptions - convenience}(Relaxing
Assumption \ref{assu: Convergence assumptions - convenience}) As
noted above, Assumption \ref{assu: Convergence assumptions - convenience}
has been introduced for ease of exposition. We therefore briefly describe
how each element of element of Assumption \ref{assu: Convergence assumptions - convenience}
can be relaxed without difficulty. 

\noindent \begin{enumerate}[label=(\roman*)]

\item In the case where the state $\boldsymbol{X}^{\ast}\left(t_{m}\right)$
depends on variables in addition to the portfolio wealth, for example
historical returns or additional variables (see for example \cite{Forsyth2019CVaR,TsangWong2020}),
it is straightforward to incorporate these extra values without materially
impacting the key aspects of the convergence analysis. However, it
is essential that portfolio wealth is included in $\boldsymbol{X}^{\ast}\left(t_{m}\right)$
as per Assumption \ref{assu: Existence and Continuity of the control}.

\item The assumption of bounded wealth (\ref{eq: assumption boundedness of wealth})
is clearly practical, in that while it is undoubtedly true that the
entire wealth of the world is very large, it remains finite. However,
from a theoretical perspective, the only reason we introduce (\ref{eq: assumption boundedness of wealth})
is to ensure that, given the minimal form of the feature vector (\ref{eq: Minimal feature vector}),
the controls take inputs in a compact domain (\ref{eq: D_phi bounded domain}).
While boundedness assumptions can be relaxed without theoretical difficulty
using straightforward localization arguments (see for example \cite{HureEtAl2021,TsangWong2020}),
this simply introduces yet further notational complexity without providing
additional insights into the fundamental arguments underlying the
subsequent proofs.

\end{enumerate}

\noindent \qed\end{brem}

In the convergence analysis of Step 2 of the proposed approach, namely
the computational estimate of the optimal control obtained using (\ref{eq: V_n_HAT approx}),
we need to introduce some additional assumptions (Assumption \ref{assu: Appendix Convergence assumptions - computational}
below) since this step involves the training dataset $\mathcal{Y}_{n}$
of the NN and numerical solution of problem (\ref{eq: V_n_HAT approx}).

\begin{assumption} \label{assu: Appendix Convergence assumptions - computational}
(Convergence analysis: Computational estimate of optimal control)
We introduce the following assumptions:

\begin{enumerate}[label=(\roman*)]

\item The training data set $\mathcal{Y}_{n}=\left\{ \boldsymbol{Y}^{\left(j\right)}:j\in\left\{ 1,...,n\right\} \right\} $
used for training the NN (see (\ref{eq: Y training}) and associated
discussion) is constructed with independent joint asset return paths
$\boldsymbol{Y}^{\left(j\right)}\in\mathcal{Y}_{n}$. As noted before,
{\color{black} this does not assume that the joint asset returns along a given path are independent or serially independent.}

\item Number of nodes in each hidden layer $\hbar\left(n\right)$,$n\in\mathbb{N}$:
As $n\rightarrow\infty$ ($n\in\mathbb{N}$), in the case of one hidden
layer $\left(\mathcal{L}^{h}=1\right)$, we assume that $\hbar\left(n\right)=o\left(n^{1/4}\right)$.
For deeper NNs $\left(\mathcal{L}^{h}>1\right)$, we assume that $\hbar\left(n\right)=o\left(n^{1/6}\right)$. 

\item For each $n\in\mathbb{N}$, the optimization algorithm used
in solving problem (\ref{eq: V_n_HAT approx}) attains the minimum
$\left(\hat{\boldsymbol{\theta}}_{n}^{\ast},\hat{\xi}_{n}^{\ast}\right)\in\mathbb{R}^{\nu_{n}+1}$
corresponding to a given training data set $\mathcal{Y}_{n}$.

\end{enumerate}

\end{assumption}

Since stochastic gradient descent (SGD) is used in training the NN,
Assumption \ref{assu: Appendix Convergence assumptions - computational}(iii)
is very strong; however, it is a standard assumption in convergence
analyses in the literature (see for example \cite{HureEtAl2021,TsangWong2020})
in order to focus on the key aspects of a proposed approach. For detailed
treatments of theoretical aspects regarding optimization errors (i.e.
the differences between the attained values and the true minima) arising
when training NNs, the reader is referred to for example \cite{BeckEtAl2022,JentzenEtAl2021}.
Note that Assumption \ref{assu: Appendix Convergence assumptions - computational}(ii),
which can also be found in \cite{TsangWong2020}, is used to establish
a version of the law of large numbers that is applicable to our setting.

{\color{black}
\begin{brem}[Increase in number of training samples as $\hbar$ increases]
Informally, Assumption \ref{assu: Appendix Convergence assumptions - computational}(ii)
requires that the number of training samples $n$ grows  faster than
$O(\hbar^4)$ for $\mathcal{L}^{h}=1$ and $O(\hbar^6)$ for $\mathcal{L}^{h}>1$, where $\hbar$ is the number of nodes in each hidden layer.
Since we require a large $\hbar$ (number of nodes in each layer) for good function approximation, this would suggest
that convergence in terms of both function approximation and sampling error
requires a very large number of sample paths.
This would appear to result in a barrier to obtaining accurate results, for practical numbers of samples.  However, our numerical examples seem to 
produce solutions with reasonable errors, hence the requirements of
Assumption \ref{assu: Appendix Convergence assumptions - computational}(ii)
are probably not sharp.
Regardless, we can certainly expect that the number of samples should be significantly
increased as we increase $\hbar$.
\end{brem}
}

\subsection{Proof of Theorem \ref{thm: Validity of NN approximation}\label{subsec:Appendix Proof-of-Theorem NN approx}}

Before presenting the proof of Theorem \ref{thm: Validity of NN approximation},
we first prove some auxiliary results that are preliminary requirements
for the proof. 

We start with Lemma \ref{lem: Appendix Universal approximation},
which combines and applies selected universal approximation results
to our setting. The use of the notation $\boldsymbol{f}_{n}^{\ast}\left(\cdot,\boldsymbol{\theta}_{n}^{\ast}\right)\in\mathcal{N}_{n}$
in Lemma \ref{lem: Appendix Universal approximation}, which has been
defined in Subsection \ref{subsec:Step 1 NN approx} as the NN using
the optimal parameter vector consistent with problem (\ref{eq: V_n in terms of N_n})-(\ref{eq: V_n in terms of theta_n}),
will be become clear in the subsequent results.
\begin{lem}
\label{lem: Appendix Universal approximation}(Convergence to optimal
control) Suppose that Assumption \ref{assu: Appendix Convergence assumptions - critical}
and Assumption \ref{assu: Convergence assumptions - convenience}
hold. As per (\ref{eq: Assumption CONTINUOUS control}), let $\boldsymbol{p}^{\ast}=\left(p_{i}^{\ast}:i=1,...,N_{a}\right)\in C\left(\mathcal{D}_{\boldsymbol{\phi}},\mathcal{Z}\right)$
denote the optimal control associated with problem (\ref{eq: Original problem - CONTINUOUS control}).
Then there exists a sequence of neural networks, $\boldsymbol{f}_{n}^{\ast}\in\mathcal{N}_{n}$,
$n\in\mathbb{N}$, where each $\boldsymbol{f}_{n}^{\ast}=\left(f_{n,i}^{\ast}:i=1,...,N_{a}\right)$
has parameter vector $\boldsymbol{\theta}_{n}^{\ast}\in\mathbb{R}^{\nu_{n}}$,such
that 
\begin{eqnarray}
\lim_{n\rightarrow\infty}\sup_{\boldsymbol{\phi}\in\mathcal{D}_{\boldsymbol{\phi}}}\left|f_{n,i}^{\ast}\left(\boldsymbol{\phi};\boldsymbol{\theta}_{n}^{\ast}\right)-p_{i}^{\ast}\left(\boldsymbol{\phi}\right)\right| & = & 0,\qquad\forall i=1,...,N_{a},\label{eq: UAT optimal control}
\end{eqnarray}
\end{lem}

\begin{proof}
For ease of reference, recall that we have defined $\mathcal{N}_{n}$
as the set of NNs with $\hbar\left(n\right)$ hidden nodes in each
of the (fixed number of) $\mathcal{L}^{h}\geq1$ hidden layers, constructed
according to Assumption \ref{subsec: Appendix NN-structural-assumptions},
\begin{equation}
\mathcal{N}_{n}=\left\{ \left.\boldsymbol{f}_{n}:\mathcal{D}_{\boldsymbol{\phi}}\rightarrow\mathcal{Z}\right|\boldsymbol{f}_{n}\left(\cdot;\boldsymbol{\theta}_{n}\right)\textrm{ satisfies Assumption }\ref{assu: Appendix NN structure assumptions}\textrm{ with }\hbar\left(n\right)\textrm{ nodes in each hidden layer}\right\} .\label{eq: Appendix Definition set N_n}
\end{equation}

Consider another sequence of NNs, $\overset{\circ}{\boldsymbol{f}}_{n},n\in\mathbb{N}$,
where each $\overset{\circ}{\boldsymbol{f}}_{n}:\mathcal{D}_{\boldsymbol{\phi}}\rightarrow\mathbb{R}^{N_{a}}$
is structurally identical to the corresponding $\boldsymbol{f}_{n}\in\mathcal{N}_{n}$
in terms of Assumption \ref{assu: Appendix NN structure assumptions},
\textit{except} that $\overset{\circ}{\boldsymbol{f}}_{n}$ uses the
identity as the (linear) output activation function. Specifically,
we assume that $\overset{\circ}{\boldsymbol{f}}_{n}$ does not apply
the activation (\ref{assu: Appendix NN structure assumptions}) at
its output layer, but instead replaces (\ref{assu: Appendix NN structure assumptions})
with $\overset{\circ}{\boldsymbol{\mathfrak{a}}}_{n}^{\left[\mathcal{L}\right]}=\left(\overset{\circ}{\mathfrak{a}}_{n,i}^{\left[\mathcal{L}\right]}:i=1,...,N_{a}\right):\mathbb{R}^{N_{a}}\rightarrow\mathbb{R}^{N_{a}}$
where
\begin{eqnarray}
\overset{\circ}{\mathfrak{a}}_{n,i}^{\left[\mathcal{L}\right]}\left(\boldsymbol{z}_{n}^{\left[\mathcal{L}\right]}\right) & = & z_{n,i}^{\left[\mathcal{L}\right]}=\sum_{k=1}^{N_{a}}x_{n,ik}^{\left[\mathcal{L}\right]}\mathfrak{a}_{n,k}^{\left[\mathcal{L}-1\right]}+b_{n,i}^{\left[\mathcal{L}\right]},\qquad\forall i=1,...,N_{a}.\label{eq: Linear identity output activation}
\end{eqnarray}
For any given $n\in\mathbb{N}$, the relationship between $\boldsymbol{f}_{n}$
and $\overset{\circ}{\boldsymbol{f}}_{n}$ are illustrated in Figure
\ref{fig: NN diagram for proof}. Note that the entire parameter vector
$\boldsymbol{\theta}_{n}$ of $\boldsymbol{f}_{n}$ is inherited by
$\overset{\circ}{\boldsymbol{f}}_{n}$, since all the weights, biases,
and hidden layers and nodes of $\overset{\circ}{\boldsymbol{f}}_{n}$
and $\boldsymbol{f}_{n}$ are identical. As a result, we define the
set $\overset{\circ}{\mathcal{N}}_{n}$
\begin{eqnarray}
\overset{\circ}{\mathcal{N}}_{n} & = & \left\{ \left.\overset{\circ}{\boldsymbol{f}}_{n}:\mathcal{D}_{\boldsymbol{\phi}}\rightarrow\mathbb{R}^{N_{a}}\right|\overset{\circ}{\boldsymbol{f}}_{n}\left(\cdot;\boldsymbol{\theta}_{n}\right)\textrm{ satisfies Assumption }\ref{assu: Appendix NN structure assumptions},\textrm{ except }\right.\nonumber \\
 &  & \left.\qquad\qquad\qquad\qquad\textrm{ output activation \eqref{eq: Softmax output} is replaced by }\eqref{eq: Linear identity output activation}.\right\} ,\label{eq: Appendix definition set N_n_circle}
\end{eqnarray}
where we note that the outputs of $\overset{\circ}{\boldsymbol{f}}_{n}$
take values which are no longer in $\mathcal{Z}\subset\mathbb{R}^{N_{a}}$,
but instead merely in $\mathbb{R}^{N_{a}}$. The main benefit of working
with $\overset{\circ}{\boldsymbol{f}}_{n}\in\overset{\circ}{\mathcal{N}}_{n}$
instead of $\boldsymbol{f}_{n}\in\mathcal{N}_{n}$, is that the linear
output layer (\ref{eq: Linear identity output activation}) means
that each $\overset{\circ}{\boldsymbol{f}}_{n}\in\overset{\circ}{\mathcal{N}}_{n}$
is in the standard form used by most universal approximation theorems
for NNs (see for example \cite{LeshnoEtAl1993,Funahashi1989,HornikEtAl1989,Hornik1991}).

Recalling for convenience the definition of the softmax function $\boldsymbol{\psi}=\left(\psi_{i}:i=1,..,N_{a}\right):\mathbb{R}^{N_{a}}\rightarrow\mathbb{R}^{N_{a}}$
in (\ref{eq: Softmax output}),
\begin{eqnarray}
\psi_{i}\left(\boldsymbol{y}\right) & = & \frac{\exp\left\{ y_{i}\right\} }{\sum_{i=1}^{N_{a}}\exp\left\{ y_{j}\right\} },\qquad\forall\boldsymbol{y}=\left(y_{i}:i=1,..,N_{a}\right)\in\mathbb{R}^{N_{a}},\label{eq: Appendix softmax}
\end{eqnarray}
we therefore observe that for any $n\in\mathbb{N}$, the NN $\boldsymbol{f}_{n}\left(\cdot;\boldsymbol{\theta}_{n}\right)\in\mathcal{N}_{n}$
can be expressed as a transformation of the corresponding NN $\overset{\circ}{\boldsymbol{f}}_{n}\left(\cdot;\boldsymbol{\theta}_{n}\right)\in\overset{\circ}{\mathcal{N}}_{n}$,
provided both NNs use the same parameter vector $\boldsymbol{\theta}_{n}\in\mathbb{R}^{\nu_{n}}$:
\begin{align}
\boldsymbol{f}_{n}\left(\cdot;\boldsymbol{\theta}_{n}\right)= & \boldsymbol{\psi}\circ\overset{\circ}{\boldsymbol{f}}_{n}\left(\cdot;\boldsymbol{\theta}_{n}\right),\qquad\textrm{where }\overset{\circ}{\boldsymbol{f}}_{n}\left(\cdot;\boldsymbol{\theta}_{n}\right)\in\overset{\circ}{\mathcal{N}}_{n}.\label{eq: f_n in terms of f_n_circle}
\end{align}

As per Assumption \ref{assu: Appendix NN structure assumptions},
recall that $\hbar\left(n\right),n\in\mathbb{N}$ satisfies $\hbar\left(n\right)<\hbar\left(n+1\right),\forall n\in\mathbb{N}$
such that $\lim_{n\rightarrow\infty}\hbar\left(n\right)=\infty$.
Inspired by the notation of \cite{Hornik1991}, we define the sets
$\mathcal{N}_{\infty}$ and $\mathcal{\overset{\circ}{\mathcal{N}}}_{\infty}$
as the sets of NNs constructed according to (\ref{eq: Appendix Definition set N_n})
and (\ref{eq: Appendix definition set N_n_circle}), respectively,
but with an arbitrarily large number of hidden nodes, 
\begin{equation}
\mathcal{N}_{\infty}=\bigcup_{n\in\mathbb{N}}\mathcal{N}_{n},\qquad\textrm{and }\qquad\overset{\circ}{\mathcal{N}}_{\infty}=\bigcup_{n\in\mathbb{N}}\mathcal{\overset{\circ}{\mathcal{N}}}_{n}.\label{eq: N_inf and N_inf_circle}
\end{equation}

Since $\mathcal{D}_{\boldsymbol{\phi}}\subset\mathbb{R}^{\eta_{X}+1}$
is compact by (\ref{eq: D_phi bounded domain}) as per Assumption
\ref{assu: Convergence assumptions - convenience} (note that this
requirement can be relaxed without difficulty as discussed in Remark
\ref{rem: Appendix Relaxing Convergence assumptions - convenience}),
we know by the results of \cite{Hornik1991,HornikEtAl1989} that $\overset{\circ}{\mathcal{N}}_{\infty}$
is uniformly dense in $C\left(\mathcal{D}_{\boldsymbol{\phi}},\mathbb{R}^{N_{a}}\right)$.
In other words, for any function $\overset{\circ}{\boldsymbol{g}}=\left(\overset{\circ}{g}_{i}:i=1,...,N_{a}\right)\in C\left(\mathcal{D}_{\boldsymbol{\phi}},\mathbb{R}^{N_{a}}\right)$
and any $\epsilon>0$, there exists a value of $n=n_{\epsilon}$ sufficiently
large such that the corresponding NN $\overset{\circ}{\boldsymbol{f}}_{n_{\epsilon}}=\left(\overset{\circ}{f}_{n_{\epsilon},i}:i=1,...,N_{a}\right)\in\overset{\circ}{\mathcal{N}}_{n_{\epsilon}}$
such that
\begin{eqnarray}
\sup_{\boldsymbol{\phi}\in\mathcal{D}_{\boldsymbol{\phi}}}\left|\overset{\circ}{f}_{n_{\epsilon},i}\left(\boldsymbol{\phi};\boldsymbol{\theta}_{n\left(\epsilon\right)}\right)-\overset{\circ}{g}_{i}\left(\boldsymbol{\phi}\right)\right| & < & \epsilon,\qquad\forall i=1,...,N_{a}.\label{eq: UAT result Hornik}
\end{eqnarray}
Note that (\ref{eq: UAT result Hornik}) holds for any given number
$\mathcal{L}^{h}\geq1$ of hidden layers (see for example Corollary
2.7 in \cite{HornikEtAl1989}).

Using the results of \cite{GaoPavel2018} , the softmax (\ref{eq: Appendix softmax})
is (Lipschitz) continuous and surjective, since $\psi_{i}\left(\boldsymbol{y}\right)=\psi_{i}\left(\boldsymbol{y}+c\right)$
for any $\boldsymbol{y}\in\mathbb{R}^{N_{a}}$ and $c\in\mathbb{R}$,
where $\boldsymbol{y}+c\coloneqq\left(y_{i}+c:i=1,..,N_{a}\right)$.
In addition, it has a continuous right-inverse; as an example, we
can simply consider the function $\overleftarrow{\boldsymbol{\psi}}\left(\boldsymbol{z}\right)=\left(\log\left(z_{i}\right):i=1,...,N_{a}\right)$
where each $z_{i}\in\left(0,1\right)$ and that $\sum_{i}z_{i}=1$.
Furthermore, by Assumption \ref{assu: Appendix NN structure assumptions},
no activation function is applied at the input layer (i.e. the ``input
activation'' is trivially injective and continuous). Using these
properties of the input and output layers of any $\boldsymbol{f}_{n}\in\mathcal{N}_{n}$
together with the results (\ref{eq: f_n in terms of f_n_circle})
and (\ref{eq: UAT result Hornik}), we can conclude by the results
of \cite{KratsiosBilokopytov2020} that the set $\mathcal{N}_{\infty}$
is uniformly dense in $C\left(\mathcal{D}_{\boldsymbol{\phi}},\mathcal{Z}\right)$.

Applying this result specifically to the optimal control $\boldsymbol{p}^{\ast}\in C\left(\mathcal{D}_{\boldsymbol{\phi}},\mathcal{Z}\right)$
as per Assumption \ref{assu: Existence and Continuity of the control},
we can conclude that, for any $\epsilon>0$, there exists a value $n=n_{\epsilon}$
sufficiently large such that the corresponding NN $\boldsymbol{f}_{n_{\epsilon}}^{\ast}\left(\cdot;\boldsymbol{\theta}_{n_{\epsilon}}^{\ast}\right)\in\mathcal{N}_{n_{\epsilon}}$
satisfies
\begin{eqnarray}
\sup_{\boldsymbol{\phi}\in\mathcal{D}_{\boldsymbol{\phi}}}\left|f_{n_{\epsilon},i}^{\ast}\left(\boldsymbol{\phi};\boldsymbol{\theta}_{n_{\epsilon}}^{\ast}\right)-p_{i}^{\ast}\left(\boldsymbol{\phi}\right)\right| & < & \epsilon,\qquad\forall i=1,...,N_{a}.\label{eq: NN in N_inf approx optimal control}
\end{eqnarray}
Note that the exact output of the NN $\boldsymbol{f}_{n_{\epsilon}}^{\ast}\in\mathcal{N}_{n_{\epsilon}}$,
which we recall has $\hbar\left(n_{\epsilon}\right)$ hidden nodes
in each hidden layer, can be attained by a NN with $\hbar\left(n_{\epsilon}+k\right)$,
$k\in\mathbb{N}$ hidden nodes, since we can always set the weights
and biases corresponding to the additional $\hbar\left(n_{\epsilon}+k\right)-\hbar\left(n_{\epsilon}\right)$
nodes identically to zero. In other words, (\ref{eq: NN in N_inf approx optimal control})
implies the existence of a sequence of NNs $\boldsymbol{f}_{n}^{\ast}\left(\cdot;\boldsymbol{\theta}_{n}^{\ast}\right),n\in\mathbb{N}$,
where each $\boldsymbol{f}_{n}^{\ast}\left(\cdot;\boldsymbol{\theta}_{n}^{\ast}\right)\in\mathcal{N}_{n}$,
such that for any $\epsilon>0$ and sufficiently large $n_{\epsilon}\in\mathbb{N}$,
we have
\begin{eqnarray}
\sup_{\boldsymbol{\phi}\in\mathcal{D}_{\boldsymbol{\phi}}}\left|f_{n,i}^{\ast}\left(\boldsymbol{\phi};\boldsymbol{\theta}_{n}^{\ast}\right)-p_{i}^{\ast}\left(\boldsymbol{\phi}\right)\right| & < & \epsilon,\qquad\forall n\geq n_{\epsilon},\;i=1,...,N_{a},\label{eq: UAT optimal control proof}
\end{eqnarray}
completing the proof of (\ref{eq: UAT optimal control}).
\end{proof}
\noindent 
\begin{figure}[!tbh]
\noindent \begin{centering}
\includegraphics[scale=0.45]{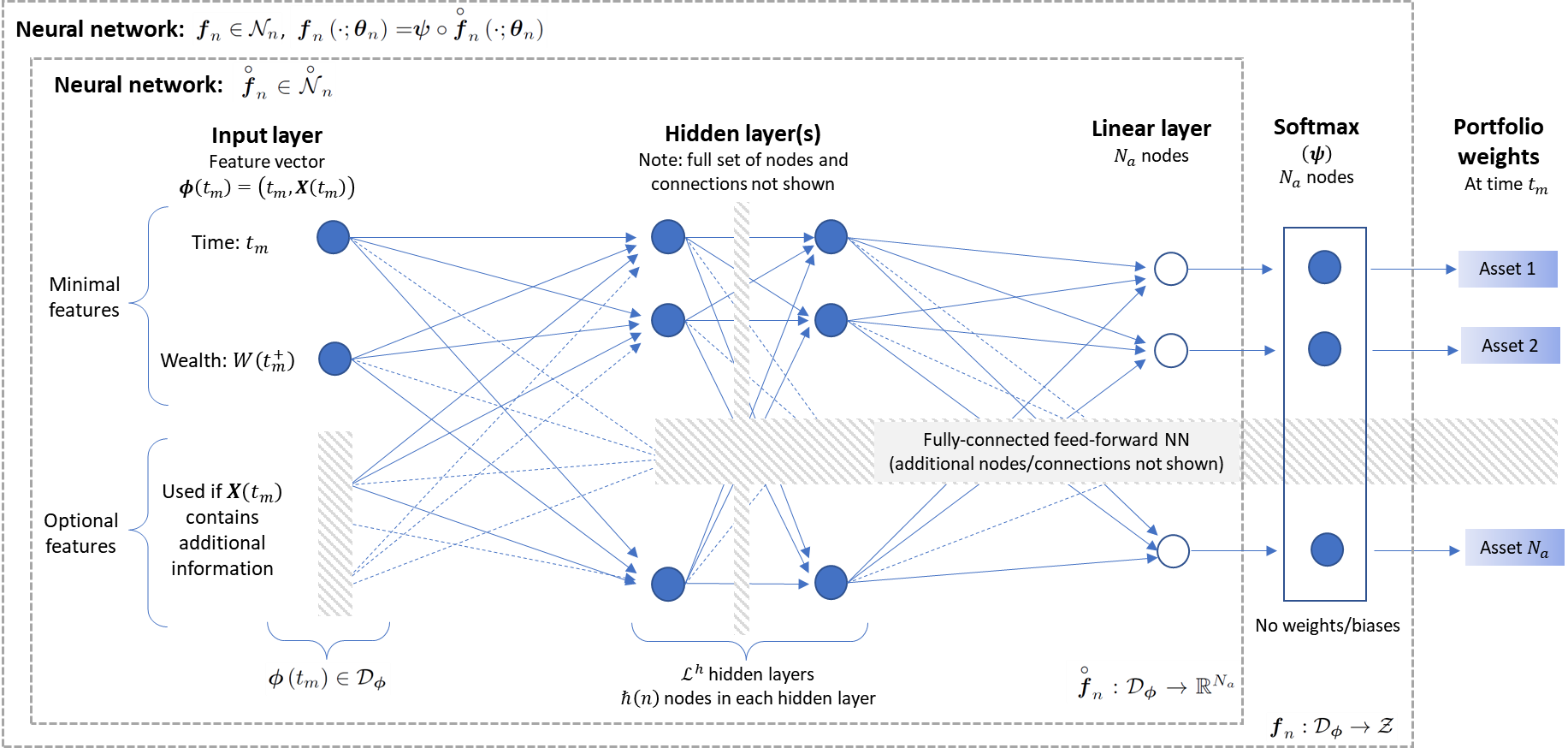}
\par\end{centering}
\caption{Illustration of the interpretation of the NN $\boldsymbol{f}_{n}\left(\cdot;\boldsymbol{\theta}_{n}\right)$
as a composition of the softmax $\boldsymbol{\psi}$ and the NN $\protect\overset{\circ}{\boldsymbol{f}}_{n}\left(\cdot;\boldsymbol{\theta}_{n}\right)$
as per equation (\ref{eq: f_n in terms of f_n_circle}). \label{fig: NN diagram for proof}}
\end{figure}

If Assumption \ref{eq: Assumption CONTINUOUS control} and Assumption
\ref{assu: Convergence assumptions - convenience} are applicable,
the wealth dynamics (\ref{eq: W dynamics - continuous control}) using
the optimal control is given by 
\begin{eqnarray}
W^{\ast}\left(t_{m+1}^{-};\boldsymbol{p}^{\ast},\boldsymbol{Y}\right) & = & W^{\ast}\left(t_{m}^{+};\boldsymbol{p}^{\ast},\boldsymbol{Y}\right)\cdot\sum_{i=1}^{N_{a}}p_{i}^{\ast}\left(t_{m},W^{\ast}\left(t_{m}^{+};\boldsymbol{p}^{\ast},\boldsymbol{Y}\right)\right)\cdot Y_{i}\left(t_{m}\right),\qquad t_{m}\in\mathcal{T},\label{eq: Appendix W dynamics using p*}
\end{eqnarray}
where we recall that $W^{\ast}\left(t_{m}^{+};\boldsymbol{p}^{\ast},\boldsymbol{Y}\right)=W^{\ast}\left(t_{m}^{-};\boldsymbol{p}^{\ast},\boldsymbol{Y}\right)+q\left(t_{m}\right)$,
$\boldsymbol{X}^{\ast}\left(t_{m}\right)=W^{\ast}\left(t_{m}^{+};\boldsymbol{p}^{\ast},\boldsymbol{Y}\right)$
and $W^{\ast}\left(t_{N_{rb}}^{-}\right)\coloneqq W^{\ast}\left(T\right)$.
Furthermore, associated with every NN in the sequence $\boldsymbol{f}_{n}^{\ast}\left(\cdot,\boldsymbol{\theta}_{n}^{\ast}\right)\in\mathcal{N}_{n}$
identified in Lemma \ref{lem: Appendix Universal approximation},
we have the corresponding wealth dynamics as per (\ref{eq: W dynamics with NN as control})
that satisfies
\begin{align}
W^{\ast}\left(t_{m+1}^{-};\boldsymbol{\theta}_{n}^{\ast},\boldsymbol{Y}\right) & =W^{\ast}\left(t_{m}^{+};\boldsymbol{\theta}_{n}^{\ast},\boldsymbol{Y}\right)\cdot\sum_{i=1}^{N_{a}}f_{n,i}^{\ast}\left(t_{m},W^{\ast}\left(t_{m}^{+};\boldsymbol{\theta}_{n}^{\ast},\boldsymbol{Y}\right);\boldsymbol{\theta}_{n}^{\ast}\right)\cdot Y_{i}\left(t_{m}\right),\qquad t_{m}\in\mathcal{T},n\in\mathbb{N}.\label{eq: Appendix W dynamics using f*}
\end{align}

The following lemma justifies the use of the notation $W^{\ast}$
in the wealth dynamics (\ref{eq: Appendix W dynamics using f*}).
\begin{lem}
\label{lem: Appendix W dynamics converges}(Convergence to optimal
wealth) Suppose that Assumption \ref{assu: Appendix Convergence assumptions - critical}
and Assumption \ref{assu: Convergence assumptions - convenience}
hold. Let $\boldsymbol{f}_{n}^{\ast}\left(\cdot,\boldsymbol{\theta}_{n}^{\ast}\right)\in\mathcal{N}_{n}$
be the sequence identified in Lemma \ref{lem: Appendix Universal approximation}
such that (\ref{eq: UAT optimal control}) holds. Then the wealth
dynamics $W^{\ast}\left(t;\boldsymbol{\theta}_{n}^{\ast},\boldsymbol{Y}\right)$
associated with each $\boldsymbol{f}_{n}^{\ast}$, obtained as per
(\ref{eq: Appendix W dynamics using f*}), converges to the true optimal
wealth dynamics $W^{\ast}\left(t;\boldsymbol{p}^{\ast},\boldsymbol{Y}\right)$
as $n\rightarrow\infty$ almost surely. In more detail, we have 
\begin{eqnarray}
\lim_{n\rightarrow\infty}W^{\ast}\left(t_{m}^{-};\boldsymbol{\theta}_{n}^{\ast},\boldsymbol{Y}\right) & = & W^{\ast}\left(t_{m}^{-};\boldsymbol{p}^{\ast},\boldsymbol{Y}\right)\quad\textrm{a.s.},\qquad\forall t_{m}\in\mathcal{T},\label{eq: Appendix W_t_m converges}
\end{eqnarray}
and
\begin{eqnarray}
\lim_{n\rightarrow\infty}W^{\ast}\left(T;\boldsymbol{\theta}_{n}^{\ast},\boldsymbol{Y}\right) & = & W^{\ast}\left(T;\boldsymbol{p}^{\ast},\boldsymbol{Y}\right)\quad\textrm{a.s.}\label{eq: Appendix W_T converges}
\end{eqnarray}
\end{lem}

\begin{proof}
Note that (\ref{eq: Appendix W_T converges}) is stated separately
since the terminal time $T$ is not a rebalancing time (see (\ref{eq: Set of rebalancing times}))
and the terminal wealth is critical in the evaluation of the objective
functional. 

At the start of the time horizon $\left[t_{0},T\right]$, we are given
the initial wealth $W\left(t_{0}^{-}\right)=w_{0}>0$. Therefore,
at the first rebalancing time $t_{0}\in\mathcal{T}$, the wealth available
for investment does not depend on the control, so that 
\begin{equation}
w_{0}^{+}\coloneqq w_{0}+q\left(t_{0}\right)=W^{\ast}\left(t_{0}^{+};\boldsymbol{\theta}_{n}^{\ast},\boldsymbol{Y}\right)=W^{\ast}\left(t_{0}^{+};\boldsymbol{p}^{\ast},\boldsymbol{Y}\right),\qquad\forall n\in\mathbb{N}.\label{eq: Appendix w_0 plus}
\end{equation}
Using dynamics (\ref{eq: Appendix W dynamics using p*}) and (\ref{eq: Appendix W dynamics using f*})
to compare the wealth at time $t_{0}+\Delta t=t_{1}\in\mathcal{T}$,
we have
\begin{eqnarray}
\lim_{n\rightarrow\infty}W^{\ast}\left(t_{1}^{-};\boldsymbol{\theta}_{n}^{\ast},\boldsymbol{Y}\right)-W^{\ast}\left(t_{1}^{-};\boldsymbol{p}^{\ast},\boldsymbol{Y}\right) & = & w_{0}^{+}\cdot\sum_{i=1}^{N_{a}}\left[\lim_{n\rightarrow\infty}f_{n,i}^{\ast}\left(t_{0},w_{0}^{+};\boldsymbol{\theta}_{n}^{\ast}\right)-p_{i}^{\ast}\left(t_{0},w_{0}^{+}\right)\right]\cdot Y_{i}\left(t_{0}\right)\nonumber \\
 & = & 0\qquad\textrm{a.s.},\label{eq: Appendix W t_1 converges}
\end{eqnarray}
which follows from Lemma \ref{lem: Appendix Universal approximation}
and the fact that $Y_{i}\left(t_{0}\right)<\infty$ a.s. by assumption
(see definition (\ref{eq: Path Y of asset returns})). 

For purposes of induction, assume that at some $t_{m}\in\mathcal{T}$,
we have
\begin{eqnarray}
\lim_{n\rightarrow\infty}W^{\ast}\left(t_{m}^{-};\boldsymbol{\theta}_{n}^{\ast},\boldsymbol{Y}\right) & = & W^{\ast}\left(t_{m}^{-};\boldsymbol{p}^{\ast},\boldsymbol{Y}\right)\qquad\textrm{a.s.}\label{eq: Appendix W_t_m *maybe* convergence}
\end{eqnarray}
If (\ref{eq: Appendix W_t_m *maybe* convergence}) holds, then we
have 
\begin{equation}
\lim_{n\rightarrow\infty}W^{\ast}\left(t_{m}^{+};\boldsymbol{\theta}_{n}^{\ast},\boldsymbol{Y}\right)=\lim_{n\rightarrow\infty}\left[W^{\ast}\left(t_{m}^{-};\boldsymbol{\theta}_{n}^{\ast},\boldsymbol{Y}\right)+q\left(t_{m}\right)\right]=W^{\ast}\left(t_{m}^{+};\boldsymbol{p}^{\ast},\boldsymbol{Y}\right)\qquad\textrm{a.s.},\label{eq: Appendix W_t_m_plus convergence}
\end{equation}
as well as
\begin{eqnarray}
 &  & \lim_{n\rightarrow\infty}\left|f_{n,i}^{\ast}\left(t_{m},W^{\ast}\left(t_{m}^{+};\boldsymbol{\theta}_{n}^{\ast},\boldsymbol{Y}\right);\boldsymbol{\theta}_{n}^{\ast}\right)-p_{i}^{\ast}\left(t_{m},W^{\ast}\left(t_{m}^{+};\boldsymbol{p}^{\ast},\boldsymbol{Y}\right)\right)\right|\nonumber \\
 & \leq & \lim_{n\rightarrow\infty}\left|f_{n,i}^{\ast}\left(t_{m},W^{\ast}\left(t_{m}^{+};\boldsymbol{\theta}_{n}^{\ast},\boldsymbol{Y}\right);\boldsymbol{\theta}_{n}^{\ast}\right)-p_{i}^{\ast}\left(t_{m},W^{\ast}\left(t_{m}^{+};\boldsymbol{\theta}_{n}^{\ast},\boldsymbol{Y}\right)\right)\right|\nonumber \\
 &  & +\lim_{n\rightarrow\infty}\left|p_{i}^{\ast}\left(t_{m},W^{\ast}\left(t_{m}^{+};\boldsymbol{\theta}_{n}^{\ast},\boldsymbol{Y}\right)\right)-p_{i}^{\ast}\left(t_{m},W^{\ast}\left(t_{m}^{+};\boldsymbol{p}^{\ast},\boldsymbol{Y}\right)\right)\right|\nonumber \\
 & \leq & \lim_{n\rightarrow\infty}\sup_{\boldsymbol{\phi}\in\mathcal{D}_{\boldsymbol{\phi}}}\left|f_{n,i}^{\ast}\left(\boldsymbol{\phi};\boldsymbol{\theta}_{n}^{\ast}\right)-p_{i}^{\ast}\left(\boldsymbol{\phi}\right)\right|\nonumber \\
 & = & 0\qquad\textrm{a.s.},\qquad\forall i=1,...,N_{a},\label{eq: Appendix convergence f_n to p with optimal wealth}
\end{eqnarray}
which is a consequence of Lemma \ref{lem: Appendix Universal approximation}
and the continuity of $\boldsymbol{p}^{\ast}\in C\left(\mathcal{D}_{\boldsymbol{\phi}},\mathcal{Z}\right)$.
From (\ref{eq: Appendix W_t_m_plus convergence}) and (\ref{eq: Appendix convergence f_n to p with optimal wealth}),
we therefore conclude that

\begin{align}
 & \lim_{n\rightarrow\infty}\left|W^{\ast}\left(t_{m}^{+};\boldsymbol{\theta}_{n}^{\ast},\boldsymbol{Y}\right)\cdot f_{n,i}^{\ast}\left(t_{m},W^{\ast}\left(t_{m}^{+};\boldsymbol{\theta}_{n}^{\ast},\boldsymbol{Y}\right);\boldsymbol{\theta}_{n}^{\ast}\right)-W^{\ast}\left(t_{m}^{+};\boldsymbol{p}^{\ast},\boldsymbol{Y}\right)\cdot p_{i}^{\ast}\left(t_{m},W^{\ast}\left(t_{m}^{+};\boldsymbol{p}^{\ast},\boldsymbol{Y}\right)\right)\right|\nonumber \\
= & 0\qquad\textrm{a.s.},\qquad\forall i=1,...,N_{a}.\label{eq: Appendix convergence wealth penultimate step}
\end{align}

Using dynamics (\ref{eq: Appendix W dynamics using p*}) and (\ref{eq: Appendix W dynamics using f*})
to compare the wealth at time $t_{m}+\Delta t=t_{m+1}$, we have
\begin{eqnarray}
\lim_{n\rightarrow\infty}W^{\ast}\left(t_{m+1}^{-};\boldsymbol{\theta}_{n}^{\ast},\boldsymbol{Y}\right)-W^{\ast}\left(t_{m+1}^{-};\boldsymbol{p}^{\ast},\boldsymbol{Y}\right) & = & \lim_{n\rightarrow\infty}W^{\ast}\left(t_{m}^{+};\boldsymbol{\theta}_{n}^{\ast},\boldsymbol{Y}\right)\cdot\sum_{i=1}^{N_{a}}f_{n,i}^{\ast}\left(t_{m},W^{\ast}\left(t_{m}^{+};\boldsymbol{\theta}_{n}^{\ast},\boldsymbol{Y}\right);\boldsymbol{\theta}_{n}^{\ast}\right)\cdot Y_{i}\left(t_{m}\right)\nonumber \\
 &  & -W^{\ast}\left(t_{m}^{+};\boldsymbol{p}^{\ast},\boldsymbol{Y}\right)\cdot\sum_{i=1}^{N_{a}}p_{i}^{\ast}\left(t_{m},W^{\ast}\left(t_{m}^{+};\boldsymbol{p}^{\ast},\boldsymbol{Y}\right)\right)\cdot Y_{i}\left(t_{m}\right)\nonumber \\
 & = & 0\qquad\textrm{a.s.},\label{eq: Appendix convergence of wealth ultimate step}
\end{eqnarray}
which follows from (\ref{eq: Appendix convergence wealth penultimate step})
and $Y_{i}\left(t_{m}\right)<\infty$ a.s. By induction, we therefore
conclude that (\ref{eq: Appendix W_t_m converges}) holds if $t_{m+1}\in\mathcal{T}$
(i.e. if $m<N_{rb}-1$), and (\ref{eq: Appendix W_T converges}) holds
in the case where $t_{m+1}=t_{N_{rb}}=T$ (i.e. $m=N_{rb}-1$).
\end{proof}
The following lemma establishes the convergence of the sequence of
objective functionals using the NN approximations identified in Lemma
\ref{lem: Appendix Universal approximation}. 
\begin{lem}
\label{lem: Appendix J converges}(Convergence of objective functionals)
Suppose that Assumption \ref{assu: Appendix Convergence assumptions - critical}
and Assumption \ref{assu: Convergence assumptions - convenience}
hold. Let $\boldsymbol{f}_{n}^{\ast}\left(\cdot,\boldsymbol{\theta}_{n}^{\ast}\right)\in\mathcal{N}_{n}$
be the sequence identified in Lemma \ref{lem: Appendix Universal approximation}
such that (\ref{eq: UAT optimal control}) holds. Then
\begin{eqnarray}
\lim_{n\rightarrow\infty}J_{n}\left(\boldsymbol{\theta}_{n}^{\ast},\xi;t_{0},w_{0}\right) & = & J\left(\boldsymbol{p}^{\ast},\xi;t_{0},w_{0}\right),\qquad\forall\xi\in\mathbb{R},\label{eq: Appendix J converges}
\end{eqnarray}
where $J_{n}$ is defined in (\ref{eq: J_n}), and $J$ is defined
in (\ref{eq: Objective functional with continuous control p}).
\end{lem}

\begin{proof}
Let $\xi\in\mathbb{R}$ be arbitrary. By Lemma \ref{lem: Appendix W dynamics converges}
and the continuity of $F$, we have 
\begin{eqnarray}
\lim_{n\rightarrow\infty}F\left(W^{\ast}\left(T;\boldsymbol{\theta}_{n}^{\ast},\boldsymbol{Y}\right),\xi\right) & = & F\left(W^{\ast}\left(T;\boldsymbol{p}^{\ast},\boldsymbol{Y}\right),\xi\right)\qquad\textrm{a.s}.\label{eq: Appendix F converges}
\end{eqnarray}
Therefore, by using the boundedness of wealth as per Assumption \ref{assu: Convergence assumptions - convenience},
the dominated convergence theorem gives
\begin{eqnarray}
\lim_{n\rightarrow\infty}E^{t_{0},w_{0}}\left[F\left(W^{\ast}\left(T;\boldsymbol{\theta}_{n}^{\ast},\boldsymbol{Y}\right),\xi\right)\right] & = & E^{t_{0},w_{0}}\left[F\left(W^{\ast}\left(T;\boldsymbol{p}^{\ast},\boldsymbol{Y}\right),\xi\right)\right].\label{eq: Appendix expectation F converges}
\end{eqnarray}
Similarly, by the continuity of $G$, Lemma \ref{lem: Appendix W dynamics converges},
the boundedness of wealth and the dominated convergence theorem, we
have
\begin{eqnarray}
 &  & \lim_{n\rightarrow\infty}E^{t_{0},w_{0}}\left[G\left(W^{\ast}\left(T;\boldsymbol{\theta}_{n}^{\ast},\boldsymbol{Y}\right),E^{t_{0},w_{0}}\left[W^{\ast}\left(T;\boldsymbol{\theta}_{n}^{\ast},\boldsymbol{Y}\right)\right],w_{0},\xi\right)\right]\nonumber \\
 & = & E^{t_{0},w_{0}}\left[G\left(W^{\ast}\left(T;\boldsymbol{p}^{\ast},\boldsymbol{Y}\right),E^{t_{0},w_{0}}\left[W^{\ast}\left(T;\boldsymbol{p}^{\ast},\boldsymbol{Y}\right)\right],w_{0},\xi\right)\right].\label{eq: Appendix G converges}
\end{eqnarray}
Finally, using the definitions of $J$ in (\ref{eq: Objective functional with continuous control p})
and $J_{n}$ in (\ref{eq: J_n}), we combine (\ref{eq: Appendix expectation F converges})
and (\ref{eq: Appendix G converges}) to conclude (\ref{eq: Appendix J converges}).
\end{proof}

\subsubsection*{Proof of Theorem \ref{thm: Validity of NN approximation}}

Using the preceding results, we are finally in the position to prove
Theorem \ref{thm: Validity of NN approximation}. Note that this proof
also motivates the use of the notation $\boldsymbol{f}_{n}^{\ast}\left(\cdot,\boldsymbol{\theta}_{n}^{\ast}\right)$
and its associated wealth $W^{\ast}\left(T;\boldsymbol{\theta}_{n}^{\ast},\boldsymbol{Y}\right)$
for the sequence of NNs identified in Lemma \ref{lem: Appendix Universal approximation}
and subsequently used in Lemmas \ref{lem: Appendix W dynamics converges}
and \ref{lem: Appendix J converges} above. 

Since $\xi\rightarrow F\left(w,\xi\right)$ and $\xi\rightarrow G\left(w,x,w_{0},\xi\right)$
are convex by Assumption \ref{assu: Appendix Convergence assumptions - critical},
and the convexity is preserved by taking the expectation of $F$,
we have the result that $\xi\rightarrow J_{n}\left(\boldsymbol{\theta}_{n},\xi;t_{0},w_{0}\right)$
and $\xi\rightarrow J\left(\boldsymbol{p},\xi;t_{0},w_{0}\right)$
are also convex, so that the infimum over $\xi\in\mathbb{R}$ in each
case can be attained and is unique. With $\boldsymbol{p}^{\ast}$
still denoting the optimal control, define $\xi^{\ast}$ as the value
\begin{eqnarray}
\xi^{\ast} & \coloneqq & \inf_{\xi\in\mathbb{R}}J\left(\boldsymbol{p}^{\ast},\xi;t_{0},w_{0}\right).\label{eq: Appendix xi star}
\end{eqnarray}

Since $\mathcal{N}_{n}\subset C\left(\mathcal{D}_{\boldsymbol{\phi}},\mathcal{Z}\right)$,
we have, for all $\xi\in\mathbb{R}$ and all $n\in\mathbb{N}$,
\begin{eqnarray}
\inf_{\boldsymbol{\theta}_{n}\in\mathbb{R}^{\nu_{n}}}J_{n}\left(\boldsymbol{\theta}_{n},\xi;t_{0},w_{0}\right)=\inf_{\boldsymbol{f}_{n}\left(\cdot;\boldsymbol{\theta}_{n}\right)\in\mathcal{N}_{n}}J\left(\boldsymbol{f}_{n},\xi;t_{0},w_{0}\right) & \geq & \inf_{\boldsymbol{p}\in C\left(\mathcal{D}_{\boldsymbol{\phi}},\mathcal{Z}\right)}J\left(\boldsymbol{p},\xi;t_{0},w_{0}\right).\label{eq: Appendix J_n greater than}
\end{eqnarray}
Taking the infimum in (\ref{eq: Appendix J_n greater than}) over
$\xi\in\mathbb{R}$, and exchanging the order of minimization, we
therefore have 
\begin{eqnarray}
\inf_{\left(\boldsymbol{\theta}_{n},\xi\right)\in\mathbb{R}^{\nu_{n}+1}}J_{n}\left(\boldsymbol{\theta}_{n},\xi;t_{0},w_{0}\right) & = & \inf_{\boldsymbol{\theta}_{n}\in\mathbb{R}^{\nu_{n}}}\inf_{\xi\in\mathbb{R}}J_{n}\left(\boldsymbol{\theta}_{n},\xi;t_{0},w_{0}\right)\nonumber \\
 & \geq & \inf_{\boldsymbol{p}\in C\left(\mathcal{D}_{\boldsymbol{\phi}},\mathcal{Z}\right)}\inf_{\xi\in\mathbb{R}}J\left(\boldsymbol{p},\xi;t_{0},w_{0}\right),\qquad\forall n\in\mathbb{N}.\label{eq: Appendix J_n greater than with inf}
\end{eqnarray}
Taking limits in (\ref{eq: Appendix J_n greater than with inf}),
we obtain 
\begin{eqnarray}
\lim_{n\rightarrow\infty}\inf_{\left(\boldsymbol{\theta}_{n},\xi\right)\in\mathbb{R}^{\nu_{n}+1}}J_{n}\left(\boldsymbol{\theta}_{n},\xi;t_{0},w_{0}\right) & \geq & \inf_{\xi\in\mathbb{R}}\inf_{\boldsymbol{p}\in C\left(\mathcal{D}_{\boldsymbol{\phi}},\mathcal{Z}\right)}J\left(\boldsymbol{p},\xi;t_{0},w_{0}\right).\label{eq: Appendix limits J_n greater than}
\end{eqnarray}
Now consider specifically the sequence $\boldsymbol{f}_{n}^{\ast}\left(\cdot,\boldsymbol{\theta}_{n}^{\ast}\right)$
identified in Lemma \ref{lem: Appendix Universal approximation} and
the value $\xi^{\ast}$ in (\ref{eq: Appendix xi star}). Since $\boldsymbol{f}_{n}^{\ast}\left(\cdot,\boldsymbol{\theta}_{n}^{\ast}\right)\in\mathcal{N}_{n}$
(so that $\boldsymbol{\theta}_{n}^{\ast}\in\mathbb{R}^{\nu_{n}}$)
and $\xi^{\ast}\in\mathbb{R}$, we have
\begin{eqnarray}
\inf_{\left(\boldsymbol{\theta}_{n},\xi\right)\in\mathbb{R}^{\nu_{n}+1}}J_{n}\left(\boldsymbol{\theta}_{n},\xi;t_{0},w_{0}\right) & \leq & J_{n}\left(\boldsymbol{\theta}_{n}^{\ast},\xi^{\ast};t_{0},w_{0}\right),\qquad\forall n\in\mathbb{N}.\label{eq: Appendix J_n lower than}
\end{eqnarray}
By Lemma \ref{lem: Appendix J converges}, we have
\begin{eqnarray}
\lim_{n\rightarrow\infty}J_{n}\left(\boldsymbol{\theta}_{n}^{\ast},\xi^{\ast};t_{0},w_{0}\right) & = & J\left(\boldsymbol{p}^{\ast},\xi^{\ast};t_{0},w_{0}\right),\quad\textrm{ where }\xi^{\ast}\textrm{ is given by \eqref{eq: Appendix xi star}}.\label{eq: Appendix J_n convergence with xi_star}
\end{eqnarray}
Therefore, taking limits in (\ref{eq: Appendix J_n lower than}) and
using (\ref{eq: Appendix J_n convergence with xi_star}), we obtain
the inequality
\begin{eqnarray}
\lim_{n\rightarrow\infty}\inf_{\left(\boldsymbol{\theta}_{n},\xi\right)\in\mathbb{R}^{\nu_{n}+1}}J_{n}\left(\boldsymbol{\theta}_{n},\xi;t_{0},w_{0}\right) & \leq & \lim_{n\rightarrow\infty}J_{n}\left(\boldsymbol{\theta}_{n}^{\ast},\xi^{\ast};t_{0},w_{0}\right)\nonumber \\
 & = & J\left(\boldsymbol{p}^{\ast},\xi^{\ast};t_{0},w_{0}\right)\nonumber \\
 & = & \inf_{\xi\in\mathbb{R}}\inf_{\boldsymbol{p}\in C\left(\mathcal{D}_{\boldsymbol{\phi}},\mathcal{Z}\right)}J\left(\boldsymbol{p},\xi;t_{0},w_{0}\right).\label{eq: Appendix limits J_n lower than}
\end{eqnarray}
Combining (\ref{eq: Appendix limits J_n greater than}) and (\ref{eq: Appendix limits J_n lower than}),
we therefore have equality in both (\ref{eq: Appendix limits J_n greater than})
and (\ref{eq: Appendix limits J_n lower than}), and obtain
\begin{eqnarray}
\lim_{n\rightarrow\infty}\inf_{\left(\boldsymbol{\theta}_{n},\xi\right)\in\mathbb{R}^{\nu_{n}+1}}J_{n}\left(\boldsymbol{\theta}_{n},\xi;t_{0},w_{0}\right) & = & \inf_{\xi\in\mathbb{R}}\inf_{\boldsymbol{p}\in C\left(\mathcal{D}_{\boldsymbol{\phi}},\mathcal{Z}\right)}J\left(\boldsymbol{p},\xi;t_{0},w_{0}\right),\label{eq: Appendix limits J_n EQUALITY}
\end{eqnarray}
which concludes the proof of Theorem \ref{thm: Validity of NN approximation}.
Finally, the notation $\boldsymbol{f}_{n}^{\ast}\left(\cdot,\boldsymbol{\theta}_{n}^{\ast}\right)$
in Lemma \ref{lem: Appendix Universal approximation} is motivated
by the fact that equality holds in (\ref{eq: Appendix limits J_n lower than}).\qed

\subsection{Proof of Theorem \ref{thm: Validity of NN computational estimate}\label{subsec:Appendix Proof-of-Theorem validity computational}}

We start with the following auxiliary result, which is essentially
a version of the law of large numbers applicable to the current setting.
\begin{lem}
\label{lem: LLN}(Applicable version of the law of large numbers)
Suppose that Assumption \ref{assu: Appendix Convergence assumptions - critical},
Assumption \ref{assu: Convergence assumptions - convenience} and
Assumption \ref{assu: Appendix Convergence assumptions - computational}
hold. Then 
\begin{eqnarray}
\sup_{\boldsymbol{\theta}_{n}\in\mathbb{R}^{\nu_{n}}}\left|\frac{1}{n}\sum_{j=1}^{n}W^{\left(j\right)}\left(T;\boldsymbol{\theta}_{n},\mathcal{Y}_{n}\right)-E^{t_{0},w_{0}}\left[W\left(T;\boldsymbol{\theta}_{n},\boldsymbol{Y}\right)\right]\right| & \overset{P}{\longrightarrow} & 0,\;\textrm{as }n\rightarrow\infty,\label{eq: Appendix LLN wealth}
\end{eqnarray}
and
\begin{eqnarray}
\sup_{\left(\boldsymbol{\theta}_{n},\xi\right)\in\mathbb{R}^{\nu_{n}+1}}\left|\frac{1}{n}\sum_{j=1}^{n}F\left(W^{\left(j\right)}\left(T;\boldsymbol{\theta}_{n},\mathcal{Y}_{n}\right),\xi\right)-E^{t_{0},w_{0}}\left[F\left(W\left(T;\boldsymbol{\theta}_{n},\boldsymbol{Y}\right),\xi\right)\right]\right| & \overset{P}{\longrightarrow} & 0,\;\textrm{as }n\rightarrow\infty.\label{eq: Appendix LLN for F}
\end{eqnarray}
\end{lem}

\begin{proof}
Since for any fixed number of hidden layers, our NN formulation also
requires $\mathcal{O}\left(\hbar_{n}\right)$ evaluations of the exponential
function, exactly the same steps as in \cite{TsangWong2020} (specifically,
see Corollary 7.4 and Theorem 4.3 in \cite{TsangWong2020}) can be
used to establish (\ref{eq: Appendix LLN wealth}) and (\ref{eq: Appendix LLN for F}).
\end{proof}
The following lemma establishes a required auxiliary result involving
the function $G$.
\begin{lem}
\label{lem: Convergence of G}(Convergence of $G$ in probability)
Suppose that Assumption \ref{assu: Appendix Convergence assumptions - critical},
Assumption \ref{assu: Convergence assumptions - convenience} and
Assumption \ref{assu: Appendix Convergence assumptions - computational}
hold. Then 
\begin{eqnarray}
\sup_{\left(\boldsymbol{\theta}_{n},\xi\right)\in\mathbb{R}^{\nu_{n}+1}}\left|\frac{1}{n}\sum_{j=1}^{n}G\left(W^{\left(j\right)}\left(T;\boldsymbol{\theta}_{n},\mathcal{Y}_{n}\right),\frac{1}{n}\sum_{k=1}^{n}W^{\left(k\right)}\left(T;\boldsymbol{\theta}_{n},\mathcal{Y}_{n}\right),w_{0},\xi\right)\right.\nonumber \\
\left.-E^{t_{0},w_{0}}\left[G\left(W\left(T;\boldsymbol{\theta}_{n},\boldsymbol{Y}\right),E^{t_{0},w_{0}}\left[W\left(T;\boldsymbol{\theta}_{n},\boldsymbol{Y}\right)\right],w_{0},\xi\right)\right]\begin{array}{c}
\!\!\!\!\!\!\!\!\\
\!\!\!\!\!\!\!\!
\end{array}\right| & \overset{P}{\longrightarrow} & 0,\label{eq: Appendix lemma G convergence}
\end{eqnarray}
as $n\rightarrow\infty$.
\end{lem}

\begin{proof}
For given values of $\xi\in\mathbb{R}$, $w_{0}>0$ and $w\in\mathbb{R}$,
consider the function $x\rightarrow G\left(x,w,w_{0},\xi\right)$.
By the results of Lemma \ref{lem: LLN}, we have 
\begin{eqnarray}
\sup_{\left(\boldsymbol{\theta}_{n},\xi\right)\in\mathbb{R}^{\nu_{n}+1}}\left|\frac{1}{n}\sum_{j=1}^{n}G\left(W^{\left(j\right)}\left(T;\boldsymbol{\theta}_{n},\mathcal{Y}_{n}\right),w,w_{0},\xi\right)-E^{t_{0},w_{0}}\left[G\left(W\left(T;\boldsymbol{\theta}_{n},\boldsymbol{Y}\right),w,w_{0},\xi\right)\right]\right| & \overset{P}{\longrightarrow} & 0,\label{eq: Appendix LLN for G - fixed second arg}
\end{eqnarray}
as $n\rightarrow\infty$. Keeping $x$ fixed, consider the function
$w\rightarrow G\left(x,w,w_{0},\xi\right):\left[0,w_{max}\right]\rightarrow\mathbb{R}$.
Since $G$ is continuous, there exists a sequence of functions $\left(G_{m}\right)_{m\in\mathbb{N}}$,
where for each $m\in\mathbb{N}$, the function $w\rightarrow G_{m}\left(x,w,w_{0},\xi\right):\left[0,w_{max}\right]\rightarrow\mathbb{R}$
is $L_{m}$-Lipschitz, such that $\left(G_{m}\right)$ converges uniformly
to $G$ on $\left[0,w_{max}\right]$ - see for example \cite{Miculescu2000}.
Therefore, for an arbitrary value of $\epsilon>0$, there exists a
sufficiently large value $\tilde{m}\in\mathbb{N}$ such that 
\begin{align}
\left|G_{\tilde{m}}\left(x,w,w_{0},\xi\right)-G\left(x,w,w_{0},\xi\right)\right|< & \frac{\epsilon}{2},\qquad\forall w\in\left[0,w_{max}\right].\label{eq: Appendix G_m vs G}
\end{align}
Observing that $\frac{1}{n}\sum_{j=1}^{n}W^{\left(j\right)}\left(T;\boldsymbol{\theta}_{n},\mathcal{Y}_{n}\right)\in\left[0,w_{max}\right]$
and by the monotonicity of expectation we also have $E^{t_{0},w_{0}}\left[W\left(T;\boldsymbol{\theta}_{n},\boldsymbol{Y}\right)\right]\in\left[0,w_{max}\right]$,
we use (\ref{eq: Appendix G_m vs G}) to obtain 
\begin{eqnarray}
 &  & \left|G_{\tilde{m}}\left(x,\frac{1}{n}\sum_{j=1}^{n}W^{\left(j\right)}\left(T;\boldsymbol{\theta}_{n},\mathcal{Y}_{n}\right),w_{0},\xi\right)-G\left(x,\frac{1}{n}\sum_{j=1}^{n}W^{\left(j\right)}\left(T;\boldsymbol{\theta}_{n},\mathcal{Y}_{n}\right),w_{0},\xi\right)\right|\nonumber \\
 &  & +\left|G_{\tilde{m}}\left(x,E^{t_{0},w_{0}}\left[W\left(T;\boldsymbol{\theta}_{n},\boldsymbol{Y}\right)\right],w_{0},\xi\right)-G\left(x,E^{t_{0},w_{0}}\left[W\left(T;\boldsymbol{\theta}_{n},\boldsymbol{Y}\right)\right],w_{0},\xi\right)\begin{array}{c}
\!\!\!\!\!\!\!\!\\
\!\!\!\!\!\!\!\!
\end{array}\right|\nonumber \\
 & < & \epsilon,\label{eq: Appendix two terms Gm min G bounded}
\end{eqnarray}
for any given values of $\xi\in\mathbb{R}$ and $w_{0}>0$. In addition,
since $G_{\tilde{m}}$ is $L_{\tilde{m}}$-Lipschitz, we have 
\begin{eqnarray}
 &  & \left|G_{\widetilde{m}}\left(x,\frac{1}{n}\sum_{j=1}^{n}W^{\left(j\right)}\left(T;\boldsymbol{\theta}_{n},\mathcal{Y}_{n}\right),w_{0},\xi\right)-G_{\widetilde{m}}\left(x,E^{t_{0},w_{0}}\left[W\left(T;\boldsymbol{\theta}_{n},\boldsymbol{Y}\right)\right],w_{0},\xi\right)\right|\nonumber \\
 & \leq & L_{\widetilde{m}}\cdot\left|\frac{1}{n}\sum_{j=1}^{n}W^{\left(j\right)}\left(T;\boldsymbol{\theta}_{n},\mathcal{Y}_{n}\right)-E^{t_{0},w_{0}}\left[W\left(T;\boldsymbol{\theta}_{n},\boldsymbol{Y}\right)\right]\right|.\label{eq: Appendix result for G_m Lipschitz}
\end{eqnarray}
Using (\ref{eq: Appendix two terms Gm min G bounded}) and (\ref{eq: Appendix result for G_m Lipschitz})
as well as the triangle inequality, we therefore have
\begin{eqnarray}
 &  & \left|G\left(x,\frac{1}{n}\sum_{j=1}^{n}W^{\left(j\right)}\left(T;\boldsymbol{\theta}_{n},\mathcal{Y}_{n}\right),w_{0},\xi\right)-G\left(x,E^{t_{0},w_{0}}\left[W\left(T;\boldsymbol{\theta}_{n},\boldsymbol{Y}\right)\right]\right),w_{0},\xi\right|\nonumber \\
 & < & \epsilon+L_{\widetilde{m}}\cdot\left|\frac{1}{n}\sum_{j=1}^{n}W^{\left(j\right)}\left(T;\boldsymbol{\theta}_{n},\mathcal{Y}_{n}\right)-E^{t_{0},w_{0}}\left[W\left(T;\boldsymbol{\theta}_{n},\boldsymbol{Y}\right)\right]\right|,\label{eq: Appendix bound for G}
\end{eqnarray}
for any given values of $\xi\in\mathbb{R}$ and $w_{0}>0$. Taking
the supremum over $\left(\boldsymbol{\theta}_{n},\xi\right)\in\mathbb{R}^{\nu_{n}+1}$
in (\ref{eq: Appendix bound for G}), using the result (\ref{eq: Appendix LLN wealth})
from Lemma \ref{lem: LLN} as well as the fact that $\epsilon>0$
was arbitrary, we therefore have 
\begin{eqnarray}
\sup_{\left(\boldsymbol{\theta}_{n},\xi\right)\in\mathbb{R}^{\nu_{n}+1}}\left|G\left(x,\frac{1}{n}\sum_{k=1}^{n}W^{\left(k\right)}\left(T;\boldsymbol{\theta}_{n},\mathcal{Y}_{n}\right),w_{0},\xi\right)-G\left(x,E^{t_{0},w_{0}}\left[W\left(T;\boldsymbol{\theta}_{n},\boldsymbol{Y}\right)\right],w_{0},\xi\right)\right| & \overset{P}{\longrightarrow} & 0.\label{eq: Appendix lemma G convergence- fixed first arg}
\end{eqnarray}

The results (\ref{eq: Appendix LLN for G - fixed second arg}) and
(\ref{eq: Appendix lemma G convergence- fixed first arg}), together
with the triangle inequality, therefore gives 
\begin{eqnarray}
 &  & \sup_{\left(\boldsymbol{\theta}_{n},\xi\right)\in\mathbb{R}^{\nu_{n}+1}}\left|\frac{1}{n}\sum_{j=1}^{n}G\left(W^{\left(j\right)}\left(T;\boldsymbol{\theta}_{n},\mathcal{Y}_{n}\right),\frac{1}{n}\sum_{k=1}^{n}W^{\left(k\right)}\left(T;\boldsymbol{\theta}_{n},\mathcal{Y}_{n}\right),w_{0},\xi\right)\right.\nonumber \\
 &  & \qquad\qquad\left.\qquad\qquad-E^{t_{0},w_{0}}\left[G\left(W\left(T;\boldsymbol{\theta}_{n},\boldsymbol{Y}\right),E^{t_{0},w_{0}}\left[W\left(T;\boldsymbol{\theta}_{n},\boldsymbol{Y}\right)\right],w_{0},\xi\right)\right]\begin{array}{c}
\!\!\!\!\!\!\!\!\\
\!\!\!\!\!\!\!\!
\end{array}\right|\nonumber \\
 & \leq & \frac{1}{n}\sum_{j=1}^{n}\sup_{\left(\boldsymbol{\theta}_{n},\xi\right)\in\mathbb{R}^{\nu_{n}+1}}\left|\frac{1}{n}\sum_{j=1}^{n}G\left(W^{\left(j\right)}\left(T;\boldsymbol{\theta}_{n},\mathcal{Y}_{n}\right),\frac{1}{n}\sum_{k=1}^{n}W^{\left(k\right)}\left(T;\boldsymbol{\theta}_{n},\mathcal{Y}_{n}\right),w_{0},\xi\right)\right.\nonumber \\
 &  & \qquad\qquad\left.\qquad\qquad-G\left(W^{\left(j\right)}\left(T;\boldsymbol{\theta}_{n},\mathcal{Y}_{n}\right),E^{t_{0},w_{0}}\left[W\left(T;\boldsymbol{\theta}_{n},\boldsymbol{Y}\right)\right],w_{0},\xi\right)\begin{array}{c}
\!\!\!\!\!\!\!\!\\
\!\!\!\!\!\!\!\!
\end{array}\right|\nonumber \\
 &  & +\sup_{\left(\boldsymbol{\theta}_{n},\xi\right)\in\mathbb{R}^{\nu_{n}+1}}\left|\frac{1}{n}\sum_{j=1}^{n}G\left(W^{\left(j\right)}\left(T;\boldsymbol{\theta}_{n},\mathcal{Y}_{n}\right),E^{t_{0},w_{0}}\left[W\left(T;\boldsymbol{\theta}_{n},\boldsymbol{Y}\right)\right],w_{0},\xi\right)\right.\nonumber \\
 &  & \qquad\qquad\left.\qquad\qquad-E^{t_{0},w_{0}}\left[G\left(W\left(T;\boldsymbol{\theta}_{n},\boldsymbol{Y}\right),E^{t_{0},w_{0}}\left[W\left(T;\boldsymbol{\theta}_{n},\boldsymbol{Y}\right)\right],w_{0},\xi\right)\right]\begin{array}{c}
\!\!\!\!\!\!\!\!\\
\!\!\!\!\!\!\!\!
\end{array}\right|\nonumber \\
 & \overset{P}{\longrightarrow} & 0\qquad\textrm{as }n\rightarrow\infty.\label{eq: Final step convergence of G in probability}
\end{eqnarray}
\end{proof}

\subsubsection*{Proof of Theorem \ref{thm: Validity of NN computational estimate}}

The expression in (\ref{eq: Result thm validity of NN computational estimate}),
together with the triangle inequality, imply that

\begin{eqnarray}
 &  & \left|\inf_{\left(\boldsymbol{\theta}_{n},\xi\right)\in\mathbb{R}^{\nu_{n}+1}}\hat{J}_{n}\left(\boldsymbol{\theta}_{n},\xi;t_{0},w_{0},\mathcal{Y}_{n}\right)-\inf_{\xi\in\mathbb{R}}\inf_{\boldsymbol{p}\in C\left(\mathcal{D}_{\boldsymbol{\phi}},\mathcal{Z}\right)}J\left(\boldsymbol{p},\xi;t_{0},w_{0}\right)\right|\nonumber \\
 & \leq & \left|\inf_{\left(\boldsymbol{\theta}_{n},\xi\right)\in\mathbb{R}^{\nu_{n}+1}}\hat{J}_{n}\left(\boldsymbol{\theta}_{n},\xi;t_{0},w_{0},\mathcal{Y}_{n}\right)-\inf_{\left(\boldsymbol{\theta}_{n},\xi\right)\in\mathbb{R}^{\eta_{n}+1}}J_{n}\left(\boldsymbol{\theta}_{n},\xi;t_{0},w_{0}\right)\right|\label{eq: Appendix stochastic error}\\
 &  & +\left|\inf_{\left(\boldsymbol{\theta}_{n},\xi\right)\in\mathbb{R}^{\nu_{n}+1}}J_{n}\left(\boldsymbol{\theta}_{n},\xi;t_{0},w_{0}\right)-\inf_{\xi\in\mathbb{R}}\inf_{\boldsymbol{p}\in C\left(\mathcal{D}_{\boldsymbol{\phi}},\mathcal{Z}\right)}J\left(\boldsymbol{p},\xi;t_{0},w_{0}\right)\right|.\label{eq: Appendix approximation error}
\end{eqnarray}

Using the definitions of $\hat{J}_{n}\left(\boldsymbol{\theta}_{n},\xi;t_{0},w_{0},\mathcal{Y}_{n}\right)$
in (\ref{eq: J_n_HAT approx}) and $J_{n}\left(\boldsymbol{\theta}_{n},\xi;t_{0},w_{0}\right)$
in (\ref{eq: J_n}), the expression (\ref{eq: Appendix stochastic error})
gives 
\begin{eqnarray}
 &  & \left|\inf_{\left(\boldsymbol{\theta}_{n},\xi\right)\in\mathbb{R}^{\nu_{n}+1}}\hat{J}_{n}\left(\boldsymbol{\theta}_{n},\xi;t_{0},w_{0},\mathcal{Y}_{n}\right)-\inf_{\left(\boldsymbol{\theta}_{n},\xi\right)\in\mathbb{R}^{\eta_{n}+1}}J_{n}\left(\boldsymbol{\theta}_{n},\xi;t_{0},w_{0}\right)\right|\nonumber \\
 & \leq & \sup_{\left(\boldsymbol{\theta}_{n},\xi\right)\in\mathbb{R}^{\nu_{n}+1}}\left|\hat{J}_{n}\left(\boldsymbol{\theta}_{n},\xi;t_{0},w_{0},\mathcal{Y}_{n}\right)-J_{n}\left(\boldsymbol{\theta}_{n},\xi;t_{0},w_{0}\right)\begin{array}{c}
\!\!\!\!\!\!\!\!\\
\!\!\!\!\!\!\!\!
\end{array}\right|\nonumber \\
 & \leq & \sup_{\left(\boldsymbol{\theta}_{n},\xi\right)\in\mathbb{R}^{\nu_{n}+1}}\left|\frac{1}{n}\sum_{j=1}^{n}F\left(W^{\left(j\right)}\left(T;\boldsymbol{\theta}_{n},\mathcal{Y}_{n}\right),\xi\right)-E^{t_{0},w_{0}}\left[F\left(W\left(T;\boldsymbol{\theta}_{n},\boldsymbol{Y}\right),\xi\right)\right]\right|\label{eq: Appendix stoch error ito F}\\
 &  & +\sup_{\left(\boldsymbol{\theta}_{n},\xi\right)\in\mathbb{R}^{\nu_{n}+1}}\left|\frac{1}{n}\sum_{j=1}^{n}G\left(W^{\left(j\right)}\left(T;\boldsymbol{\theta}_{n},\mathcal{Y}_{n}\right),\frac{1}{n}\sum_{k=1}^{n}W^{\left(k\right)}\left(T;\boldsymbol{\theta}_{n},\mathcal{Y}_{n}\right),w_{0},\xi\right)\right.\nonumber \\
 &  & \qquad\qquad\left.\qquad\qquad-E^{t_{0},w_{0}}\left[G\left(W\left(T;\boldsymbol{\theta}_{n},\boldsymbol{Y}\right),E^{t_{0},w_{0}}\left[W\left(T;\boldsymbol{\theta}_{n},\boldsymbol{Y}\right)\right],w_{0},\xi\right)\right]\begin{array}{c}
\!\!\!\!\!\!\!\!\\
\!\!\!\!\!\!\!\!
\end{array}\right|.\label{eq: Appendix stoch error ito G}
\end{eqnarray}
As per Lemma \ref{lem: LLN} and Lemma \ref{lem: Convergence of G},
(\ref{eq: Appendix stoch error ito F}) and (\ref{eq: Appendix stoch error ito G})
converge to zero in probability as $n\rightarrow\infty$. As a result,
since (\ref{eq: Appendix stochastic error}) therefore converges to
zero in probability as $n\rightarrow\infty$ and, by Theorem \ref{thm: Validity of NN approximation},
(\ref{eq: Appendix approximation error}) converges to zero as $n\rightarrow\infty$,
we conclude that the result (\ref{eq: Result thm validity of NN computational estimate})
of Theorem \ref{thm: Validity of NN computational estimate} holds.
\qed

\section{NN approach: Selected practical considerations\label{sec: Appendix B - NN approach - practical considerations}}

We summarize some practical considerations with respect to the NN
approach:

\begin{enumerate}[label=(\roman*)]

\item Constructing training and testing datasets $\mathcal{Y}_{n}$
and $\mathcal{Y}_{\hat{n}}^{test}$: Since these sets correspond to
finite samples of $\boldsymbol{Y}$ and $\boldsymbol{Y}^{test}$,
any data generation technique generating paths of underlying asset
returns can be used for the construction of training and testing data
sets. As illustrated in Section \ref{sec:Numerical results}, data
generation techniques like (i) Monte Carlo simulation of parametric
asset dynamics, (ii) block bootstrap resampling of empirical returns,
or for example (iii) GAN-generated synthetic returns can all be employed
without difficulty, but we emphasize that the approach remains agnostic
regarding the underlying data generation methodology. Note that the
underlying data generation assumptions typically differ for $\mathcal{Y}_{n}$
and $\mathcal{Y}_{\hat{n}}^{test}$, respectively, depending on for
example the time periods of empirical data considered for in-sample
and out-of-sample testing. 

As for the number of paths $n$ in each of $\mathcal{Y}_{n}$ and
$\mathcal{Y}_{\hat{n}}^{test}$, experiments show that in the case
of measures of tail risk in the objectives such as CVaR (see (\ref{eq: MCV objective})),
a significantly larger number of paths are required in order to obtain
a sufficiently large sample of tail outcomes in the training and testing
data, than for example in cases where variance is the risk measure.
To give a concrete examples, at least 2 million paths in the training
set of the NN in Subsection \ref{subsec: Ground truth problem MCV}
is required to produce reliable results for the CVaR, whereas 1 million
paths in the training set of the NN in Subsection \ref{subsec:Ground-truth: DSQ with cont rebal}
are more than sufficient to obtain reliable results.

\item Depth (number of hidden layers $\mathcal{L}^{h}$) and width
(number of nodes in each hidden layer $\hbar\left(n\right)$) of the
NN: As the examples in Section \ref{sec:Numerical results} show,
remarkably accurate can be obtained with NNs no deeper than 2 hidden
layers and a relatively small number of nodes in each hidden layer.
For objectives involving more complex investment strategies such as
MCV and Mean - Semi-variance (where, even in the case of two assets,
the behavior of the optimal strategy is clearly more complex than
in the case of for example the MV-optimal strategy), experiments show
that two hidden layers lead to stable and reliable results, with the
number of hidden nodes in each hidden layer chosen to be slightly
more than the number of assets, for example $\hbar\left(n\right)=N_{a}+2$.
For objectives such as DSQ and MV, a single hidden layer is often
sufficient. 

\item Activation functions: As highlighted in Assumption \ref{assu: Appendix NN structure assumptions},
we use logistic sigmoid activations as a concrete example for convergence
analysis purposes, but that these theoretical results can be modified
for any of the commonly-used activations (see for example \cite{SonodaMurata2017}).
Note that since NNs of one or two hidden layers were found to be very
effective in solving the problems under consideration, we did not
encounter any problems related to vanishing or exploding gradients
in the case of logistic sigmoid activations. However, if deeper NNs
are required, activation functions could be changed to e.g. ReLU or
ELU without affecting the theoretical foundations for the proposed
approach.

\item For the solution of (\ref{eq: V_n_HAT approx}) by gradient
descent, we used the Gadam algorithm of \cite{GranzioEtAl2020}. This
is simply a combination of the Adam algorithm (\cite{KingmaBa2015})
with tail iterate averaging for improved convergence properties and
variance reduction (\cite{PolyakJuditsky1992,MuckeEtAl2019,NeuRosasco2018}).
For the Adam algorithm component, the default algorithm parameters
of \cite{KingmaBa2015} performed well in our setting, typically with
no more than 50,000 SGD steps. Note that the mini-batch size selected
depends on the problem to be solved: we found that mini-batch sizes
of at least 1,000 paths of the training data set $\mathcal{Y}_{n}$
are required for measures of tail risk in the objective (such as CVaR),
since smaller batch sizes typically means that the tail of the returns
distribution is not sufficiently well represented in choosing the
descent direction, leading to unreliable results in ground truth analyses. 

While the technical results of Section \ref{sec:Numerical results}
formally do not require continuous differentiability (in addition
to continuity) of the functions $F$ and $G$, improved convergence
properties of the SGD algorithm can be obtained if the objective is
at least continuously differentiable (see for example \cite{ShapiroWardi1996}).
For implementation purposes, we can therefore smooth objectives like
(\ref{eq: MCV objective}) in a straightforward way, by for example
replacing $\max\left(x,0\right)$ in (\ref{eq: MCV objective}) with
a continuously differentiable approximation used in \cite{AlexanderColemanLi2006},
\begin{eqnarray}
\max\left(x,0\right) & \simeq & \begin{cases}
x, & \textrm{if }x>\lambda_{mcv},\\
\frac{1}{4\lambda_{mcv}}x^{2}+\frac{1}{2}x+\frac{1}{4}\lambda_{mcv}, & \textrm{if }-\lambda_{mcv}\leq x\leq\lambda_{mcv}\\
0, & \textrm{otherwise},
\end{cases},\label{eq: psi CVAR min}
\end{eqnarray}
where $\lambda_{mcv}$ is some small smoothing parameter (e.g. $\lambda_{mcv}=10^{-3})$. 

In addition to considering the smoothing of certain objectives, minor modifications to objective functions to avoid (mathematical) ill-posedness may be desirable in certain situations. For  example, in the case of the OSQ objective (\ref{eq: OSQ objective}), the term $\epsilon W(\cdot)$ is added to ensure the problem remains well-posed even if $W(t) \gg \gamma$. In this case, when implementing the numerical solution, small values of $\epsilon$ (for example $\epsilon = 10^{-6}$ was chosen in the numerical results of Section \ref{sec:Numerical results}) do not have a noticeable effect on either the summary statistics or the optimal controls.

\end{enumerate}

\section{Additional parameters for numerical results\label{sec: Appendix Parameters-for-numerical results}}

In this appendix, additional parameters related to the numerical results
of Section \ref{sec:Numerical results} are discussed.

The historical returns data for the basic assets such as the T-bills/bonds
and the broad market index were obtained from the CRSP \footnote{Calculations were based on data from the Historical Indexes 2020©,
Center for Research in Security Prices (CRSP), The University of Chicago
Booth School of Business. Wharton Research Data Services was used
in preparing this article. This service and the data available thereon
constitute valuable intellectual property and trade secrets of WRDS
and/or its third party suppliers.}, whereas factor data for Size and Value (see \cite{FamaFrench1992,FamaFrench2015a})
were obtained from Kenneth French's data library\footnote{See https://mba.tuck.dartmouth.edu/pages/faculty/ken.french/data\_library.html}
(KFDL). The detailed time series sourced for each asset is as follows:
\begin{enumerate}
\item T30 (30-day Treasury bill): CRSP, monthly returns for 30-day Treasury
bill.
\item B10 (10-year Treasury bond): CRSP, monthly returns for 10-year Treasury
bond.\footnote{ {\color{black}The 10-year Treasury index was 
constructed from monthly returns from CRSP back to 1941. The data for
1926-1941 were interpolated from annual returns in \cite{Homer_2015}} }

\item Market (broad equity market index): CRSP, monthly returns, including
dividends and distributions, for a capitalization-weighted index consisting
of all domestic stocks trading on major US exchanges (the VWD index).
\item Size (Portfolio of small stocks): KFDL, ``Portfolios Formed on Size'',
which consists of monthly returns on a capitalization-weighted index
consisting of the firms (listed on major US exchanges) with market
value of equity, or market capitalization, at or below the 30th percentile
(i.e. smallest 30\%) of market capitalization values of NYSE-listed
firms.
\item Value (Portfolio of value stocks): KFDL, ``Portfolios Formed on Book-to-Market'',
which consists of monthly returns on a capitalization-weighted index
of the firms (listed on major US exchanges) consisting of the firms
(listed on major US exchanges) with book-to-market value of equity
ratios at or above the 70th percentile (i.e. highest 30\%) of book-to-market
ratios of NYSE-listed firms.
\end{enumerate}
The historical asset returns time series are inflation-adjusted using
inflation data from the US Bureau of Labor Statistics\footnote{The annual average CPI-U index, which is based on inflation data for
urban consumers, were used - see \texttt{http://www.bls.gov.cpi}}.

For the purposes of obtaining the parameters for (\ref{eq: Parametric dynamics of underlying assets})
in Subsections (\ref{subsec:Ground-truth: DSQ with cont rebal}) and
(\ref{subsec: Ground truth problem MCV}), we use the same calibration
methodology as outlined in \cite{DangForsyth2016,ForsythVetzal2016},
and assume the jump dynamics of the \cite{KouOriginal} model. 

In particular, we assume that in the dynamics (\ref{eq: Parametric dynamics of underlying assets}),
$\log\vartheta_{i}$ has a asymmetric double-exponential distribution,
\begin{eqnarray}
f_{\vartheta_{i}}\left(\vartheta_{i}\right) & = & \nu_{i}\zeta_{i,1}\vartheta_{i}^{-\zeta_{i,1}-1}\mathbb{I}_{\left[\vartheta_{i}\geq1\right]}\left(\vartheta_{i}\right)+\left(1-\nu_{i}\right)\zeta_{i,2}\vartheta_{i}^{\zeta_{i,2}-1}\mathbb{I}_{\left[0\leq\vartheta_{i}<1\right]}\left(\vartheta_{i}\right),\label{eq: Jump pdf - KOU}
\end{eqnarray}
where $\upsilon_{i}\in\left[0,1\right]\textrm{ and }\zeta_{i,1}>1,\zeta_{i,2}>0$.
In (\ref{eq: Jump pdf - KOU}), $\nu_{i}$ denotes the probability
of an upward jump given that a jump occurs. The resulting parameters
are obtained using the filtering technique for the calibration of
jump diffusion processes - see \cite{DangForsyth2016,ForsythVetzal2016}
for the relevant methodological details. For calibration purposes,
a jump threshold equal to 3 has been used in the methodology of \cite{DangForsyth2016}.

Table \ref{eq: Jump pdf - KOU} and Table \ref{tab: Params for ground truth MCV}
summarize the parameters for the asset dynamics for Subsections (\ref{subsec:Ground-truth: DSQ with cont rebal})
and (\ref{subsec: Ground truth problem MCV}), respectively.

\noindent 
\begin{table}[!hbt]
\caption{\label{tab: Params for ground truth DSQ with cont rebal} Calibrated,
inflation-adjusted parameters for asset dynamics in Subsection \ref{subsec:Ground-truth: DSQ with cont rebal}:
Ground truth - $DSQ\left(\gamma\right)$ with continuous rebalancing.
In this example, the first asset is assumed to be a risk-free asset,
so we set $\mu_{1}=r$, while the second asset follows jump dynamics.
The parametric asset returns are (trivially) uncorrelated, and parameters
are based on the inflation-adjusted returns of the T30 and VWD time
series, respectively, over the period 1926:01 to 2019:12}

\noindent \centering{}{\footnotesize{}}%
\begin{tabular}{|c||c|c|c|c|c|c|}
\hline 
{\footnotesize{}Parameter} & {\footnotesize{}$\mu$$_{i}$} & {\footnotesize{}$\sigma$$_{i}$} & {\footnotesize{}$\lambda_{i}$} & {\footnotesize{}$\upsilon_{i}$} & {\footnotesize{}$\zeta_{i,1}$} & {\footnotesize{}$\zeta_{i,2}$}\tabularnewline
\hline 
{\footnotesize{}Asset 1 (T30)} & {\footnotesize{}0.0043} & {\footnotesize{}-} & {\footnotesize{}-} & {\footnotesize{}-} & {\footnotesize{}-} & {\footnotesize{}-}\tabularnewline
\hline 
{\footnotesize{}Asset 2(VWD)} & {\footnotesize{}0.0877} & {\footnotesize{}0.1459} & {\footnotesize{}0.3191} & {\footnotesize{}0.2333} & {\footnotesize{}4.3608} & {\footnotesize{}5.504}\tabularnewline
\hline 
\end{tabular}{\footnotesize\par}
\end{table}

\noindent 
\begin{table}[!hbt]
\caption{\label{tab: Params for ground truth MCV} Calibrated, inflation-adjusted
parameters for asset dynamics in Subsection \ref{subsec: Ground truth problem MCV}:
Ground truth - problem $MCV\left(\rho\right)$. In this example, there
are two assets with jump dynamics (see \cite{ForsythVetzal2022CutLosses}),
with parameters based on the inflation-adjusted returns of the T30
and VWD time series over the period 1926:01 to 2019:12. The Brownian
motions in (\ref{eq: Parametric dynamics of underlying assets}) have
correlation $dZ_{1}dZ_{2}=\rho_{1,2}dt$.}

\noindent \centering{}{\footnotesize{}}%
\begin{tabular}{|c||c|c|c|c|c|c|c|}
\hline 
{\footnotesize{}Parameter} & {\footnotesize{}$\mu$$_{i}$} & {\footnotesize{}$\sigma$$_{i}$} & {\footnotesize{}$\lambda_{i}$} & {\footnotesize{}$\upsilon_{i}$} & {\footnotesize{}$\zeta_{i,1}$} & {\footnotesize{}$\zeta_{i,2}$} & {\footnotesize{}$\rho_{1,2}$}\tabularnewline
\hline 
{\footnotesize{}Asset 1 (T30)} & {\footnotesize{}0.0045} & {\footnotesize{}0.0130} & {\footnotesize{}0.5106} & {\footnotesize{}0.3958} & {\footnotesize{}65.85} & {\footnotesize{}57.75} & {\footnotesize{}0.08228}\tabularnewline
\hline 
{\footnotesize{}Asset 2(VWD)} & {\footnotesize{}0.0877} & {\footnotesize{}0.1459} & {\footnotesize{}0.3191} & {\footnotesize{}0.2333} & {\footnotesize{}4.3608} & {\footnotesize{}5.504} & {\footnotesize{}0.08228}\tabularnewline
\hline 
\end{tabular}{\footnotesize\par}
\end{table}

\end{appendices}
\end{document}